\documentclass[journal]{IEEEtran}
\IEEEoverridecommandlockouts

\def\BibTeX{{\rm B\kern-.05em{\sc i\kern-.025em b}\kern-.08em
    T\kern-.1667em\lower.7ex\hbox{E}\kern-.125emX}}

\usepackage{amsmath,amsfonts,amssymb,amsthm}
\usepackage{textcomp}
\usepackage{stfloats}
\usepackage{url}
\usepackage{verbatim}
\usepackage{graphicx}    % Always load graphicx first.
\usepackage{epstopdf}    % This triggers EPS-->PDF conversion.
\usepackage[table]{xcolor}
\usepackage{pgfplots}

\usepackage{balance}
\usepackage{mathrsfs}
\usepackage{cite}
\usepackage[bookmarks=true,colorlinks=true,allcolors=blue,allbordercolors=blue]{hyperref}
\usepackage{multirow}
\usepackage{bm}
\usepackage{comment}

\usepackage{array}
\usepackage{booktabs}
\usepackage{pifont}  % for \ding{55}

\usepackage[super]{nth} % when you write \nth{2} you get a nice superscript 2nd, use also for 1st, 3rd, etc.

\newcommand{\cmark}{\ding{51}}  % check mark
\newcommand{\xmark}{\ding{55}}  % cross mark

\usepackage{colortbl}
\definecolor{lightblue}{RGB}{220,230,245}

\usepackage{color}
\newcommand{\blue}[1]{{\color{black}{#1}}} % blue color
\newcommand{\black}[1]{{\color{black}{#1}}}
\newcommand{\cyan}[1]{{\color{black}{#1}}}

\definecolor{amaranth}{rgb}{0.9, 0.17, 0.31}
\definecolor{AlirezaPurple}{RGB}{150, 0, 250}

\usepackage[font=footnotesize]{subfig}
\usepackage[font=footnotesize]{caption}
\usepackage{pdfpages}

\usepackage{tikz}
\usetikzlibrary{tikzmark,decorations.pathreplacing, patterns}
\usetikzlibrary{arrows, arrows.meta}

\usepackage{rotating}

\usepackage{algorithm}
\usepackage{algorithmic}

\usepackage{fancyhdr}

\usepackage{acronym}

\usepackage{authblk}

\usepackage{soul}

\usepackage{amsthm}
\usepackage{babel}

\theoremstyle{plain}

\newtheorem{prop}{\protect\propname}

\theoremstyle{remark}

\providecommand{\theoremname}{Theorem}
\providecommand{\propname}{Proposition}
\providecommand{\remarkname}{Remark}

\providecommand{\corollaryname}{Corollary}

\usepackage[super]{nth}
\usepackage{acronym}

\acrodef{2d}[2D]{bi-dimensional}
\acrodef{5g}[5G]{fifth-generation}
\acrodef{6g}[6G]{sixth-generation}
\acrodef{aoa}[AoA]{angle of arrival}
\acrodef{ao}[AO]{alternating optimization}
\acrodef{aod}[AoD]{angle of departure}
\acrodef{bs}[BS]{base station}
\acrodef{cav}[CAV]{connected and automated vehicle}
\acrodef{csi}[CSI]{channel state information}
\acrodef{cp}[CP]{canonical polyadic}
\acrodef{cfar}[CFAR]{constant false alarm rate}
\acrodef{ca-cfar}[CA-CFAR]{cell averaging \ac{cfar}}
\acrodef{ccrb}[CCRB]{constrained Cramér-Rao bound}
\acrodef{crb}[CRB]{Cramér-Rao bound}
\acrodef{dais}[DAIS]{delay-angle information spoofing}
\acrodef{dft}[DFT]{discrete Fourier transform}
\acrodef{esprit}[ESPRIT]{estimation of signal parameters via rotational invariance techniques}
\acrodef{fft}[FFT]{fast Fourier transform}
\acrodef{fim}[FIM]{Fisher information matrix}
\acrodef{fr}[FR]{frequency range}

\acrodef{gnss}[GNSS]{global navigation satellite system}
\acrodef{gospa}[GOSPA]{generalized optimal sub-pattern assignment}
\acrodef{hd}[HD]{high-definition}
\acrodef{its}[ITS]{intelligent transportation systems}
\acrodef{los}[LoS]{line-of-sight}
\acrodef{ls}[LS]{least-squares}

\acrodef{mcrb}[MCRB]{mismatched \ac{crb}}
\acrodef{mmWave}[mmWave]{Millimeter-wave}
\acrodef{mimo}[MIMO]{multiple-input multiple-output}
\acrodef{miso}[MISO]{multiple-input single-output}
\acrodef{mpc}[MPC]{multipath component}
\acrodef{ml}[ML]{maximum likelihood}
\acrodef{mf}[MF]{matched-filter}
\acrodef{nlos}[NLoS]{non-line-of-sight}
\acrodef{nmse}[NMSE]{normalized mean square error}
\acrodef{ofdm}[OFDM]{orthogonal frequency division multiplexing}
\acrodef{omp}[OMP]{orthogonal matching pursuit}
\acrodef{peb}[PEB]{position error bound}
\acrodef{pla}[PLA]{physical-layer authentication}
\acrodef{poam}[POAM]{penalized optimal assignment metric}
\acrodef{prmse}[PRMSE]{penalized root mean squared error}
\acrodef{prs}[PRS]{positioning reference signal}
\acrodef{roi}[ROI]{region of interest}
\acrodef{rf}[RF]{radio frequency}
\acrodef{rfc}[RFC]{radio frequency chain}
\acrodef{rcs}[RCS]{radar cross section}
\acrodef{ris}[RIS]{reflective intelligent surface} 
\acrodef{rmse}[RMSE]{root mean square error}
\acrodef{simo}[SIMO]{single-input multiple-output}
\acrodef{siso}[SISO]{single-input single-output}
\acrodef{snr}[SNR]{signal-to-noise ratio}
\acrodef{sota}[SoTA]{state-of-the-art}
\acrodef{sp}[SP]{scatter point}
\acrodef{srs}[SRS]{sounding reference signal}
\acrodef{svd}[SVD]{singular value decomposition}
\acrodef{tdoa}[TDoA]{time difference of arrival}
\acrodef{toa}[ToA]{time of arrival}
\acrodef{qcqp}[QCQP]{quadratically constrained quadratic program}

\acrodef{ue}[UE]{user equipment} 
\acrodef{ula}[ULA]{uniform linear array}
\acrodef{ura}[URA]{uniform rectangular array}
\acrodef{va}[VA]{virtual anchor}

\begin{document}
\bstctlcite{BSTcontrol}

%\title{User Position Privacy Preservation in Single Base-station MIMO-OFDM Communication}
%\title{HoloTrace: a Location Privacy Preservation Solution for Single Base-Station MIMO-OFDM Positioning}
%\title{HoloTrace: a Location Privacy Preservation Scheme for Single Base Station MIMO-OFDM Systems
% \\ \mn{I propose to modify as the reference to single BS seems to be a limiting factor. I would mention it later in the abstract.}
\title{HoloTrace: a Location Privacy-\blue{Preserving Framework} for mmWave MIMO-OFDM Systems}

\author{
Lorenzo~Italiano,~\IEEEmembership{Graduate Student Member,~IEEE},
Alireza~Pourafzal,~\IEEEmembership{Member,~IEEE},
Hui~Chen,~\IEEEmembership{Member,~IEEE},
Mattia~Brambilla,~\IEEEmembership{Senior Member,~IEEE},
Gonzalo~Seco-Granados,~\IEEEmembership{Fellow,~IEEE},
Monica~Nicoli,~\IEEEmembership{Senior Member,~IEEE}, 
and~Henk~Wymeersch,~\IEEEmembership{Fellow,~IEEE}
\thanks{This work was supported, in part, by the European Union—NextGenerationEU under the National Sustainable Mobility Center under Grant CN00000023, Italian Ministry of University and Research (MUR) Decree n. 352–09/04/2022, the Spanish Agency of Research (AEI) under grant PID2023-152820OB-I00 funded by MICIU/AEI/10.13039/501100011033, the AGAUR-ICREA Academia Program, VR project 2023-03821, and Chalmers Transport Area of Advance. \emph{(Corresponding author: Alireza Pourafzal.)}}
\thanks{
L. Italiano and M. Brambilla are with the Dipartimento di Elettronica, Informazione e Bioingegneria, Politecnico di Milano, 20133 Milan, Italy (email: \{lorenzo.italiano, mattia.brambilla\}@polimi.it).
}
\thanks{
A. Pourafzal, H. Chen, and H. Wymeersch are with the Department of Electrical Engineering, Chalmers University of Technology, 41296 G\"oteborg, Sweden (email: \{alireza.pourafzal, hui.chen, henkw\}@chalmers.se).
}
\thanks{
G. Seco-Granados is with the Department of Telecommunication Engineering, Autonomous University of Barcelona, 08193 Barcelona, Spain (email: gonzalo.seco@uab.cat).
}
\thanks{
M. Nicoli is with the Dipartimento di Ingegneria Gestionale, Politecnico di Milano, 20133 Milan, Italy (email: monica.nicoli@polimi.it).
}
\vspace{-20pt}
}

\maketitle
\begin{abstract}
\cyan{High-resolution radio localization is becoming increasingly important for next-generation vehicular and mobile networks, enabling infrastructure-assisted positioning for connected and automated mobility. However, the same mmWave \ac{mimo}-\ac{ofdm} links used for communication may allow a serving \ac{bs} to infer the location of a vehicle-mounted or mobile \ac{ue} from routine uplink pilot transmissions, even without an explicit positioning request.}
This paper introduces HoloTrace, a signal-level location-privacy framework for single-\ac{bs} localization \cyan{in mmWave \ac{mimo}-\ac{ofdm} systems}.
HoloTrace modifies user-side pilot symbols while keeping the analog precoder fixed to bias localization-relevant channel features, including \ac{aoa}, \ac{aod}, and \ac{tdoa}, toward a prescribed fictitious geometry.
We formulate pilot-level spoofing as a subspace projection-based signal design problem and derive closed-form oracle designs assuming \ac{ue}-side channel-gain knowledge.
We further develop blind designs that operate without complex path-gain information, while relying on geometric side information, and explicitly characterize their limitations in multipath scenarios.
Simulation results with a two-path single-anchor geometry show that oracle HoloTrace can steer the inferred position toward a desired spoofed location, whereas blind variants mainly provide obfuscation or approximate spoofing due to angle-delay pairing ambiguities.
\cyan{The communication impact is geometry-dependent rather than intrinsic to the proposed method: selecting a spoofed location compatible with the serving beams can preserve channel-rate quality while providing effective location privacy.}
% \blue{High-resolution localization in beyond-5G cellular systems creates new opportunities for location-based services, but also enables \acp{bs} to infer user location from routine pilot transmissions, raising significant privacy concerns. This paper introduces HoloTrace, a signal-level location-privacy framework for single-\ac{bs} localization based on mmWave \ac{mimo}-\ac{ofdm} communication. HoloTrace modifies user-side pilot symbols while keeping the analog precoder fixed to bias localization-relevant channel features, including \ac{aoa}, \ac{aod}, and \ac{tdoa}, toward a prescribed fictitious geometry.}
% \blue{We formulate pilot-level spoofing as a subspace projection-based signal design problem and derive closed-form oracle designs assuming \ac{ue}-side channel-gain knowledge. We further develop blind designs that operate without complex path-gain information, while relying on geometric side information, and explicitly characterize their limitations in multipath scenarios. Simulation results with a two-path single-anchor geometry show that oracle HoloTrace can steer the inferred position toward a desired spoofed location, whereas blind variants mainly provide obfuscation or approximate spoofing due to angle-delay pairing ambiguities. Notably, the communication impact is not fixed by the proposed method, but depends on the selected spoofed geometry. An appropriate spoofed location choice can preserve channel-rate quality while achieving effective location privacy.}
%close to the no-spoofing case in the evaluated scenario.}
\end{abstract}


\acresetall 
 
\vspace{-3mm}
\section{Introduction}
%\IEEEPARstart{M}illimeter-wave (mmWave) MIMO is ...
\vspace{-3mm}
\subsection{Motivation}
Next-generation \cyan{vehicular and mobile networks will increasingly integrate high-precision radio localization with communication, particularly relevant to \acp{cav}. Infrastructure-assisted positioning can complement conventional navigation and onboard sensing in challenging environments, while supporting applications ranging from navigation and traffic management to cooperative and automated mobility~\cite{saleh2026survey,ge2024experimental, piavanini2026magazine}.}
Accurate localization is enabled by the sparsity of the wireless channel in the newly targeted \acp{fr} for beyond \ac{5g} and \ac{6g} systems, i.e., \ac{fr}3 (7.125--24.25 GHz), \ac{fr}2 (24.25--52.6 GHz), and sub-THz (100--300 GHz), combined with the high delay and angle resolution brought by the large available bandwidth and large antenna arrays, respectively~\cite{italiano2025tutorial,behravan2022positioning}.
\cyan{While these capabilities are particularly attractive for \acp{cav} and \ac{its}, they also introduce location-privacy concerns. Vehicular location information can support a wide range of services, but its unauthorized inference may expose users and vehicles to tracking~\cite{xu2022personalized,yadav2020efficient}.}
A critical case arises when the \ac{bs} infers the \ac{ue} location opportunistically from routine uplink pilot transmissions, originally intended for communication and channel estimation, without an explicit positioning request.
\cyan{This represents a distinct physical-layer privacy leakage, since location information can be inferred even when the \ac{ue} does not explicitly report its coordinates or request a positioning service. This paper focuses on such non-consensual secondary use of communication signals by an honest-but-curious infrastructure node, rather than on disrupting legitimate cooperative or safety-critical positioning services.}
Within this line of physical-layer privacy research, this paper proposes a signal-level framework where the \ac{ue} modifies its transmitted pilots to bias the localization-relevant channel features estimated by the \ac{bs}: in this context, obfuscation degrades or prevents localization, whereas spoofing induces a fictitious location estimate.
%This has motivated growing interest in physical-layer location privacy, where \acp{ue}  manipulate their transmitted signals to obfuscate or spoof their location (see Section~\ref{sec:related work}). 
%In this privacy-preserving context, obfuscation degrades or prevents localization, whereas spoofing deliberately induces a fictitious location estimate.
%Although not the main focus of this paper, we note that spoofing is not only limited to connected devices, but also extends to the sensing of passive targets~\cite{Bartoletti:WCL2025,yildirim2025ofdm}. 
\vspace{-8pt}

\subsection{Related Work \blue{on Physical-Layer Location Privacy}}
\label{sec:related work}
Early work in physical-layer \blue{privacy-preserving localization} includes \cite{checa2020location}, which pioneered a \blue{location-privacy} beamforming technique that conceals the \acp{aod} of propagation paths, thereby preventing the \ac{bs} from inferring the \ac{ue}’s direction, at a rate penalty. \blue{A similar approach was considered in~\cite{zhang_privacy_2025} via \ac{crb} optimization, degrading the localization performance of unauthorized \ac{bs} while leaving the legitimate \acp{bs}  unaffected.} The tradeoff between performance and privacy was also explored in~\cite{ayyalasomayajula_users_2023}, which showcased a WiFi-based prototype setup that allows users to obfuscate their location without compromising throughput, adding a delay on the shortest path. 
%delayed the shortest path, to maintain high rate and high privacy. 
%A weakness of this type of approach is that it requires knowledge of the \ac{csi}. 
A key limitation of these approaches is their reliance on accurate \ac{csi} knowledge (either with feedback or leveraging channel reciprocity). An alternative is to inject artificial noise into the transmitted signal to degrade any unauthorized localization attempts~\cite{tomasin2022beamforming}. 
%While this approach is \ac{csi}-free, it is unable to spoof the location. 
\blue{Although this approach has the advantage of being \ac{csi}-free, it can only obfuscate the \ac{ue} location and cannot induce (i.e., spoof) a desired fictitious location estimate.}
An intermediate approach, termed fake path injection~\cite{li2024channel}, avoids the need for \ac{csi} knowledge and harnesses the geometric structure of the mmWave channel to generate additional non-physical paths at the unauthorized \ac{bs}. 
The approach was extended in~\cite{li2025delay} with \ac{dais}, creating fake yet geometrically reasonable paths, thereby enabling the \ac{ue} to spoof its location. 
\blue{A legitimate \ac{bs} can recover the true location through a small secret mapping, whereas an adversary estimates an incorrect one.}
%To preserve functionality, a legitimate \ac{bs} is provided with a small secret message containing the fake parameters and the mapping to the true ones, enabling correct \ac{ue} localization. In contrast, an adversary without this information derives an entirely incorrect location estimate.
%can be given the fake parameters (a small secret message), enabling it to localize the \ac{ue} correctly, whereas an adversary lacking this information obtains a completely wrong location estimate. 

\begin{table*}[t]
\centering
\caption{Comparison of physical-layer location privacy approaches. [S\,=\,Spoofing, O\,=\,Obfuscation].}
\renewcommand{\arraystretch}{1.1}
\begin{tabular}{
|l|c|c|c|c|c|c|c|c|
% |>{\raggedright}p{0.08\linewidth}|
%  >{\centering\arraybackslash}p{0.06\linewidth}|
%  >{\raggedright}p{0.08\linewidth}|
%  >{\raggedright}p{0.08\linewidth}|
%  >{\raggedright}p{0.1\linewidth}|
%  >{\raggedright}p{0.08\linewidth}|
%  >{\raggedright}p{0.08\linewidth}|
%  >{\raggedright\arraybackslash}p{0.10\linewidth}|
}
\hline
\textbf{Ref.} & 
\textbf{Approach} & 
\textbf{Angle privacy} & 
\textbf{Delay privacy} & 
\textbf{Single-BS} & 
\textbf{\blue{UE-BS sync.}} & 
\textbf{CSI-free} & 
\textbf{Beamforming} &%\textbf{Hybrid, digital, or analog} \\
\textbf{Domain}\\
\hline
\cite{checa2020location} & O & \cmark & \xmark & \cmark~(MIMO) & \cmark & \xmark &  Digital & Spatial \\
\cite{zhang_privacy_2025} & O & \cmark & \xmark & \cmark~(MIMO) & \cmark & \xmark &  Analog & Spatial \\
\cite{ayyalasomayajula_users_2023} & S & \xmark & \cmark & \cmark~(MIMO) & \xmark & \xmark &  Digital & Spatial\\
\cite{tomasin2022beamforming} & O & \cmark & \xmark & \cmark~(MIMO) & \xmark & \xmark &  Digital & Spatial \\
\cite{li2024channel} & O & \cmark & \cmark & \cmark~(MISO) & \cmark & \cmark &  Digital & Spatial, Spectral\\
\cite{li2025delay} & S & \cmark & \cmark & \cmark~(MISO) & \cmark & \cmark &  Digital  & Spatial, Spectral\\
\cite{pham2025leveraging, srinivasan2024aoa, pourafzal2025cooperative} & S & \cmark & \xmark & \xmark~(AoA-only) & \cmark & Some cases &  All & Temporal\\
\cite{gao_your_2023,Lv:JSAC2024} & O & \xmark & \cmark & \xmark & \cmark & \cmark &  N/A& Temporal, Spectral \\
\rowcolor{lightblue}
\textbf{This paper} & S & \cmark & \cmark & \cmark~(MIMO) & \xmark & Some cases &  Analog & Temporal, Spectral\\
\hline
\end{tabular}
\label{tab:privacy-comparison}
\vspace{-10pt}
\end{table*}

%The physical-layer privacy preservation techniques that spoof locations have a dual interpretation in the context of \ac{pla}~\cite{tomasin2024physical}. \black{While spoofing can serve as a tool for privacy preservation by misleading unauthorized entities,} a device's radio channel response or spatial signature is \black{also} used as a fingerprint to verify identity. Hence, if a device can spoof its location, it can potentially be incorrectly authenticated. 
\blue{Physical-layer location-spoofing techniques have a dual interpretation in the context of \ac{pla}. While location-spoofing can serve as a privacy-preserving tool to mislead unauthorized entities, a device's radio channel response, or spatial signature, is also commonly used for fingerprint-based identity verification. Hence, if a device can spoof its location, it may circumvent location-based authentication mechanisms.}
For instance,~\cite{pham2025leveraging} investigated the feasibility of spoofing the \ac{aoa}, showing that such attacks are ineffective when the \ac{bs} is equipped with a digital array. 
By contrast, successful spoofing has been demonstrated under analog arrays~\cite{srinivasan2024aoa} and hybrid arrays~\cite{pourafzal2025cooperative}. 
Spoofing is not limited to angles; it can also affect \ac{toa} measurements. This was demonstrated in~\cite{gao_your_2023,Lv:JSAC2024}, which exploited the fact that \acp{prs} are not authenticated, providing the opportunity for an attacker to shift the location estimation of a legitimate \ac{ue}.
Interestingly, many of the above works focused on the problem of privacy preservation for \textit{single-\ac{bs} localization}. This is well-motivated, since at higher frequencies users will generally be in \ac{los} to a single serving \ac{bs} for high-rate communication purposes. Localization is then enabled by either a two-way protocol~\cite{sun2021comparative, abu2020locbound} or via multipath exploitation~\cite{ge2024experimental,kakkavas2021reflective}, without requiring multi-\ac{bs} connectivity. This makes the single-\ac{bs} localization scenario both practically relevant and potentially vulnerable from a privacy perspective.

Table~\ref{tab:privacy-comparison} summarizes the key characteristics of recent works on physical-layer location privacy. The table reveals that existing studies do not fully address spoofing in \textit{realistic single-\ac{bs} positioning scenarios}, despite their increasing relevance in practice~\cite{ge2024experimental}. Specifically, most prior works overlook the combined presence of multiple antennas at both the \ac{ue} and \ac{bs}, unsynchronized \acp{ue}, and analog array architectures. 
{These omissions are not marginal. For instance, assuming a single-antenna \ac{bs} drastically simplifies the problem, making location spoofing easier and deviating from realistic deployment conditions.}
Similarly, models based on perfectly synchronized \acp{ue} often suggest spoofing strategies such as signal delays that do not meaningfully impact single-anchor (i.e., inter-path) \ac{tdoa} localization \cite{SA-JURE, ge2024experimental}. 
Finally, analog arrays, especially at the \ac{ue} side, are the predominant hardware implementation at higher frequencies \cite{ge2024experimental}, and their constraints significantly affect spoofing design. 
\blue{As highlighted in Table~\ref{tab:privacy-comparison}, this paper addresses this gap by proposing a pilot-level spoofing framework for single-\ac{bs} localization that explicitly accounts for analog-array and asynchronous-\ac{ue} settings.}

From Table \ref{tab:privacy-comparison}, a noteworthy pattern emerges: many existing approaches rely on manipulating the spatial-domain signals (via precoders or combiners) \cite{checa2020location,zhang_privacy_2025,ayyalasomayajula_users_2023,tomasin2022beamforming}, possibly in combination with the spectral pilots \cite{li2024channel,li2025delay}. Only a subset of the prior work manipulates the temporal pilots 
\cite{pham2025leveraging, srinivasan2024aoa, pourafzal2025cooperative} in a narrowband system, or the time-frequency pattern in an OFDM system \cite{gao_your_2023,Lv:JSAC2024}.
Our approach does not require changes to the spatial-domain beamformer, \blue{making it compatible with analog arrays whose precoders are fixed or only slowly reconfigured.} While precoder-based attacks can be theoretically more powerful, their practical feasibility is limited in analog beamforming systems, where the precoder is typically implemented through fixed phase shifters, and it is difficult to reconfigure in real time. Moreover, in real-world deployments, the precoder is often tightly integrated with the rest of the hardware and beam management procedures, further limiting an attacker's ability to manipulate it dynamically. On the other hand, modifying the pilot signal offers a more flexible and implementable strategy, especially in scenarios where the attacker has access to the signal generation stage but not to the physical-layer beamforming hardware. 
\blue{As a result, the achievable spoofing performance is more constrained than with full precoder control and depends on available channel knowledge and propagation geometry.}
%\blue{The resulting spoofing capability is more constrained than full precoder control and depends on the channel knowledge and geometry. The achievable spoofing performance is inherently determined by the available channel knowledge and the propagation geometry.}
%In contrast to some prior works that rely on manipulating the precoding matrix to perform spoofing or location deception, our approach targets the pilot signal itself. While precoder-based attacks can be theoretically more powerful, their practical feasibility is limited in analog beamforming systems, where the precoder is typically implemented through fixed phase shifters, and it is difficult to reconfigure in real time. Moreover, in real-world deployments, the precoder is often tightly integrated with hardware constraints and beam management procedures, further limiting an attacker's ability to manipulate it dynamically. On the other hand, modifying the pilot signal offers a more flexible and implementable strategy, especially in scenarios where the attacker has access to the signal generation stage but not to the physical layer beamforming hardware. This makes pilot-level manipulation a more realistic and scalable threat model for evaluating spoofing strategies in current and future wireless systems.

\vspace{-2mm}
\subsection{Contributions}
\blue{This paper develops a physical-layer framework for location-privacy in single-\ac{bs} mmWave \ac{mimo}-\ac{ofdm} links. We consider a realistic positioning scenario in which the \ac{ue} is not synchronized to the \ac{bs}, and both devices are equipped with analog antenna arrays.
Our key insight is that, despite these constraints, a single perturbation block, embedded in the \ac{ue}'s pilot generator, can bias the geometric inversion procedure employed by the \ac{bs} for positioning, while leaving the nominal precoder unchanged. We refer to this solution as \textit{HoloTrace}, short for \emph{holographic trace}, as it operates by slightly reshaping the geometric imprint (or \emph{trace}) left by the transmitted waveform at the physical layer, thereby creating a fictitious (\emph{holographic}) image. Our results show that waveform-level pilot modification can provide location privacy in the evaluated single-\ac{bs} setting; the achievable spoofing accuracy and communication-rate loss depend on the available \ac{csi} and the selected spoofed geometry.}
The main contributions of this work are summarized as follows:

\begin{itemize}
\item \blue{\textbf{Pilot-level spoofing framework for realistic single-BS localization:} 
We introduce a unified framework for physical-layer pilot spoofing targeting location estimation in single-\ac{bs} setups. The framework jointly manipulates key signal features (\ac{aoa}, \ac{aod}, and \ac{tdoa}) and is applicable across a broad range of localization architectures, including uplink and downlink, \ac{simo}, \ac{miso}, and \ac{mimo} \ac{ofdm} systems. 
The proposed framework operates at the physical-signal level;  protocol-stack aspects are beyond the scope of this work.}
\blue{In a 5G NR context, the perturbation is applied to uplink positioning pilots (e.g., \acp{srs}), which occupy dedicated time-frequency resources distinct from those used for data transmission 
%separate from the data channels 
(physical uplink shared channel and its associated demodulation pilots).} 
%Consequently, data transmission and demodulation remain unaffected.}
%We formally introduce the spoofing problem and also present the concept of perfect spoofing
%We formulate location privacy preservation as a signal-level spoofing problem, where the UE strategically perturbs the pilot signals to deceive single-BS localization algorithms. This approach largely preserves communication functionality by maintaining precoder consistency and protocol transparency, without requiring network-side modifications.

\item \textbf{Unified spoofing approach with \ac{csi}:}
\blue{Assuming channel-gain knowledge at the \ac{ue}, we derive closed-form HoloTrace pilot designs that exactly reproduce the desired noiseless received signal corresponding to prescribed \ac{aoa}, \ac{aod}, \ac{toa}, and their combinations. This oracle setting provides an upper bound on spoofing capability.} 

%We formulate \ac{aoa}, \ac{aod}, and \ac{tdoa} spoofing as a unified rank-constrained projection problem. We provide closed-form pilot design strategies under both oracle path-gain knowledge and blind \ac{ls} estimation.

\item \textbf{Unified spoofing approach without \ac{csi}:}
\blue{We propose blind spoofing methods that require no channel-gain knowledge, while relying on geometrical information. We derive approximate Hadamard- and Kronecker-based pilot designs and analyze their limitations, including power ambiguity and path association loss. For angular parameters, we provide a unified blind design, while for \ac{tdoa} we introduce a novel fake-path injection technique.}

%\item \textbf{Joint \ac{aoa}-\ac{aod}-\ac{tdoa} spoofing for \ac{mimo}-\ac{ofdm}:}  

%By integrating the single-parameter spoofing modules, we construct a composite perturbation block that remains transparent to the physical layer but significantly degrades the localization capabilities of the single \ac{bs}.
\item \blue{\textbf{Comprehensive theoretical and numerical evaluation:}
We assess the spoofing and localization performance of HoloTrace under varying transmit power and uncertainties, and compare against a representative shift-based spoofing strategy. Results show that oracle HoloTrace can closely realize the target spoofed location in the considered two-path scenario, while blind variants induce obfuscation or approximate spoofing depending on path association. Communication-rate loss is geometry-dependent, and spoofed positions can be selected to minimize rate degradation relative to the true position.}
% \black{\item \textbf{Analytical privacy bound.}  
%           We derive a \ac{mcrb}, showing that, for realistic antenna and bandwidth settings, the \ac{peb} inflation grows ...}
%We adopt the \ac{gospa} metric~\cite{GOSPA} to evaluate spoofing effectiveness in realistic multipath environments. Using a $24 \times 16$ antenna configuration, 400~MHz bandwidth, and 3GPP-compliant two-path channels, we show that the proposed solution ...
% \begin{itemize}
%     \item inflates localization RMSE from ...~m to $>...$~m at 23~dB SNR,
%     \item incurs $<\!1$~dB loss in uncoded BER and $<\!3\%$ DSP overhead.
% \end{itemize}
\end{itemize}

The paper is organized as follows.
% Section~\ref{sec:sys_model} introduces the system model, describing the scenario, the channel, and the signal formulation.
% Section~\ref{sec:spoof_model} presents the localization framework, detailing the operational channel, the positioning procedure, and the communication link.
Section~\ref{sec:sys_model} introduces the system model and localization framework, describing scenario, channel, signal formulation, positioning procedure, and communication link.
Section~\ref{sec:holotrace} outlines the proposed HoloTrace spoofing strategy, covering both the coherent (oracle) and non-coherent (blind) regimes.
Section~\ref{sec:results} describes the simulation setup and analyzes the results.
Lastly, Section~\ref{sec:conclusion} concludes the paper by summarizing key findings, discussing limitations, and outlining directions for future work.

\paragraph*{Notation}
Matrices are defined with bold uppercase (e.g., $\mathbf{A}$), vectors with bold lowercase (e.g., $\mathbf{a}$), and tensors with calligraphic (e.g., $\mathcal{A}$). Element indices are indicated with lowercase superscripts (e.g., $\mathbf{a}^i$). Operations include conjugate~$(\cdot)^\ast$, transpose~$(\cdot)^\mathsf{T}$, conjugate transpose~$(\cdot)^\mathsf{H}$, Moore-Penrose inverse~$(\cdot)^\dagger$, Kronecker product~$(\otimes)$, Hadamard product~$(\odot)$, Hadamard division~$(\oslash)$, Khatri-Rao product~$(\circledast)$, and mode-$i$ tensor product~$(\times_i^{})$ between the $i$-th dimension of the first tensor and the \nth{2} dimension of the second one. The diag$(\cdot)$ operator is used to create a diagonal matrix from a vector (e.g., diag$(\mathbf{a})$), \black{$\mathcal{R}(\cdot)$ is the range of a matrix}, vec$(\cdot)$ indicates the vectorization function, $\left[\cdot\right]_j^{}$ denotes the tensor unfolding operation over the $j$-th dimension, $\Vert \cdot\Vert$ defines the norm-2 \blue{for vectors and the Frobenius norm for matrices or tensors}, $\mathbf{I}_M^{}\in\mathbb{C}^{M\times M}$ is the identity matrix, and $\mathbf{0}_{M\times N}^{}\in\mathbb{C}^{M\times N}$ and $\mathbf{1}_{M\times N}^{}\in\mathbb{C}^{M\times N}$ are all-zero and all-one matrices, respectively. \black{$\mathbf{P}^{}_{\mathbf{A}}=\mathbf{A}\mathbf{A}^\dagger$ is the projector onto the column space of $\mathbf{A}$, and $\mathbf{P}^{\perp}_{\mathbf{A}}=\mathbf{I}-\mathbf{P}^{}_{\mathbf{A}}$ is the corresponding orthogonal projector.}
%, penetrating face-splitting product ($\left[\circ\right]$)
%$\mathbf{0}_M\in\mathbb{C}^{M\times 1}$ is a zeros vector, $\mathbf{1}_M\in\mathbb{C}^{M\times 1}$ is a ones vector, 

\section{System Model}
\label{sec:sys_model}
We consider a $2$D uplink scenario where a \blue{fixed} single \ac{bs} at known position $\mathbf{p}_\mathrm{BS}^{} \in \mathbb{R}^2$ and orientation $o_\mathrm{BS}^{}$ receives signals from a \blue{quasi-static} \ac{ue} with unknown position $\mathbf{p}_\mathrm{UE}^{}$ and orientation $o_\mathrm{UE}^{}$, in a propagation environment with $L$ paths: one \ac{los} path ($\ell=0$) and $L-1$ single-bounce \ac{nlos} paths reflected by distinct \acp{sp} at positions $\{\mathbf{p}_\ell^{}\}_{\ell=1}^{L-1}$ (see~Fig.\,\ref{fig:overview}). The \ac{ue} is asynchronous with the \ac{bs}, i.e., a clock offset $b_{\text{UE}}^{}$ is present.

\input{figures/scenario}

\subsection{Channel and Signal Model}
The \ac{mimo}-\ac{ofdm} channel at subcarrier $k \in \{0,\dots,K-1\}$ is
\begin{equation}
    \mathbf{H}_k^{} = \sum_{\ell=0}^{L-1} \alpha_\ell^{}\, e^{-j2\pi k\Delta f \tau_\ell^{}}\,
    \mathbf{a}_{N_\mathrm{R}^{}}^{}(\theta_\ell^{})\, \mathbf{a}_{N_\mathrm{T}^{}}^\mathsf{T}(\phi_\ell^{}),
    \label{eq:channel}
\end{equation}
where $\alpha_\ell^{} \in \mathbb{C}$ is the complex path gain, \black{$\Delta f$ is the subcarrier spacing,} $\tau_\ell^{}$ is the path delay, $\theta_\ell^{}$ and $\phi_\ell^{}$ are the \ac{aoa}/\ac{aod}, $\mathbf{a}_{N_{\mathrm{R}^{}}}^{}(\theta)$ and $\mathbf{a}_{N_{\mathrm{T}^{}}}^{}(\phi)$ are the normalized \ac{ula} responses, with $N_\mathrm{T}^{}$ and $N_\mathrm{R}^{}$ indicating the number of antenna elements at the \ac{ue} and \ac{bs} sides, respectively. For a generic \ac{ula} of size $N$, the array response at angle $\vartheta$ is \blue{$\mathbf{a}_N^{}(\vartheta)^{} = \frac{1}{\sqrt{N}}[1 \; e^{j\pi \sin\vartheta} \; \cdots \; e^{j\pi(N-1)\sin\vartheta}]^\mathsf{T}$}.
%i.e.,
% \begin{equation}
%     \mathbf{a}_N^{}(\vartheta)^{} = [1 \; e^{j\pi \sin\vartheta} \; \cdots \; e^{j\pi(N-1)\sin\vartheta}]^\mathsf{T} .
%     \label{eq:array_response}
% \end{equation}
%for array size $N$ and angle $\vartheta$.
The \ac{los} measurements $\{\tau_0^{},\theta_0^{},\phi_0^{}\}$ are related to the \ac{ue} state by 
\begin{align}
\tau_0^{} & =    \Vert  \mathbf{p}_\mathrm{BS}^{}-\mathbf{p}_\mathrm{UE}^{}\Vert /c + b_{\text{UE}}^{},\\
\theta_0^{} & = \blue{\text{atan2}(y_{\text{UE}}^{} - y_{\text{BS}}^{},\; x_{\text{UE}}^{} - x_{\text{BS}}^{}) - o_{\text{BS}}^{}},\\
\phi_0^{} & = \text{atan2}(y_{\text{BS}}^{} - y_{\text{UE}}^{},\; x_{\text{BS}}^{} - x_{\text{UE}}^{}) - o_{\text{UE}}^{},
\end{align}
\black{where $c$ denotes the speed of light.} It is important to note that $b_{\text{UE}}^{}$ and $o_{\text{UE}}^{}$ are unknowns in the positioning problem. 

The transmit pilot at subcarrier $k$, \ac{ofdm} symbol $m$, and block $s$ is $x_{m,s,k}^{}~\in~\mathbb{C}$, with $|x_{m,s,k}^{}|=1$, for $m \in \{0,\dots,M-1\}$ and $s~\in~\{0,\dots,S-1\}$. \black{The channel is assumed to remain constant across all $m$ and $s$.}
The exhaustive beamforming approach is employed \cite{beamforming}: the transmit precoder $\mathbf{f}_s^{}~\in~\mathbb{C}^{N_\mathrm{T}^{}\times 1}$ remains fixed within block $s$, and the receiving combiner $\mathbf{w}_m^{}~\in~\mathbb{C}^{N_\mathrm{R}^{}\times 1}$ cycles over $M$ unit-norm vectors across symbols ($\|\mathbf{w}_m^{}\|^2 = 1$, $\|\mathbf{f}_s^{}\|^2 = 1$). 
\blue{The precoder codebook $\mathbf{F}=[\mathbf{f}_0^{}\,\cdots\,\mathbf{f}_{S-1}^{}]$ is applied across the $S$ blocks following a configured beam-management procedure (e.g., based on \acp{srs}), and is therefore known at the \ac{bs}. Although the \ac{aod} of a \ac{nlos} path cannot be solely inferred from the geometry of an unknown reflector, it appears in the received signal via the known transmit factor $\mathbf{C}=\mathbf{F}^\mathsf{H}\mathbf{A}_{N_\mathrm{T}^{}}(\bm{\phi})$ in \eqref{eq:def_tensor_H}. Consequently, the \ac{aod} is jointly identifiable with the \ac{aoa} and delay by means of the structured tensor estimator (e.g., tensor- or subspace-based methods), as commonly assumed in single-\ac{bs} multipath-assisted localization~\cite{NuriaSAM2022}. The pilot modifications proposed in HoloTrace do not alter this beam-management codebook.}
Assuming that the total transmit power $P_t^{}$ is uniformly distributed across all the $K$ subcarriers over bandwidth $K\Delta f$, such that the energy per subcarrier is $E_s^{} = P_t^{} /(K \Delta f)$, the received signal at subcarrier $k$, symbol $m$, and block $s$ is
\begin{equation}
    y_{m,s,k}^{} = {\sqrt{E_s^{}}}\mathbf{w}_m^\mathsf{H} \mathbf{H}_k^{} \mathbf{f}_s^{\ast} x_{m,s,k}^{} + q_{m,s,k}^{},
    \label{eq:rx_signal}
\end{equation}
%{with $q_{m,s,k}^{} \sim \mathcal{CN}(0,\sigma_q^2)$, and  $\sigma^2_q = {N_0\Delta f}$, where $N_0^{}$ is the noise power spectral density.}
\blue{with $q_{m,s,k}^{} \sim \mathcal{CN}(0,N_0^{})$, where $N_0^{}$ denotes the noise energy per received subcarrier sample.}
%\mb{The unit of measure do not match} \li{pls explain/correct}
%\mb{in case you need to keep $N_0$, there should be the multiplication by the bandwidth. By the way, do we need to parametrize with respect to $N_0$? Isn't enough to specify a sigma?} \li{we need it for the communication strategy section}
The received signal model is equivalently written in tensor form:
\begin{equation}\label{eq:def_ten_Y}
    \mathcal{Y} = \mathcal{H} \odot \mathcal{X} + \mathcal{Q},
\end{equation}
where $\mathcal{Y},\mathcal{X},\mathcal{Q}$ are order-3 tensors indexed by $(m,s,k)$ and the \blue{effective noiseless channel tensor} admits the decomposition
\begin{equation} \label{eq:def_tensor_H}
    \blue{\mathcal{H} = \sqrt{E_s^{}}\,\mathcal{A} \times_1^{} \mathbf{B} \times_2^{} \mathbf{C} \times_3^{} \mathbf{D}} ,
\end{equation}
with $\mathcal{A}$ diagonal, $\mathcal{A}^{\ell,\ell,\ell}~=~\alpha_\ell^{}$, $\mathbf{B}~=~\mathbf{W}^\mathsf{H} \mathbf{A}_{N_\mathrm{R}^{}}^{}${, where $\mathbf{W}=[\mathbf{w}_0^{} \, \cdots \, \mathbf{w}_{M-1}^{}]$, }$\mathbf{C}~=~\mathbf{F}^\mathsf{H} \mathbf{A}_{N_\mathrm{T}^{}}^{}${, where $\blue{\mathbf{F}=[\mathbf{f}_0^{} \, \cdots \, \mathbf{f}_{S-1}^{}]}$, $\mathbf{A}_{N_\mathrm{R}^{}}^{}~=~[\mathbf{a}_{N_\mathrm{R}^{}}(\theta_0^{}) \, \cdots \, \mathbf{a}_{N_\mathrm{R}^{}}(\theta_{L-1}^{})]$ and $\blue{\mathbf{A}_{N_\mathrm{T}^{}}^{}~=~[\mathbf{a}_{N_\mathrm{T}^{}}(\phi_0^{}) \, \cdots \, \mathbf{a}_{N_\mathrm{T}^{}}(\phi_{L-1}^{})]}$} are the array response matrices at \ac{bs} and \ac{ue}, and $\mathbf{D} = [\mathbf{d}(\tau_0^{}) \, \cdots \, \mathbf{d}(\tau_{L-1}^{})]$, with $\mathbf{d}(\tau_\ell^{}) = [1 \; e^{-j2\pi \Delta f \tau_\ell^{}} \; \cdots \; e^{-j2\pi (K-1)\Delta f \tau_\ell^{}}]^\mathsf{T}$.

\subsection{Operational Model and Localization Procedure}

\label{sec:op_model_chan_est}
The \ac{bs} estimates the channel for communication purposes, leveraging the known pilot structure. 
%After removing the known phase of the unimodular pilots, we will consider that 
\blue{We assume that the nominal (agreed) pilot sequence is} $\mathcal{X}^{:,:,k}~=~\mathbf{1}_{M\times S}^{},$ $\forall k \in \{0, \dots, K-1\}$. %, without loss of generality. 
The received signal tensor is thus expressed as
%\begin{equation}
    $\mathcal{Y} = \mathcal{H} + \mathcal{Q}$,
    %\label{eq:tensorform_comm}
%\end{equation}
where $\mathcal{Y}$, $\mathcal{H}$, and $\mathcal{Q}$ are defined as in~\eqref{eq:def_ten_Y}.
\blue{Since the \ac{bs} estimates the channel under the assumed nominal pilot sequence, channel effects and pilot manipulation are indistinguishable within the adopted signal model.}
%\blue{Since the \ac{bs} attributes the received signal entirely to the channel under the assumed nominal (constant-modulus) pilots, it cannot separate the channel from a possibly modified pilot, and thus has no means to verify that the transmitted pilots are indeed constant-modulus.}
%In the oracle regime, the spoofed signal is itself a valid sparse geometric channel, so the manipulation is undetectable from a single passive observation. This is precisely the property exploited by HoloTrace.}
%Since the \ac{bs} assumes nominal constant-modulus pilots and attributes the received signal entirely to the channel, it cannot distinguish the true channel from a manipulated pilot or verify that the pilots are constant-modulus. In the oracle regime, the spoofed signal is itself a valid sparse geometric channel, so the manipulation is undetectable from a single passive observation. This is exactly the property exploited by HoloTrace.

Given a channel parameter estimator $\bm{\zeta}(\cdot)$, capable of jointly recovering the tuple $(\tau_\ell^{}, \theta_\ell^{}, \phi_\ell^{})$ for all the paths, we define the output of the estimator as
%\begin{equation}
$\widehat{\boldsymbol{\rho}}=\bm{\zeta}(\mathcal{Y})$,
%\end{equation}
where $\widehat{\boldsymbol{\rho}} = [\widehat{\boldsymbol{\rho}}_0^{} \; \cdots \; \widehat{\boldsymbol{\rho}}_{\widehat{L}-1}^{}] \in \mathbb{R}^{3 \times \widehat{L}}$ collects the estimated parameters for all the $\widehat{L}$ detected paths, with $\widehat{\boldsymbol{\rho}}_\ell^{}=[ \widehat{\tau}_\ell^{}\,
\widehat{\theta}_\ell^{} \,
\widehat{\phi}_\ell^{}]^{\mathsf{T}}$, $\ell \in \{ 0, \ldots, \widehat{L}-1\}$.
%\begin{equation}
   % \widehat{\boldsymbol{\rho}}_\ell^{} = \begin{bmatrix}
   %     \widehat{\tau}_\ell^{} \\
    %    \widehat{\theta}_\ell^{} \\
     %   \widehat{\phi}_\ell^{}
    %\end{bmatrix}, \quad \ell = 0, %\ldots, \widehat{L}-1.
%\end{equation}
Here, $\widehat{\tau}_\ell^{}$ denotes the estimated delay, and $\widehat{\theta}_\ell^{}$ and $\widehat{\phi}_\ell^{}$ are the estimated \ac{aoa} and \ac{aod} of the $\ell$-th path, respectively.

\blue{The collection of channel parameter estimates $\widehat{\boldsymbol{\rho}}$ enables \ac{ue} localization.} Considering the \ac{los} and one \ac{nlos} path (selected by its amplitude, delay, or other criteria), with their respective indexes $\ell=0$ and $\ell=1$, a position estimate\footnote{More sophisticated methods can be found, e.g., in \cite{gonzalez2024integrated, NuriaSAM2022}, which can be used as initializers for maximum likelihood estimation.} is \cite{chen2023learning}
\begin{equation}
    \widehat{\mathbf{p}}_\mathrm{UE}^{} = \mathbf{p}_\mathrm{BS}^{} + \widehat{d}_0^{}\, \mathbf{u}(\widehat{\theta}_0^{}\blue{+o_\mathrm{BS}^{}}),
    \label{eq:p_est}
\end{equation}
where $\mathbf{u}(\vartheta) = [\cos\vartheta \; \sin\vartheta]^{\mathsf{T}}$, and $\widehat{d}_0^{}$ is the estimated \ac{bs}-\ac{ue} distance, computed via the law of sines:
\begin{equation}
    \widehat{d}_0^{} = \frac{c\, \Delta\widehat{\tau}_{1,0}^{} \sin(\Delta\widehat{\theta}_{1,0}^{} + \Delta\widehat{\phi}_{1,0}^{})}{\sin(\Delta\widehat{\theta}_{1,0}^{}) + \sin(\Delta\widehat{\phi}_{1,0}^{}) - \sin(\Delta\widehat{\theta}_{1,0}^{} + \Delta\widehat{\phi}_{1,0}^{})},
    \label{eq:d_los}
\end{equation}
where $\Delta\widehat{\tau}_{1,0}^{}=\widehat{\tau}_1^{}-\widehat{\tau}_0^{}$, $\Delta\widehat{\theta}_{1,0}^{} = |\widehat{\theta}_1^{} - \widehat{\theta}_0^{}|$, and $\Delta\widehat{\phi}_{1,0}^{} = |\widehat{\phi}_1^{} - \widehat{\phi}_0^{}|$. 
\blue{The localization expression is well-defined under the standard regularity conditions for single-\ac{bs} localization, including correct path detection and association, resolvable delay-angle signatures, and a non-singular geometry (i.e., the denominator in \eqref{eq:d_los} is bounded away from zero).} 
This approach is \ac{bs}-\ac{ue} clock offset independent.

%\hui{Is MLE considered for refining FLEX? MLE is also useful for bias estimation in MCRB.} \li{it is not considered for refining FLEX}

% Formally,
% \begin{equation}
%     \widehat{\boldsymbol{\rho}} = \mathrm{FLEX}(\mathcal{Y}),
% \end{equation}
% where $\widehat{\boldsymbol{\rho}} \in \mathbb{C}^{3 \times \widehat{L}}$ is the estimation of measurement vector $\boldsymbol{\rho} = \left[ \rho_1, \rho_2, \ldots, \rho_{{L}} \right] \in \mathbb{C}^{3 \times {L}}$ defined as $\rho_l = \left[ \widehat{{\tau_l}},\, \sin\widehat{{\theta}_l},\, \sin\widehat{{\phi_l}} \right]^{\mathsf{T}}$.
% \begin{equation}
%     \left\{ \widehat{\boldsymbol{\tau}},\, \widehat{\boldsymbol{\theta}},\, \widehat{\boldsymbol{\phi}} \right\} = \mathrm{FLEX}(\mathcal{Y}),
% \end{equation}
% with $\widehat{\boldsymbol{\tau}} \in \mathbb{R}^{\widehat{L} \times 1}$, $\widehat{\boldsymbol{\theta}}, \widehat{\boldsymbol{\phi}} \in \mathbb{R}^{\widehat{L} \times 2}$.
% \li{we could include here the estimated measurement vector $\widehat{\bm{\rho}}$}
\vspace{-10pt}
\subsection{Communication Strategy}
The \ac{bs} may optimize the communication by selecting combiner and precoder beams based on different strategies.
% In a coherent strategy, the strongest path (typically $\ell=0$ for \ac{los}) is used to select \li{remove it}
% \begin{align}
%     \mathbf{w}_\mathrm{comm}^{} &= \underset{ \mathbf{w}\in\mathbf{W}}{\text{arg\,max}}  \, \mathbf{w}^{\mathsf{H}} \mathbf{a}_\mathrm{BS}^{}(\widehat{\theta}_0^{}), \\
%     \mathbf{f}_\mathrm{comm}^{} &= \underset{ \mathbf{f}\in\mathbf{F}}{\text{arg\,max}} \, \mathbf{f}^{\mathsf{H}} \mathbf{a}_\mathrm{UE}^{}(\widehat{\phi}_0^{}).
% \end{align}
% Alternatively, 
In a non-coherent strategy, the \ac{bs} selects the beam pair maximizing total received power, i.e.,
%\begin{equation}
    $\blue{(\widehat{m}, \widehat{s}) =  \underset{ m' \in \{0,\ldots,M-1\},\, s' \in \{0,\ldots,S-1\} }{\text{arg\,max}}  \sum_{k=0}^{K-1} |y_{m',s',k}^{}|^2}$, %\quad 
%\end{equation}
where $\mathbf{w}_\mathrm{comm}^{} = \mathbf{w}_{\widehat{m}}^{},$ and $ \mathbf{f}_\mathrm{comm}^{} = \mathbf{f}_{\widehat{s}}^{}.$ 
The achievable sum rate is \cite{NuriaSAM2022}
\begin{equation}
    R = \sum_{k=0}^{K-1} \Delta f \log_2 \left(1 + \frac{\gamma_k^{} E_s^{}}{N_0^{}} \right),
    \label{eq:comm_rate}
\end{equation}
% with $\gamma_k = \left| \sum_{\ell=0}^{L-1} \alpha_\ell\, \mathbf{w}_\mathrm{comm}^{\mathsf{T}} \mathbf{a}_\mathrm{BS}(\theta_\ell)\, \mathbf{a}_\mathrm{UE}^{\mathsf{T}}(\phi_\ell)\, \mathbf{f}_\mathrm{comm}\, e^{-j2\pi k\Delta f \tau_\ell} \right|^2$.
with $\gamma_{k}^{}=\left|\sum_{\ell=0}^{L-1}\alpha_{\ell}^{}\mathbf{w}_{\text{comm}}^{\mathsf{H}}\mathbf{a}_{\text{BS}}^{}(\theta_{\ell}^{})\mathbf{a}_{\text{UE}}^{\mathsf{T}}(\phi_{\ell}^{})\mathbf{f}_{\text{comm}}^{\ast}e^{-\jmath2\pi k\Delta{f}\tau_{\ell}^{}}\right|^{2}$.
\blue{Since the serving beams are selected from the received power and the rate \eqref{eq:comm_rate} is evaluated on the true channel, a spoofing design that preserves the \ac{los} \ac{aoa}/\ac{aod} retains beams aligned with the true \ac{los}. The single-anchor position estimate \eqref{eq:p_est}--\eqref{eq:d_los} is then displaced mainly through the range $\widehat{d}_0^{}$, which depends on the inter-path differences. A spoofed location along the true \ac{los} bearing but at a different range therefore moves the inferred position while preserving the communication beams, and hence the rate. The communication impact is thus geometry-dependent rather than intrinsic.}
%$N_0^{}$ the noise power spectral density, and $E_s^{} = P_t^{} /(K \Delta f)$.%$E_s^{}=\vert x_{m^\star, s^\star, k}^{} \vert^2.$ 

\vspace{-10pt}
\subsection{Threat Model}
The operational paradigm considered above enables the \ac{bs} to infer the \ac{ue} location without explicit cooperation or awareness from the \ac{ue}. \blue{The adversary is modeled as an \emph{honest-but-curious} single-anchor observer: it complies with the communication protocol but reuses routine uplink pilots to perform \textit{non-consensual} localization. We study the signal-level feasibility of such attacks under a single-snapshot channel-estimation model. Protocol-level mechanisms such as pilot authentication, repeated-transmission consistency checks, or active spoofing detection are outside the scope of this work.} 
%Although we refer to it as the \ac{bs}, the framework applies to any such passive single-anchor observer with known position and orientation.} 
In this non-consensual scenario, the \ac{bs} exploits standard channel estimation procedures and known pilot structures, originally intended for communication, to recover geometric parameters $\widehat{\boldsymbol{\rho}}$
% $(\tau_\ell^{}, \theta_\ell^{}, \phi_\ell^{})$ 
that uniquely determine the \ac{ue} position via passive observations. The \ac{ue} is neither notified of nor able to control the subsequent use of these estimates for localization purposes. This presents a significant privacy risk: the \ac{bs} can continuously and covertly track the \ac{ue} by reusing routine uplink pilot transmissions, without the \ac{ue}'s consent or participation, and without any protocol-level safeguard. 

\blue{To mitigate non-consensual localization, we propose HoloTrace, a signal-level privacy-preserving framework in which the UE intentionally modifies the transmitted pilot symbols $x_{m,s,k}$ to bias the location inferred by the \ac{bs}. HoloTrace does not require protocol changes or network-side support, but its impact on communication performance depends on the selected spoofed location.}
\vspace{-5pt}

\section{HoloTrace: Location Privacy Framework}
\label{sec:holotrace}
%%%%%%%%%%%%%%%%%%%%%%%%%%%%%%% Alireza Version %%%%%%%%%%%%%%%%%%%%%%%%
%To mitigate uncooperative localization, we consider a privacy-preserving framework in which the \ac{ue} generates a modified pilot tensor $\widetilde{\mathcal{X}} \in \mathbb{C}^{M \times S \times K}$ to induce the \ac{bs} to estimate a prescribed, fictitious location, while the analog/digital precoder remains unchanged. 
%\mb{In case the modification in the last paragraph are applied, I would start this section as:
% "
% To induce the \ac{bs} to estimate a prescribed, fictitious location, still keeping the analog/digital precoder unchanged, the \ac{ue} generates a modified pilot tensor $\widetilde{\mathcal{X}} \in \mathbb{C}^{M \times S \times K}$.
% "
% }

To induce the \ac{bs} to estimate a prescribed, fictitious location, still keeping the analog/digital precoder unchanged, the \ac{ue} generates a modified pilot tensor $\widetilde{\mathcal{X}} \in \mathbb{C}^{M \times S \times K}$.
Since the \ac{bs} exploits both the scattering geometry and single-anchor \ac{tdoa} measurements to resolve the \ac{ue} position, geometric constraints, inherent to the measurement model, preclude effective location spoofing by independently shifting each parameter. Consequently, a successful spoofing attack requires jointly inducing fictitious estimates of both the \ac{ue} and \ac{sp} locations, such that the manipulated measurements remain consistent with the propagation geometry.

\vspace{-7pt}
%%%% HENK Version %%%%
\subsection{Position Spoofing Formulation}
\label{sec:spoofing_formulation}

Given an \emph{assumed model} at the \ac{bs}, based on the agreed pilots $\mathcal{X}$ (which in the remainder of the paper will be set to all ones and can thus be omitted), $\mathcal{Y}_{\text{am}}^{}\sim\mathcal{CN}(\mathcal{Y};
\mathcal{H}(\bm{\rho},\bm{\alpha}) \odot \mathcal{X},\sigma_q^{2}\mathcal{I})$, 
\black{with $\bm{\rho}=[\bm{\tau}^{\mathsf{T}} \, \bm{\theta}^{\mathsf{T}} \, \bm{\phi}^{\mathsf{T}} ]^\mathsf{T}$, $\bm{\tau}= [\tau_0^{}\, \cdots \, \tau_{L-1}^{}]^\mathsf{T}$, $\bm{\theta}= [\theta_0^{}\, \cdots \, \theta_{L-1}^{}]^\mathsf{T}$, $\bm{\phi}= [\phi_0^{}\, \cdots \, \phi_{L-1}^{}]^\mathsf{T}$ and $\bm{\alpha}= [\alpha_0^{}\, \cdots \, \alpha_{L-1}^{}]^\mathsf{T}$,} the \ac{bs} always estimates the channel parameters as 
\begin{equation}
\widehat{\bm{\rho}},\widehat{\bm{\alpha}}= \underset{ \bm{\rho},\bm{\alpha} }{\text{arg\,min}} \, \Vert\mathcal{Y}^{} - \mathcal{H}(\bm{\rho},\bm{\alpha}) \odot \mathcal{X}\Vert^{2},
\end{equation}
\black{irrespective of how  $\mathcal{Y}$ is generated.} 
The \ac{ue} modifies the pilot signal to $\widetilde{\mathcal{X}}$ (subject to \ac{ue}-side power constraints) so that the \emph{true model} becomes $\mathcal{Y}_{\text{tm}}^{}(\widetilde{\mathcal{X}})\sim\mathcal{CN}(\mathcal{Y};
\mathcal{H}(\bm{\rho},\bm{\alpha}) \odot \widetilde{\mathcal{X}},\sigma_q^{2}\mathcal{I})$. 
 In order to create a spoofed channel with parameters $\overline{\bm{\rho}}$ (corresponding to a spoofed location), the value of $\widetilde{\mathcal{X}}$ is selected so that 
\begin{equation}
C(\bm{\alpha}',\bm{\rho}'|\widetilde{\mathcal{X}})=\Vert\mathcal{Y}_{\text{tm}}^{}(\widetilde{\mathcal{X}})-\mathcal{H}(\bm{\rho}',\bm{\alpha}') \odot \mathcal{X}\Vert^{2} \label{eq:spoofingCriterion1}
\end{equation}
is minimal for $\bm{\rho}'=\overline{\bm{\rho}}$. \black{Since the \ac{ue} does not have access to the noisy observation $\mathcal{Y}_{\text{tm}}^{}(\widetilde{\mathcal{X}})$
at the \ac{bs}, it can replace $\mathcal{Y}_{\text{tm}}^{}(\widetilde{\mathcal{X}})$
with the noise-free spoofed signal, i.e., the mean $\mathcal{H}(\bm{\rho},\bm{\alpha}) \odot \widetilde{\mathcal{X}}$ (which depends on the true $\bm{\rho}$ and $\bm{\alpha}$ values), leading to the following expression 
\begin{equation}
\mathcal{C}(\bm{\alpha}',\bm{\rho}'|\widetilde{\mathcal{X}},\bm{\alpha},\bm{\rho})=\Vert\mathcal{H}(\bm{\rho},\bm{\alpha}) \odot \widetilde{\mathcal{X}}-\mathcal{H}(\bm{\rho}',\bm{\alpha}') \odot \mathcal{X}\Vert^{2} \label{eq:spoofingCriterion2}
\end{equation}
to be minimal for $\bm{\rho}'=\overline{\bm{\rho}}$ and for some $\bm{\alpha}'\neq \bm{0}_L^{}$. }
%Since $C(\bm{\alpha}',\bm{\rho}'|\widetilde{\mathcal{X}})$ is quadratic in $\bm{\alpha}'$, a closed-form estimate exists for $\bm{\alpha}'$ conditional on $\bm{\rho}'$, say $\bm{\alpha}(\bm{\rho}')$. This leads to the following spoofing signal design:
%\begin{equation}
%\widetilde{\mathcal{X}}=
 %\underset{ {\mathcal{K}}}{\text{arg\,min}} \, C(\bm{\alpha}(\overline{\bm{\rho}'}),\overline{\bm{\rho}'}|{\mathcal{K}}).
% \label{eq:spoofingCriterion2}
%\end{equation}
%Since the \ac{ue} does not have access to the noisy observation $\mathcal{Y}_{\text{tm}}^{}(\widetilde{\mathcal{X}})$ at the \ac{bs}, it can replace $\mathcal{Y}_{\text{tm}}^{}(\widetilde{\mathcal{X}})$
%with the noise-free spoofed signal, i.e., the mean $\mathcal{H}(\bm{\rho},\bm{\alpha}) \odot \widetilde{\mathcal{X}}$ (depending on the true $\bm{\rho}$ and $\bm{\alpha}$ values)
This is equivalent to designing a certain bias term in the sense of the \ac{mcrb} \cite{fortunati2017performance}. 
%\begin{align}
%    \widetilde{\mathcal{X}}=
% \underset{ {\mathcal{K}}}{\text{arg\,min}} \, \Vert\mathcal{H}(\bm{\rho},\bm{\alpha}) \odot {\mathcal{K}}-\mathcal{H}(\bm{\rho}',\bm{\alpha}') \odot \mathcal{X}\Vert^{2} .
% \label{eq:spoofingCriterion3}
%\end{align}
%
An even stronger \emph{perfect spoofing} design can also be formulated as a constrained feasibility problem: 
\begin{align}
\text{find} & \quad\widetilde{\mathcal{X}}
\label{eq:spoofingCriterion3}\\ 
\nonumber \text{s.t.} & \quad 
\mathcal{C}(\bm{\alpha}',\overline{\bm{\rho}}|\widetilde{\mathcal{X}},\bm{\alpha},\bm{\rho})
=0 
\end{align}
for some $\bm{\alpha}'\neq \bm{0}$. 

Since $\mathcal{H}(\bm{\rho},\bm{\alpha}) \odot \widetilde{\mathcal{X}}$ depends on the values of the complex channel gains $\bm{\alpha}$ and the location-dependent geometric channel parameters $\bm{\rho}$, pilot spoofing breaks down into two regimes: a coherent regime where the \ac{ue} knows both $\bm{\alpha}$ and $\bm{\rho}$ (called \emph{oracle} solution) and a non-coherent regime where the \ac{ue} knows $\bm{\rho}$ but not $\bm{\alpha}$ (called \emph{blind} solution), in which case \eqref{eq:spoofingCriterion2} and \eqref{eq:spoofingCriterion3} should hold $\forall \bm{\alpha}$. Variations include the knowledge of statistics on $\bm{\alpha}$ and on the number of channel paths. 
\blue{
The geometric channel parameters (\ac{aoa}, \ac{aod}, and \ac{toa}) can be obtained either from external side information or from channel estimation.  For instance, given the \ac{ue} position from \ac{gnss}, the known \ac{bs} position, and an environmental map, ray tracing can identify the relevant propagation paths, including the associated \acp{sp}. Alternatively, these parameters can be directly estimated from pilot-based channel measurements.
In contrast, the complex channel gains are harder to acquire, as they depend on small-scale fading, reflection coefficients, blockage, and instantaneous propagation conditions. In reciprocal channels, however, the \ac{ue} may estimate the downlink gains. \blue{Under quasi-static geometry, the uplink path gains coincide with the estimated downlink gains up to a common global phase, which is irrelevant here because the localization estimator is phase-invariant and the spoofed gains $\bm{\lambda}$ are free design variables. This reciprocity is required only for the calibrated effective channel (e.g., via standard TDD reciprocity calibration, which is outside our scope) and only in the oracle regime; the blind regime does not assume it.} 
Therefore, in the oracle regime, we assume that the \ac{ue} has access to both the geometric parameters and an estimate of the complex \ac{csi}.}
\blue{The oracle and blind regimes are described next.}
\vspace{-10pt}

\begin{figure*}[!ht]
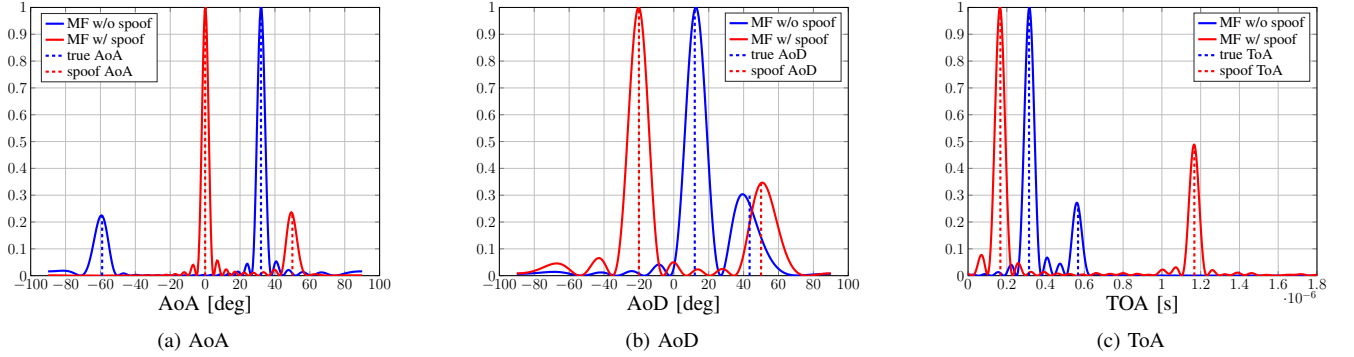

\centering
    \subfloat[AoA \label{fig:aoa_oracle}]{  
    \resizebox{0.33\linewidth}{!}{\input{figures/aoa_oracle.tex}}
    }
    \subfloat[AoD \label{fig:aod_oracle}]{    
    \resizebox{0.33\linewidth}{!}{\input{figures/aod_oracle.tex}}
    }
    \subfloat[ToA \label{fig:tdoa_oracle}]{       
    \resizebox{0.33\linewidth}{!}{\input{figures/tdoa_oracle.tex}}
    }
    \caption{Oracle spoofing: (a) MF AoA estimation with $N_{\mathrm{R}}^{}=M=24$, (b) MF AoD estimation with $N_{\mathrm{T}}^{}=S=8$, and (c) MF ToA estimation with $K=128$. The solid red and blue lines denote, respectively, the MF estimations with and without spoofing. The dashed lines represent the true (blue) and spoofing (red) values.}
    \label{fig:oracle_spoof}
    \vspace{-10pt}
\end{figure*}

\subsection{Oracle Spoofing}
\blue{Assuming known \ac{csi} at the transmitter}, the HoloTrace framework can induce the \ac{bs} to estimate arbitrary feasible positions within its coverage. \blue{We emphasize that HoloTrace is restricted to a per-resource-element (diagonal) pilot scaling under a fixed analog precoder, rather than a full linear precompensation of the channel, which would require digital precoder control. The oracle design below shows that this restricted, pilot-only actuation is already sufficient for exact spoofing when \ac{csi} is available. }%(a Hadamard, rather than matrix, inverse); the blind regime of Section~\ref{sec:blind} shows that, without \ac{csi}, the same actuation generally cannot achieve exact multipath spoofing.}
\blue{In the following, we adopt a bottom-up design approach: we first give a unified single-domain oracle model that covers \ac{aoa}, \ac{aod}, and \ac{toa} spoofing, and then develop the general HoloTrace attack under the oracle regime. The single-domain model drops unnecessary dimensions for exposition, while preserving the common projection structure of the three cases.}
% \begin{prop}
%     Given a model of the form 
%     \begin{equation}
%     \mathbf{y} = \operatorname{diag}(\mathbf{x})\, \mathbf{A}({\bm{\varpi}})\, \bm{\alpha} + \mathbf{q},
%     \end{equation}
%     where $\mathbf{A}({\bm{\varpi}})$ is a matrix dependent by a measurement vector $\bm{\varpi}$, then, a perfect spoofing signal $\widetilde{\mathbf{x}}$ is of the form 
%     \begin{equation}
%      \widetilde{\mathbf{x}} = \operatorname{diag}^{-1}\big( \mathbf{A}({\bm{\varpi}})\, \bm{\alpha} \big)\, \mathbf{A}(\overline{\bm{\varpi}})\, \bm{\lambda}.
%     \end{equation}
%     \textbf{Proof.} Appendix \ref{app:general_solution}
%     \label{prop:general_solution}
% \end{prop}

% \subsection{Oracle CSI Knowledge}
% Given perfect or partial channel gain knowledge, HoloTrace is able to locate the target position basically anywhere in the space. 
% \mb{I would not say anywhere. It is still limited by the coverage of the communication link}
% \li{make some comments over the channel gain knowledge, how to estimate it}
% In the following section, we will give a complete overview of the different scenarios considered. Fig. \ref{fig:oracle_spoof} gives illustrative examples of the simpler spoofing cases (i.e., \ac{aoa}, \ac{aod}, and \ac{tdoa}) using \ac{ml} estimation of the channel parameters.

\subsubsection[Single-Domain Oracle Spoofing]{\blue{Single-Domain Oracle Spoofing}}
\label{sec:oracle_aoa}
\blue{The \ac{aoa}, \ac{aod}, and \ac{toa} oracle spoofing problems share the same algebraic structure. To avoid repeating the same derivation three times, we introduce a single-domain model}
\begin{equation}
\blue{
    \mathbf{y}_{\mathrm{tm}}^{}
    =
    \operatorname{diag}(\mathbf{x})\,
    \mathbf{A}(\bm{\omega})\,
    \bm{\alpha}
    +
    \mathbf{q},
}
\label{eq:single_domain_oracle_model}
\end{equation}
\blue{where $\bm{\omega}$ denotes the parameter vector of one channel domain and $\mathbf{A}(\bm{\omega})$ is the corresponding steering matrix. The \ac{bs} estimates this parameter by minimizing $\widehat{\bm{\omega}}=\underset{\bm{\varpi}}{\operatorname{arg\,min}}\,\|\mathbf{y}_{\mathrm{tm}}^{}-\mathbf{A}(\bm{\varpi})\,\bm{\alpha}(\bm{\varpi})\|^2$, where $\bm{\alpha}(\bm{\varpi})$ is the conditional path-gain estimate for the candidate $\bm{\varpi}$.}

\begin{prop}[Single-domain oracle spoofing]
\label{prop:single_domain_oracle}
\label{prop:AOASpoof}
\blue{Assume that $\mathbf{A}(\bm{\omega})\bm{\alpha}$ has no zero entries. To make the noise-free received signal in \eqref{eq:single_domain_oracle_model} lie in the subspace associated with a target parameter vector $\overline{\bm{\omega}}$, the \ac{ue} can transmit}
\begin{equation}
\blue{
    \widetilde{\mathbf{x}}
    =
    \operatorname{diag}^{-1}
    \big(
    \mathbf{A}(\bm{\omega})\bm{\alpha}
    \big)
    \mathbf{A}(\overline{\bm{\omega}})
    \bm{\lambda},
}
\label{eq:single_domain_oracle_pilot}
\end{equation}
\blue{where $\bm{\lambda}\in\mathbb{C}^{L}$ is a spoofed path-gain vector that can also be chosen to meet the pilot-power constraint.}
\end{prop}
\begin{proof}
    \blue{See Appendix~\ref{app:single_domain_oracle}.}
\end{proof}

\blue{The three single-domain cases are obtained by selecting the appropriate steering matrix in \eqref{eq:single_domain_oracle_model}:}
\begin{align}
    \mathbf{y}_{\text{tm}}^{}
    & =
    \operatorname{diag}(\mathbf{x})\,
    \mathbf{B}(\bm{\theta})\,
    \bm{\alpha}
    +
    \mathbf{q},
    \quad
    \mathbf{B}(\bm{\theta})
    =
    \mathbf{W}^{\mathsf{H}}
    \mathbf{A}_{N_\mathrm{R}^{}}(\bm{\theta}),
\label{eq:simo_model}\\[-1mm]
    \mathbf{y}_{\text{tm}}^{}
    & =
    \operatorname{diag}(\mathbf{x})\,
    \mathbf{C}(\bm{\phi})\,
    \bm{\alpha}
    +
    \mathbf{q},
    \quad
    \mathbf{C}(\bm{\phi})
    =
    \mathbf{F}^{\mathsf{H}}
    \mathbf{A}_{N_{\mathrm{T}}^{}}(\bm{\phi}),
\label{eq:miso_model}\\[-1mm]
    \mathbf{y}_{\text{tm}}^{}
   &  =
    \operatorname{diag}(\mathbf{x})\,
    \mathbf{D}(\bm{\tau})\,
    \bm{\alpha}
    +
    \mathbf{q}.
\label{eq:toa_model}
\end{align}
\blue{Thus, \eqref{eq:single_domain_oracle_pilot} becomes an \ac{aoa} spoofing pilot when $\mathbf{A}=\mathbf{B}$ and $\bm{\omega}=\bm{\theta}$, an \ac{aod} spoofing pilot when $\mathbf{A}=\mathbf{C}$ and $\bm{\omega}=\bm{\phi}$, and a \ac{toa}-domain spoofing pilot when $\mathbf{A}=\mathbf{D}$ and $\bm{\omega}=\bm{\tau}$. In the last case, the directly manipulated delay estimates are later used through \ac{tdoa}-based single-anchor localization.}

Fig.~\ref{fig:aoa_oracle} shows the \ac{aoa} specialization in a \ac{simo} setting. The \ac{aoa} spectrum obtained via a \ac{mf} estimator\footnote{The \ac{mf} estimator does not presume knowledge of the number of paths and correlates $\mathbf{y}$ with \black{$\mathbf{b}(\theta)=\mathbf{W}^\mathsf{H} \mathbf{a}^{}_{N_\mathrm{R}^{}}(\theta)$.} } without spoofing (blue) displays peaks at the true angles (blue dashed). Under the spoofing pilot (red), the estimator is forced to yield a peak at the spoofed angle (red dashed), while the true peak is suppressed. \blue{Thus, in this illustrative oracle setting, pilot design can move the dominant \ac{aoa} estimate toward the prescribed spoofed value.}\footnote{{Note that these plots are for illustration only; more sophisticated estimators will be used in the results.}}

Fig.~\ref{fig:aod_oracle} demonstrates the \ac{aod} specialization in the \ac{miso} setting. Without attack (blue), the \ac{aod} spectrum obtained via a \ac{mf} estimator\footnote{\black{The \ac{mf} estimator correlates $\mathbf{y}$ with $\mathbf{c}(\phi)=\mathbf{F}^\mathsf{H} \mathbf{a}^{}_{N_\mathrm{T}^{}}(\phi)$.}} reveals peaks at the true departure angles (blue dashed). Under the spoofing pilot (red), the estimator output is shifted to the spoofed angles (red dashed), with the original peaks suppressed. \blue{This illustrates that the oracle pilot can redirect the \ac{aod} estimate toward the prescribed spoofed values under the matched-filter estimator.}

Fig.~\ref{fig:tdoa_oracle} illustrates the \ac{toa} specialization in the \ac{ofdm} \ac{siso} scenario. Without attack (blue), the delay spectrum exhibits peaks at the true \ac{toa} values (blue dashed). With the spoofing pilot (red), the estimator output is shifted to the prescribed spoofed delay (red dashed), with the true peaks suppressed\blue{, confirming that, for the illustrative \ac{siso} oracle case, pilot design can redirect the delay estimate toward the prescribed spoofed delay.}

\subsubsection{Joint AoA/AoD/TDoA Spoofing}
\label{sec:oracle}
Finally, we consider the full \ac{mimo}-\ac{ofdm} channel tensor model, where the parameters of interest, \ac{aoa}, \ac{aod}, and \ac{tdoa}, are embedded in the factor matrices $\mathbf{B}$, $\mathbf{C}$, and $\mathbf{D}$, respectively. \blue{The noise-free tensor $\mathcal{H}$ in~\eqref{eq:def_tensor_H} is unfolded along the frequency (third) mode and vectorized as}
\begin{equation}
    \mathbf{h} = \operatorname{vec}([\mathcal{H}]_{3}^{}) = \black{\sqrt{E_s^{}}}(\mathbf{B} \circledast \mathbf{C} \circledast \mathbf{D}) \bm{\alpha}.
\end{equation}

%In the absence of noise, and f
For received tensor $\mathcal{Y}_\text{tm}^{}$ in~\eqref{eq:def_ten_Y}, define $\mathbf{y}_\text{tm}^{} = \operatorname{vec}([\mathcal{Y}_\text{tm}^{}]_{3}^{})$. The  path gain estimate is
\begin{equation}
    \widehat{\bm{\alpha}} = \mathbf{Z}^\dagger \mathbf{y}_\text{tm}^{}, \qquad \blue{\mathbf{Z} = \sqrt{E_s^{}}(\mathbf{B} \circledast \mathbf{C} \circledast \mathbf{D})}.
\end{equation}

\begin{prop}[\blue{Oracle Multi-Domain HoloTrace Design}]
    To achieve perfect spoofing  of geometric parameters 
$\blue{\overline{\bm{\rho}} = [\overline{\bm{\rho}}_0^{} \; \cdots \; \overline{\bm{\rho}}_{{L-1}}^{}] \in \mathbb{R}^{3 \times {L}}}$, construct the pilot tensor as
\begin{align}
    \widetilde{\mathcal{X}}=\left(\mathcal{L}\times_1^{} \overline{\mathbf{B}}\times_2^{} \overline{\mathbf{C}}\times_3^{} \overline{\mathbf{D}} \right) \oslash \left(\mathcal{A}\times_1^{} \mathbf{B}\times_2^{} \mathbf{C}\times_3^{} \mathbf{D}\right),
\end{align}
where  $\mathcal{L}$ is a diagonal tensor with entries $\mathcal{L}^{\ell,\ell,\ell} = \lambda_\ell^{}$ (the spoofed path gains). The matrices 
$
\overline{\mathbf{B}} = \mathbf{B}(\bm{\overline\theta}), \overline{\mathbf{C}} = \mathbf{C}(\bm{\overline\phi}), \overline{\mathbf{D}} = \mathbf{D}(\bm{\overline\tau})
$ 
are constructed from the desired (spoofed) \ac{aoa}, \ac{aod}, and \ac{toa}, respectively.
\end{prop}
\begin{proof}
    \blue{Let $\mathcal{H}=\mathcal{A}\times_1^{}\mathbf{B}\times_2^{}\mathbf{C}\times_3^{}\mathbf{D}$ and $\overline{\mathcal{H}}=\mathcal{L}\times_1^{}\overline{\mathbf{B}}\times_2^{}\overline{\mathbf{C}}\times_3^{}\overline{\mathbf{D}}$ be the true and target spoofed noiseless tensors. If all entries of $\mathcal{H}$ are nonzero,} 
   % \begin{equation*}
    \blue{
        $\mathcal{H}\odot\widetilde{\mathcal{X}}
        =
        \mathcal{H}\odot
        \left(\overline{\mathcal{H}}\oslash\mathcal{H}\right)
        =
        \overline{\mathcal{H}}$,
    }
%    \end{equation*}
    \blue{so the noise-free received tensor coincides with the spoofed channel $\overline{\mathcal{H}}$ (parameters $\overline{\bm{\rho}}$, gains $\bm{\lambda}$), and the residual in \eqref{eq:spoofingCriterion3} is zero. The free gains $\bm{\lambda}$ normalize $\widetilde{\mathcal{X}}$ to the prescribed average transmit power.}
\end{proof}
\blue{\emph{Remark.} Unlike the nominal pilots, $\widetilde{\mathcal{X}}$ is in general not constant-modulus: where the true beamformed channel $\mathcal{H}$ is close to a beamforming null, the entries of $\overline{\mathcal{H}}\oslash\mathcal{H}$ become large, so under a fixed average-power budget, less energy reaches the useful resource elements. The resulting effective-\ac{snr} penalty is geometry-dependent.}
%\vspace{-5pt}

%%%%%%%%%%%%%%% Old Version %%%%%%%%%%%%%%%%%%%%%%%%
% Recalling \eqref{eq:channel}, \eqref{eq:rx_signal}, and \eqref{eq:tensorform}, the \ac{aoa}, \ac{aod}, and \ac{tdoa} information is obtained respectively from the matrices $\mathbf{B}$, $\mathbf{C}$, and $\mathbf{D}$.
% The channel tensor $\mathcal{H}$ as in~\eqref{eq:def_tensor_H} $=\mathcal{A}\times_1^{} \mathbf{B}\times_2^{} \mathbf{C}\times_3^{} \mathbf{D}$ can be unfolded and written in the following way:
% \begin{align}
%     \left[\mathcal{H}\right]_3^{} &= \mathbf{D} \bm{\Lambda} \left( \mathbf{B} \circledast \mathbf{C} \right)^\mathsf{T},\\
%     \text{vec}\left(\left[\mathcal{H}\right]_3^{}\right) &= \left(\mathbf{B} \circledast \mathbf{C} \circledast \mathbf{D}\right)\bm{\alpha}.
%     \label{eq:h3}
% \end{align}
% In noise-free conditions, defining $\mathbf{Z}= \mathbf{B} \circledast \mathbf{C} \circledast \mathbf{D}$ and $\mathbf{y}=\text{vec}\left(\left[\mathcal{Y}\right]_3^{}\right)$, it follows that:
% \begin{align}
%     \widehat{\bm{\alpha}}= \mathbf{Z}^\dagger \mathbf{y} .
% \end{align}
% Following the same procedure as the simpler case, the spoofing attack is done by crafting the signal:
% \begin{align}
%     \widetilde{\mathcal{X}}=\left(\mathcal{L}\times_1^{} \overline{\mathbf{B}}\times_2^{} \overline{\mathbf{C}}\times_3^{} \overline{\mathbf{D}} \right) \oslash \left(\mathcal{A}\times_1^{} \mathbf{B}\times_2^{} \mathbf{C}\times_3^{} \mathbf{D}\right),
% \end{align}
% with a diagonal core tensor, where $\mathcal{L}^{\ell, \ell, \ell}=\lambda_\ell^{}$.

\subsubsection[Operation under Imperfect CSI]{\blue{Operation under Imperfect CSI}}
\blue{For notational simplicity, we state the imperfect-\ac{csi} effect for the \ac{aoa}-only \ac{simo} specialization of Section~\ref{sec:oracle_aoa}. The same projection-residual interpretation applies to the \ac{aod}, \ac{toa}, and full \ac{mimo}-\ac{ofdm} cases by replacing the \ac{aoa} basis $\mathbf{B}(\bm{\theta})$ with $\mathbf{C}(\bm{\phi})$, $\mathbf{D}(\bm{\tau})$, or $\mathbf{B}(\bm{\theta})\circledast\mathbf{C}(\bm{\phi})\circledast\mathbf{D}(\bm{\tau})$, respectively.}

\begin{prop}[Oracle \ac{aoa} spoofing with imperfect \ac{csi}]
\label{prop:imperfect_csi_aoa}
\blue{Consider the \ac{simo} oracle design with effective true channel $\mathbf{h}=\mathbf{B}(\bm{\theta})\bm{\alpha}$ and target spoofed signal $\overline{\mathbf{h}}=\mathbf{B}(\overline{\bm{\theta}})\bm{\lambda}$. If the \ac{ue} only has the estimate $\widehat{\mathbf{h}}=\mathbf{h}+\delta\mathbf{h}$, and constructs the oracle pilot as $\widetilde{\mathbf{x}}=\operatorname{diag}^{-1}(\widehat{\mathbf{h}})\overline{\mathbf{h}}$, then the projection cost evaluated at the spoofed \ac{aoa} satisfies}
\begin{equation}
\blue{
    C(\overline{\bm{\theta}})
    =
    \left\|
    \mathbf{P}_{\mathbf{B}(\overline{\bm{\theta}})}^\perp
    (\delta\mathbf{y}+\mathbf{q})
    \right\|^2 ,
}
\label{eq:imperfect_csi_projected_cost}
\end{equation}
\blue{where $\delta\mathbf{y}=(\operatorname{diag}(\mathbf{h}\oslash\widehat{\mathbf{h}})-\mathbf{I})\overline{\mathbf{h}}$ is the signal perturbation induced by the \ac{csi} mismatch, and $\mathbf{q}$ is the receiver-noise vector in the effective \ac{simo} observation. For small relative channel errors, $\delta\mathbf{y}\approx-\operatorname{diag}(\overline{\mathbf{h}})(\delta\mathbf{h}\oslash\mathbf{h})$.}
\end{prop}
\begin{proof}
    \blue{See Appendix~\ref{app:imperfect_csi_aoa}.}
\end{proof}

\blue{
Proposition~\ref{prop:imperfect_csi_aoa} shows that imperfect \ac{csi} does not fundamentally change the oracle spoofing mechanism; rather, it introduces an additional residual term in the spoofed subspace test. 
%Proposition~\ref{prop:imperfect_csi_aoa} shows that imperfect \ac{csi} does not create a new spoofing mechanism; it creates an additional residual in the spoofed subspace test. 
Oracle spoofing, therefore, remains effective as long as this residual is small compared with the projected receiver noise. In practical terms, the required channel-gain accuracy becomes more stringent as the effective \ac{snr} increases.}
%Furthermore, the invariance to common-mode errors implies that only relative inaccuracies among multipath components degrade spoofing performance.}

%%%%%%%%%%%%%%%%%%%%%%%%%%%%%%%%%%% Alireza's version %%%%%%%%%%%%%%%%%%%%%%%%%%%%%%
\subsection{Blind Spoofing}
We consider a scenario in which the \ac{ue} has no knowledge of the complex path gains, i.e., neither the amplitudes nor the phases of the multipath components are available. \blue{We use the term \emph{blind} to refer specifically to this absence of channel-gain (\ac{csi}) knowledge; the geometric parameters $\bm{\rho}$ are still assumed available, as discussed in Section~\ref{sec:spoofing_formulation}.} In this regime, although pilot design for exact parameter spoofing is not possible \black{for $L>1$}, it is still feasible to manipulate the observable geometric parameters even without  $\bm{\alpha}$ knowledge. 

\black{As in the oracle case, we proceed in a bottom-up fashion, starting from the simpler spoofing scenarios and ending up in the joint attack design.}
For angular spoofing, we exploit the property that, under a single-path assumption, the angle estimators are invariant to the unknown path gain. Thus, the \ac{ue} can construct pilot sequences that force the estimator to return a prescribed angle, regardless of the true gain. In the case of delay spoofing, however, a simple temporal shift is generally ineffective in multipath settings, as the position estimator relies on the \ac{tdoa} due to the unknown clock bias. To overcome this, we introduce a \emph{fake path injection} strategy, in which the \ac{ue} embeds additional delay components in the pilot structure to alter the output of the delay estimator and bias the \ac{tdoa}. The explicit construction of these spoofing pilots is detailed in the subsequent sections. %Fig.~\ref{fig:blind_spoof} gives illustrative examples of the fundamental spoofing cases (i.e., \ac{aoa}, \ac{aod}, and \ac{tdoa}) using an \ac{mf}-based estimation of the channel parameters.

\begin{figure*}[!ht]
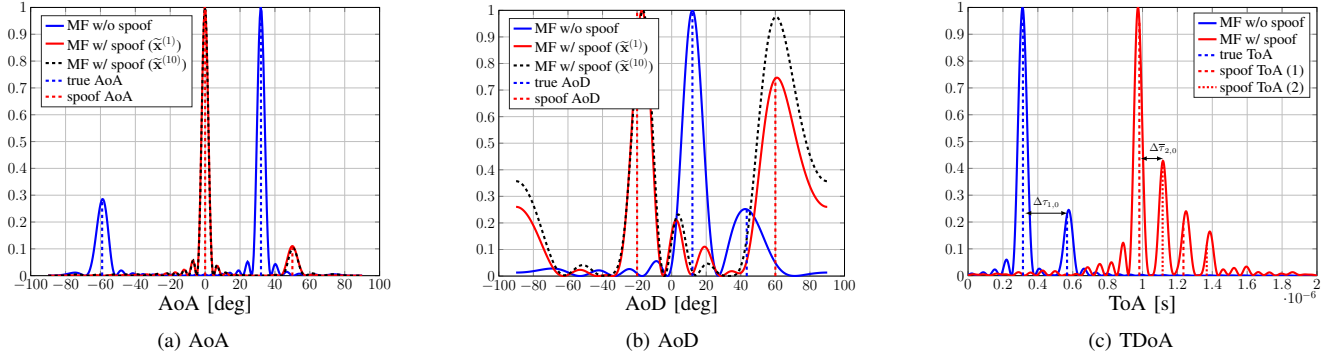

\centering
    \subfloat[AoA \label{fig:aoa_blind}]{  
    \resizebox{0.33\linewidth}{!}{\input{figures/aoa_blind.tex}}
    }
    \subfloat[AoD \label{fig:aod_blind}]{    
    \resizebox{0.33\linewidth}{!}{\input{figures/aod_blind.tex}}
    }
    \subfloat[TDoA \label{fig:tdoa_blind}]{       
    \resizebox{0.33\linewidth}{!}{\input{figures/tdoa_blind.tex}}
    }
    \caption{Blind spoofing: (a) AoA estimation with $N_{\mathrm{R}}^{}=M=24$, (b) AoD estimation with $N_{\mathrm{T}}^{}=S=8$, and (c) ToA estimation with \blue{$K=128$}.  \black{The solid blue lines show the MF estimations without spoofing. The red and black dashed lines correspond to the first and 10th iterations of the blind spoofing method, respectively. The dashed vertical lines indicate the true values (blue) and the spoofed values (red). The red dotted lines highlight the injected path pair in the \ac{tdoa} spoofing case.}
    % \caption{Blind spoofing: (a) AoA estimation with $N_{\mathrm{R}}^{}=M=24$, (b) AoD estimation with $N_{\mathrm{T}}^{}=S=8$, and (c) ToA estimation with $K=120$ kHz.  The red and blue lines are, respectively, the MF estimations with and without spoofing, the dashed lines represent the true (blue) and spoofing (red) values, \black{ and the black dashed lines denote the first and 10th iteration of the blind spoofing methodology, respectively,} and the red dotted lines are used for the injected path pair in \ac{tdoa} spoofing.%\li{TODO: add arrows for TDOAs}
   % \mb{check if it's what you mean} %\hw{perfect}
   %\mb{Consider adding the $\tau_i$ in the figure?}
    }
    \label{fig:blind_spoof}
    \vspace{-10pt}
\end{figure*}

%%%%%%%%%%%%%%%%%%%%%%%%%%%%%%%%%%% Old version %%%%%%%%%%%%%%%%%%%%%%%%%%%%%%
% \subsection{Blind LS Estimation}
% In the case where the \ac{ue} has no knowledge of the \ac{csi}, it is still possible to derive relative, heuristic-based perturbation strategies. For angular parameters, the \ac{ls} estimators can be constructed based on the single-path assumptions, under which the channel gain has no influence on the angle estimation process. Conversely, for the delay component, the single-path solution, which introduces only a temporal shift, is not applicable in multipath scenarios. To address this, we propose a \emph{fake path injection} to trick the channel parameter estimators and manipulate the inferred single anchor \ac{tdoa}. Fig. \ref{fig:blind_spoof} gives illustrative examples of the simpler spoofing cases (i.e., \ac{aoa}, \ac{aod}, and \ac{tdoa}) using \ac{ml} estimation of the channel parameters.

\subsubsection{AoA Spoofing} 

\black{In contrast to the case of perfect channel gain knowledge under oracle spoofing, perfect blind spoofing of the \ac{aoa} is generally not possible. The starting point is the fact that blind spoofing requires
\begin{align}
\mathcal{R}\!\big(\!\operatorname{diag}(\widetilde{\mathbf{x}})\mathbf B(\bm\theta)\big) \subseteq \mathcal{R}\!\big(\mathbf B(\overline{\bm\theta})\big).    \label{eq:blindAOAreq}
\end{align}

\begin{prop}
Blind perfect spoofing of the \ac{aoa} is generally possible when $L=1$ and when each element of $\mathbf{b}({\theta}_0^{})$ is non-zero, with %\begin{align}
    $\widetilde{\mathbf{x}} = \lambda \, \mathbf{b}(\overline{\theta}_0^{}) \oslash \mathbf{b}({\theta}_0^{})$, 
%\end{align}
for a power control parameter $\lambda \in \mathbb{C}$. 
When $L>1$, blind perfect spoofing is generally impossible, i.e., there exists no non-zero $\mathbf{x}$ that satisfies \eqref{eq:blindAOAreq} for arbitrary $\overline{\bm\theta}$. 
\end{prop}
\begin{proof}
    See Appendix \ref{app:blindAOA}.
    \vspace{-6pt}
\end{proof}
Note that the case $L=1$  %requires each element of $\mathbf{b}({\theta}_0^{})$ to be non-zero and 
was first derived and discussed in \cite{srinivasan2024aoa}.} 
\black{To allow for non-perfect blind spoofing, we can design $\widetilde{\mathbf{x}}$ to minimize the projector residual, i.e., 
\begin{align}
\min_{\widetilde{\mathbf{x}}}\ 
\big\|\big(\mathbf I_M^{} - \mathbf P_{\mathbf B(\overline{\bm\theta})}\big)\,\operatorname{diag}(\widetilde{\mathbf{x}})\,\mathbf B(\bm\theta)\big\|^2,
\label{eq:AoALSsurrogate}
\end{align}
subject to a constraint on $\Vert\mathbf{x}\Vert$. 
Problem \eqref{eq:AoALSsurrogate} is a %nonconvex (due to the unit-norm constraint on $\widetilde{\mathbf{x}}$) 
\ac{qcqp}. % and hard to solve globally. 
To obtain a solution with low complexity, we introduce an auxiliary mixing matrix $\bm{\Omega} \in \mathbb{C}^{L\times L}$, which makes the problem become bilinear. Specifically, we express \eqref{eq:AoALSsurrogate} as 
\begin{align}
\min_{\widetilde{\mathbf{x}}, \bm{\Omega}}\ 
\big\|\mathbf B(\overline{\bm\theta})\bm{\Omega}-\,\operatorname{diag}(\widetilde{\mathbf{x}})\,\mathbf B(\bm\theta)\big\|^2,
\label{eq:AoALSsurrogate2}
\end{align}
which can be solved for $\widetilde{\mathbf{x}}$ and $\bm{\Omega}$ in an alternating fashion. At iteration $t$, given $\bm{\Omega}^{(t)}$, $\widetilde{\mathbf{x}}^{(t)}$ is found by solving a row-wise \ac{ls} problem:  let $\mathbf b_{m}^\mathsf{T}$ denote the $m$-th row of $\mathbf B(\bm\theta)$ and $\mathbf b^{\mathsf{T}}_{\Omega,m}$  the $m$-th row of $\mathbf B(\overline{\bm\theta})\bm{\Omega}^{(t)}$, the update is $x_m^{(t)} = \lambda (\mathbf b^{\mathsf{T}}_{\Omega,m}\,\mathbf b_{m}^*) / \Vert \mathbf b_{m}^{}\Vert^2 $, where $\lambda >0$ ensures the power constraint. Given $\widetilde{\mathbf{x}}^{(t)}$, we solve \eqref{eq:AoALSsurrogate2} for $\bm{\Omega}$ as $\bm{\Omega}^{(t+1)} = \big(\mathbf B^{\mathsf{H}}(\overline{\bm\theta})\mathbf B(\overline{\bm\theta})\big)^{-1}\mathbf B^{\mathsf{H}}(\overline{\bm\theta})\operatorname{diag}(\widetilde{\mathbf{x}}^{(t)})\mathbf B(\bm\theta)$.
    %\begin{align*}
    %x_m^{(t+1)} = \frac{\mathbf b^{\mathsf{H}}_{s,m}\,\mathbf b_{m}^{}}{\mathbf b_{m}^\mathsf{H}\,\mathbf b_{m}^{}}, \qquad m=1,\dots,M.
    %\end{align*}
The entire procedure can be initialized by  $\bm{\Omega}^{(0)}=\mathbf{I}_{L}^{}$.}
\black{Although global optimality is not guaranteed, the iterations provide effective approximate spoofing with low complexity~\cite{ao-ref}.

Fig.~\ref{fig:aoa_blind} gives an illustrative example of \ac{aoa} spoofing for $L=2$, using an \ac{mf}-based estimation of the channel parameters. We observe that the angles can be spoofed, though we have no control over the relative path amplitudes. We also want to highlight that, after ten iterations of the proposed method (black dashed line), the spoofing performance does not improve compared to the first iteration.}

\subsubsection{AoD Spoofing}
\black{Due to the symmetry between the \ac{aoa} model \eqref{eq:simo_model} and the \ac{aod} model \eqref{eq:miso_model}, the same reasoning applies. 
Perfect spoofing is possible only for $L=1$, with pilot
%\begin{align}
$\widetilde{\mathbf{x}} = \lambda \,\mathbf c(\overline{\phi}_0)\oslash \mathbf c(\phi_0)$. 
%\end{align}
For $L\ge2$, blind perfect spoofing is generally impossible, and the attacker must instead resort to approximate spoofing, as in the \ac{aoa} case.
\blue{Fig.~\ref{fig:aod_blind}} gives an illustrative example of \ac{aod} spoofing for $L=2$, using an \ac{mf}-based estimation of the channel parameters. Again, we observe that the angles can be spoofed well, though we have no control over the relative path amplitudes.} 

\subsubsection{TDoA Spoofing}
{For blind \ac{tdoa} spoofing, perfect spoofing requires $\widetilde{\mathbf{x}}$
such that 
\begin{equation}
\text{diag}(\widetilde{\mathbf{x}})\mathbf{D}(\bm{\tau})\bm{\alpha}=\mathbf{D}(\overline{\bm{\tau}})\bm{\beta}
\end{equation}
for some $\bm{\beta}\in\mathbb{C}^{L}$. 
\begin{prop}
    Achieving a perfect spoofing on a target \ac{toa} vector $\overline{\bm{\tau}}$ in absence of knowledge $\bm{\alpha}$ has no solution, except when: (i) $L=1$, or (ii) $\overline{\bm{\tau}}=\bm{\tau}+\Delta^{}_\tau\mathbf{1}^{}_L$, with $\Delta^{}_\tau$ a delay shift. In both cases, the \ac{ue} designs pilot $\widetilde{\mathbf{x}}$ in \eqref{eq:toa_model} as
    %\begin{equation}
    $\widetilde{\mathbf{x}}=\lambda\,\mathbf{d}(\Delta^{}_\tau)$, %\label{eq:gen_sol_tdoa}
    %\end{equation}
    where $\lambda \in \mathbb{C}$ is a scalar design parameter.
    \label{prop:toa_spoof}
\end{prop}
\begin{proof}
    See Appendix \ref{app:tdoa_blind}.
\end{proof}
However, the spoofing
$\overline{\bm{\tau}}=\bm{\tau}+\Delta^{}_\tau\mathbf{1}^{}_L$ is not of interest,
since the solution to \eqref{eq:d_los} is invariant to the value of $\Delta^{}_\tau$.} To enable non-coherent spoofing of the \ac{tdoa} estimator in multipath, we propose a \emph{fake path injection} strategy {inspired by \cite{li2024channel}}. 

\paragraph{Fake Path Injection}
We focus first on the case $L=2$, i.e., the \ac{los} path and one \ac{nlos} path.  
The \ac{ue} synthesizes the transmitted pilot as a linear combination of multiple delayed steering vectors, generating a received signal equivalent to the superposition of four delays. Explicitly, we construct
\begin{align}
    \widetilde{\mathbf{x}} = \lambda_0^{} \mathbf{d}(\Delta_{\tau_0^{}}^{}) + \lambda_1^{} \mathbf{d}(\Delta_{\tau_2^{}}^{}),
\end{align}
where $\Delta_{\tau_0^{}}^{}=\overline{\tau}_0^{}-\tau_0^{}$ and $\Delta_{\tau_2^{}}^{}=\overline{\tau}_2^{}-\tau_0^{}$. This injected signal produces the equivalent of four observed delays at the receiver, namely $\overline{\tau}_{0}^{}$, $\overline{\tau}_{1}^{}=\overline{\tau}_{0}^{}+(\tau_1^{} - \tau_0^{})$, $\overline{\tau}_{2}^{}$, and $\overline{\tau}_{3}^{} = \overline{\tau}_{2}^{} + (\tau_1^{} -\tau_0^{})$, with respective complex channel gains $\alpha_0^{} \lambda_0^{}$, $\alpha_1^{} \lambda_0^{}$, $\alpha_0^{} \lambda_1^{}$, and $\alpha_1^{} \lambda_1^{}$. 
\black{To ensure that the \ac{tdoa} estimator infers the desired delay difference when the two smallest \acp{toa} are used for positioning, we require}
%Taking $\overline{\tau}_{0}^{}$ as a reference delay, the \acp{tdoa} will be $\tau_1^{} - \tau_0^{} $ (which is the correct and thus undesired value for spoofing), $\overline{\tau}_2^{} - \overline{\tau}_0^{}$ (which is the desired spoofed \ac{tdoa}) and $(\overline{\tau}_2^{} - \overline{\tau}_0^{}) + (\tau_1^{} - \tau_0^{})$. Hence, if the smallest \ac{tdoa} is used for positioning (see Section \ref{sec:op_model_chan_est}), full control of the \ac{tdoa} for positioning requires 
\begin{align}
    \black{\vert \Delta^{}_{\tau^{}_2}-\Delta^{}_{\tau^{}_0} \vert= \vert \overline{\tau}_2^{} - \overline{\tau}_0^{} \vert < \tau_1^{} - \tau_0^{}.}
\end{align}
Introducing $\Delta\overline{\tau}_{2,0}^{} = \vert \overline{\tau}_2^{} - \overline{\tau}_0^{} \vert$ and $\Delta\tau_{1,0}^{} = \tau_1^{} - \tau_0^{}$, this condition becomes 
   $\Delta\overline{\tau}_{2,0}^{} < \Delta\tau_{1,0}^{}$.

{Fig.~\ref{fig:tdoa_blind} illustrates the fake path injection approach with $L=2$, where a second pair of the original path (indicated with dotted red lines) is added in between the two paths of the first pair (indicated with dashed red lines), to modify the estimated \ac{tdoa} (shown with horizontal arrows).}

%\paragraph{Generalization to $L$ Paths}
%For the case of $L$ true paths, blind \ac{tdoa} spoofing via fake path injection requires synthesizing \black{$L^2_{}$} effective delays. The \ac{ue} transmits a pilot sequence of the form
% \begin{equation}\label{eq:tdoa_ls_spoof}
%     \widetilde{\mathbf{x}} = \sum_{i=0}^{L-1} \lambda_i^{}\, \mathbf{d}(\Delta_{\tau_i^{}}^{}), 
% \end{equation}
% where the sets $\{\Delta_{\tau_i^{}}^{}\}$ 
% define the $L$ injected delay differences and the coefficients $\{\lambda_i^{}\}$ are design variables for the synthetic path amplitudes (in particular, setting $\lambda_i^{}=0$ removes the fake path $i$). 
% The injected delays are assigned such that for each true path $i$ (with delay $\tau_i^{}$), there exist \black{$L$} corresponding fake delays:
% \begin{align}
%     \overline{\tau}_i^{(a)} &= \overline{\tau}_0^{(a)} + (\tau_i^{} - \tau_0^{}), \\
%     \overline{\tau}_i^{(b)} &= \overline{\tau}_0^{(b)} + (\tau_i^{} - \tau_0^{}),
% \end{align}
% for $i = 0,\ldots,L-1$, where $\overline{\tau}_0^{(a)}$ and $\overline{\tau}_0^{(b)}$ are anchor (reference) spoofed delays.
% By appropriately setting the differences
% \begin{align}
% \Delta\overline{\tau}_{i,j}^{} = \overline{\tau}_i^{(a)} - \overline{\tau}_j^{(b)},
% \end{align}
% the attacker can ensure that the estimator will select a desired set of delay pairs (i.e., a desired set of pairwise \ac{tdoa} values) among the \black{$L^2$} candidates, allowing for arbitrary manipulation of the \ac{tdoa}-based position estimate, as long as the estimator uses a shortest path or similar selection rule.

\paragraph{Generalization to $L$ Paths}
\blue{For the case of $L$ true paths, blind \ac{tdoa} spoofing via fake path injection can be described by choosing $L_s^{}$ injected delay shifts, with $L_s^{} \in \mathbb{N}$.}
\blue{The \ac{ue} transmits a pilot sequence of the form}
\begin{equation}\label{eq:tdoa_ls_spoof}
    \blue{\widetilde{\mathbf{x}} = \sum_{r=0}^{L_s^{}-1} \lambda_r^{}\, \mathbf{d}(\Delta_{\tau_r^{}}^{})},
\end{equation}
\blue{where $\{\Delta_{\tau_r^{}}^{}\}_{r=0}^{L_s^{}-1}$ are the injected delay shifts and $\{\lambda_r^{}\}_{r=0}^{L_s^{}-1}$ are design variables for the synthetic path amplitudes. Each true path $i\in\{0,\ldots,L-1\}$ then produces $L_s^{}$ candidate delays}
\begin{align}
    \blue{\overline{\tau}_{i,r}^{} = \tau_i^{}+\Delta_{\tau_r^{}}^{}, \qquad r=0,\ldots,L_s^{}-1,}
\end{align}
\blue{for a total of $L L_s^{}$ candidate delay peaks. Candidate \ac{tdoa} values are obtained from differences}
\begin{align}
\blue{
\Delta\overline{\tau}_{i,j}=
\overline{\tau}_{i,r}^{}-\overline{\tau}_{j,r'}^{}
=
(\tau_i^{}-\tau_j^{})+
(\Delta_{\tau_r^{}}^{}-\Delta_{\tau_{r'}^{}}^{}).
}
\end{align}
\blue{By choosing the injected shifts, the attacker can bias this candidate delay-pair set toward desired pairwise \ac{tdoa} values. This manipulation is conditional on the injected paths being detectable and on the estimator selecting the engineered pair, e.g., through a shortest-path or similar path-selection rule.} %; it should therefore be interpreted as a path-selection attack rather than an unconditional guarantee of arbitrary \ac{tdoa}-based position control.}
\blue{In summary, the \emph{fake path injection} generalizes as follows: \textit{(i)} choose $L_s^{}$ injected shifts, which generate $L L_s^{}$ candidate delay peaks from the $L$ true paths; \textit{(ii)} design the shifts so that the intended perceived \ac{tdoa} values are present in the estimator's candidate set and can be selected under the assumed path-selection rule; \textit{(iii)} select amplitudes $\{\lambda_r^{}\}_{r=0}^{L_s^{}-1}$ for energy normalization or further control.}

\subsubsection{Joint AoA/AoD Spoofing}
In the special case of a single path ($L=1$), we have:
\begin{align}
\mathbf{Y}_{\text{tm}}^{}  & =\alpha_0^{}\,\mathbf{W}^{\mathsf{H}}\mathbf{a}_{N_{\text{R}}^{}}(\theta_0^{}) \mathbf{a}_{N_{\text{T}}^{}}^{\mathsf{T}}(\phi_0^{})\mathbf{F}^\ast \odot \mathbf{X}+\mathbf{Q},
\end{align}
which can be vectorized to
\begin{align}
\mathbf{y}_{\text{tm}}^{}  & =\mathrm{vec}\!\left(\alpha_0^{}\,\mathbf{W}^{\mathsf{H}}\mathbf{a}_{N_{\text{R}}^{}}(\theta_0^{}) \mathbf{a}_{N_{\text{T}}^{}}^{\mathsf{T}}(\phi_0^{})\mathbf{F}^\ast \odot \mathbf{X}+\mathbf{Q}\right).
\end{align}
Let $\mathbf{b}=\mathbf{W}^{\mathsf{H}}\mathbf{a}_{N_{\text{R}}^{}}$, $\mathbf{c}=\mathbf{F}^{\mathsf{H}}\mathbf{a}_{N_{\text{T}}^{}}$, and $\mathbf{z}=\mathbf{c}\otimes\mathbf{b}$, it follows that %Assuming $\mathbf{X}=\mathbf{1}_{M\times S}$
\begin{align}
\mathbf{y}_{\text{tm}}^{}  & =\alpha_0^{}\,\mathbf{z}(\theta_0^{},\phi_0^{})\odot \mathrm{vec}\!\left(\mathbf{X}\right) +\mathbf{q}.
\end{align}
Given a modified pilot signal $\widetilde{\mathbf{X}}$, the cost function becomes
\begin{align}
C(\overline{\theta}_0^{},\overline{\phi}_0^{}; \widetilde{\mathbf{X}}) 
= \left\Vert  \mathbf{P}^{\perp}_{\mathbf{z}(\overline{\theta}_0^{},\overline{\phi}_0^{})}\, \big(\mathbf{z}(\theta_0^{},\phi_0^{}) \odot \mathrm{vec}(\widetilde{\mathbf{X}})\,\alpha_0^{} \big) \right\Vert^2.
\end{align}
\black{Perfect spoofing ($C=0$) is achieved if
\(\mathbf{z}(\theta_0^{},\phi_0^{}) \odot \mathrm{vec}(\widetilde{\mathbf{X}}) \in \mathcal{R}\left( \mathbf{z}(\overline{\theta}_0^{},\overline{\phi}_0^{}) \right)\),
which yields
\begin{equation}
\mathrm{vec}(\widetilde{\mathbf{X}}) = \lambda \, \mathbf{z}(\overline{\theta}_0^{},\overline{\phi}_0^{})\oslash\mathbf{z}(\theta_0^{},\phi_0^{}),
\end{equation}
\blue{provided that $\mathbf{z}(\theta_0^{},\phi_0^{})$ has nonzero entries.}
%\medskip
%\noindent\textbf{Remark:} 
For $L \ge 2$ paths with unknown gains, perfect spoofing condition would require subspace inclusion $\mathcal R(\operatorname{diag}(\mathrm{vec}(\widetilde{\mathbf{X}}))\,\mathbf Z(\bm\theta,\bm\phi)) \subseteq \mathcal R(\mathbf Z(\overline{\bm\theta},\overline{\bm\phi}))$\blue{, with $\mathbf{Z}(\bm\vartheta, \bm\varphi) = \mathbf{B}(\bm\vartheta) \circledast \mathbf{C}(\bm\varphi)$}, which is in general not achievable, as a single diagonal scaling cannot satisfy the subspace relation for multiple independent paths. %which is generically impossible with a single diagonal scaling. 
Approximate spoofing, as in the \ac{aoa} case, can be applied. \blue{For instance, setting $\bm{\Omega} = \mathbf I_{L}^{}$, yields} %, since allowing $\mathbf S$ provides marginal performance compared to the added complexity, and optimize over 
%the optimal $\widetilde{\mathbf{X}}$ is found as 
% \begin{equation}
% \operatorname{vec}\big(\widetilde{\mathbf{X}}\big)
% = \big(\mathbf z(\overline{\bm\theta}, \overline{\bm\phi}) \odot \mathbf z^\ast(\bm\theta,\bm\phi)\big)
%    \oslash |\mathbf z(\bm\theta,\bm\phi)|^2,
% \end{equation}
\begin{equation}
\widetilde{\mathbf{X}}
= (\mathbf B(\overline{\bm\theta}) \bm\Gamma \mathbf C^\mathsf{T}(\overline{\bm\phi}))
   \oslash (\mathbf B(\bm\theta) \mathbf C^\mathsf{T}(\bm\phi)),
   \label{eq:blind_spatial}
\end{equation}
with $\bm\Gamma=\operatorname{diag}(\bm\lambda)$. This value of $\widetilde{\mathbf{X}}$ can then be used to refine $\bm{\Omega}$ as before.}

\subsubsection{Joint AoA/AoD/TDoA Spoofing}
\label{sec:blind}
From the tensor channel model~\eqref{eq:def_tensor_H}, the vectorized received signal is
\begin{align}
    \mathbf{y}_{\text{tm}} &=\left(\mathbf{B} \circledast \mathbf{C} \circledast \mathbf{D}\right)\bm{\alpha} \odot \mathrm{vec}\!\left(\mathbf{X}\right) \nonumber \\
    &=\sum^{L-1}_{\ell=0}\alpha_\ell^{}\left(\mathbf{b}_\ell^{} \otimes \mathbf{c}_\ell^{} \otimes \mathbf{d}_\ell^{}\right) \odot \mathrm{vec}\!\left(\mathbf{X}\right),
\end{align}
where $\mathbf{b}_\ell^{}$, $\mathbf{c}_\ell^{}$, $\mathbf{d}_\ell^{}$ are the $\ell$-th columns of $\mathbf{B}$, $\mathbf{C}$, $\mathbf{D}$, respectively. 

\begin{prop}[\blue{Blind Multi-Domain HoloTrace Design}]
\blue{Consider a \ac{mimo}-\ac{ofdm} system} with $L$ propagation paths. The received signal is modeled as
\begin{equation}
   \mathbf{y}_{\textnormal{tm}}^{}  = \sum_{\ell=0}^{L-1} \alpha_\ell^{} \left( \mathbf{b}_\ell^{} \otimes \mathbf{c}_\ell^{} \otimes \mathbf{d}_\ell^{} \right) \odot \operatorname{vec}(\mathbf{X}). 
\end{equation}
If the spoofed pilot is constructed as\footnote{\blue{Note that an alternative solution is 
$
\mathrm{vec}(\widetilde{\mathbf{X}}) = \mathrm{vec}(\widetilde{\mathbf{X}}_\mathrm{bc}) \otimes  \widetilde{\mathbf{x}}_\mathrm{d},
$
which decouples the spatial (with $\widetilde{\mathbf{X}}_\mathrm{bc}$ from \eqref{eq:blind_spatial}) and time domain manipulation.}}
\begin{equation}
\operatorname{vec}(\widetilde{\mathbf{X}}) =  \widetilde{\mathbf{x}}_{\mathrm{b}} \otimes \widetilde{\mathbf{x}}_{\mathrm{c}} \otimes \widetilde{\mathbf{x}}_{\mathrm{d}},
\end{equation}
with spoofing vectors designed independently to manipulate \ac{aoa}, \ac{aod}, and delay dimensions, then the spoofed signal is
\begin{equation}
\mathbf{y}_{\textnormal{tm}}^{}  = \sum_{\ell=0}^{L-1} \alpha_\ell^{} 
\left( \mathbf{b}_\ell^{} \odot \widetilde{\mathbf{x}}_{\mathrm{b}} \right) \otimes 
\left( \mathbf{c}_\ell^{} \odot \widetilde{\mathbf{x}}_{\mathrm{c}} \right) \otimes 
\left( \mathbf{d}_\ell^{} \odot \widetilde{\mathbf{x}}_{\mathrm{d}} \right).
\end{equation}
\end{prop}

\begin{proof}

The received signal under pilot \( \widetilde{\mathbf{X}} \) is:
\begin{align}
\mathbf{y}_{\text{tm}}^{}  = \sum_{\ell=0}^{L-1} \alpha_\ell \left( \mathbf{b}_\ell \otimes \mathbf{c}_\ell \otimes \mathbf{d}_\ell \right) \odot \operatorname{vec}(\widetilde{\mathbf{X}}).
\end{align}
Substituting the spoofed pilot $\operatorname{vec}(\widetilde{\mathbf{X}})$ 
and using the identity
\[
(\mathbf{a}_1^{} \otimes \mathbf{a}_2^{} \otimes \mathbf{a}_3^{}) \odot (\mathbf{b}_1^{} \otimes \mathbf{b}_2^{} \otimes \mathbf{b}_3^{})\!=\! (\mathbf{a}_1^{} \odot \mathbf{b}_1^{}) \otimes (\mathbf{a}_2^{} \odot \mathbf{b}_2^{}) \otimes (\mathbf{a}_3^{} \odot \mathbf{b}_3^{}),
\]
we obtain:
\[
\mathbf{y}_{\text{tm}}^{}  = \sum_{\ell=0}^{L-1} \alpha_\ell^{} 
\left( \mathbf{b}_\ell^{} \odot \widetilde{\mathbf{x}}_{\mathrm{b}}^{} \right) \otimes 
\left( \mathbf{c}_\ell^{} \odot \widetilde{\mathbf{x}}_{\mathrm{c}}^{} \right) \otimes 
\left( \mathbf{d}_\ell^{} \odot \widetilde{\mathbf{x}}_{\mathrm{d}}^{} \right).
\]
This shows that spoofing acts independently along each tensor dimension, and the resulting signal preserves the separable structure. \blue{The remaining claims on lack of association control and potential path multiplication follow from the construction and prior analysis directly.}
 %Substituting, the received signal becomes
%\begin{align}
 %   \mathbf{y} = \sum_{\ell=0}^{L-1} \alpha_\ell^{} 
  %  \big(\mathbf{b}_\ell^{} \odot \widetilde{\mathbf{x}}_{\mathrm{b}}^{}\big) \otimes 
   % \big(\mathbf{c}_\ell^{} \odot \widetilde{\mathbf{x}}_{\mathrm{c}}^{}\big) \otimes 
    %\big(\mathbf{d}_\ell^{} \odot \widetilde{\mathbf{x}}_{\mathrm{d}}^{}\big).
%\end{align}
 %   The key insight is that the spoofing signal for each dimension can be designed independently, leveraging the solutions developed for angular and delay spoofing in the previous sections. Specifically, we set
%\textit{(i)} $\widetilde{\mathbf{x}}_{\mathrm{b}}^{} $ according to the \ac{aoa} spoofing solution \eqref{eq:AOAspoofblind}, \textit{(ii)} $\widetilde{\mathbf{x}}_{\mathrm{c}}^{} $ according to the \ac{aod} spoofing solution \eqref{eq:AODspoofblind}, and \textit{(iii)} $\widetilde{\mathbf{x}}_{\mathrm{d}}^{} $ according to the \ac{tdoa} fake-path injection solution  \eqref{eq:tdoa_ls_spoof}. Finally, the composite spoofing pilot is the Kronecker (outer) product of the three dimension-specific signals, i.e., $\operatorname{vec}(\widetilde{\mathbf{X}}) = \widetilde{\mathbf{x}}_{\mathrm{b}} \otimes \widetilde{\mathbf{x}}_{\mathrm{c}} \otimes \widetilde{\mathbf{x}}_{\mathrm{d}}$.\footnote{Note that an alternative solution is 
%$
%\mathrm{vec}(\widetilde{\mathbf{X}}) = \mathrm{vec}(\widetilde{\mathbf{X}}_\mathrm{bc}) \otimes  \widetilde{\mathbf{x}}_\mathrm{d},
%$
%which decouples the spatial and time domain manipulation. }
\end{proof}
\blue{Proposition~6 should be interpreted as a separable blind pilot construction rather than as a guarantee of exact geometric spoofing in multipath. The Kronecker structure allows the \ac{ue} to perturb the \ac{aoa}, \ac{aod}, and delay dimensions independently and simultaneously without path-gain knowledge. However, in multipath channels the resulting estimator output may still suffer from path association errors, power imbalance, or injected-path selection effects; consequently, the construction provides approximate spoofing or obfuscation unless the induced paths are associated with the intended geometry.}

\definecolor{color_spoof}{rgb}{0.85098,0.32549,0.09804}
\definecolor{color_aob}{rgb}{0.49412,0.18431,0.55686}
\definecolor{color_bht}{rgb}{0.46667,0.67451,0.18824}
\definecolor{color_dais}{rgb}{0.92941,0.69412,0.12549}
\definecolor{color_clean}{rgb}{0.00000,0.44700,0.74100}

\section{Performance Evaluation}
\label{sec:results}
\subsection{Simulation Setup}
Using the \ac{bs} as a reference point for our local coordinates, we placed the \ac{bs} in 
\(\mathbf{p}_\mathrm{BS}^{}=\left[0 \; 0 \right]^\mathsf{T}\mathrm{m}\) with \(o_\mathrm{BS}^{}=0 \,\mathrm{rad}\), 
the \ac{ue} in \(\mathbf{p}_\mathrm{UE}^{}=\left[10\; 5\right]^\mathsf{T}\mathrm{m}\) with \(o_\mathrm{UE}^{}=\blue{\text{-}\frac{2}{3}\pi\,\mathrm{rad}}\), 
and one \ac{sp} 
in \(\mathbf{p}_1^{}=\left[7\;\text{-}15\right]^\mathsf{T}\mathrm{m}\). 
The channel in \eqref{eq:channel} and the received signal in \eqref{eq:rx_signal} are modeled with the following \black{3GPP standard parameters \cite{tr138901}}: $N_{\mathrm{R}}^{}=24$, $N_{\mathrm{T}}^{}=16$, $L=2$, $M=24$, $S=16$, $K=3300$, $\Delta f=120$~kHz, with a bandwidth of $\mathrm{BW}=K\Delta f=396$~MHz, and $f_c^{}=27.8$~GHz. With the given positions, the true measurements are $\bm{\theta}=\left[0.46\; \text{-}1.13\right]^\mathsf{T}\mathrm{rad}$, $\bm{\phi}=\left[\text{-}0.58\; 0.37\right]^\mathsf{T}\mathrm{rad}$, and $\bm{\tau}=c^{-1}\left[11.18\; 36.78\right]^\mathsf{T}\mathrm{s}$.
The spoofed position is chosen at \(\overline{\mathbf{p}}_\mathrm{UE}^{}=\left[30\; \text{-}20\right]^\mathsf{T}\mathrm{m}\) with \(\overline{o}_\mathrm{UE}^{}=\pi\,\mathrm{rad}\), and the \ac{sp} in \(\overline{\mathbf{p}}_1^{}=\left[40\; \text{-}10\right]^\mathsf{T}\mathrm{m}\). 
The respective spoofing measurements are $\overline{\bm{\theta}}=\left[\text{-}0.59\; \text{-}0.24\right]^\mathsf{T}\mathrm{rad}$, $\overline{\bm{\phi}}=\left[0.98\; \text{-}0.79\right]^\mathsf{T}\mathrm{rad}$, and $\overline{\bm{\tau}}=c^{-1}\left[36.06\; 55.37\right]^\mathsf{T}\mathrm{s}$.

The channel path gain is given by 
%\begin{equation}
% $\alpha_\ell^{} = \sqrt{\eta_\ell^{} \,P_t^{} G_{\mathrm{BS}}^{} G_{\mathrm{UE}}^{} }\, e^{j\omega_\ell^{}}$, 
%\end{equation}
%where $P_{t}^{}$ is the transmitted power, 
$\alpha_\ell^{} = \sqrt{\eta_\ell^{} G_{\mathrm{BS}}^{} G_{\mathrm{UE}}^{} }\, e^{j\omega_\ell^{}}$, 
where $G_{\mathrm{BS}}^{}$ and $G_{\mathrm{UE}}^{}$ are respectively the \ac{bs} and \ac{ue} antenna gains in linear scale, assuming  $G^{\mathrm{dBi}}_{\mathrm{BS}} = 7$  dBi and $G^{\mathrm{dBi}}_{\mathrm{UE}}=3$ dBi, $\omega_\ell^{}$ is the random phase, and $\eta_\ell^{}$ is the path-dependent power attenuation term defined as $\eta_0^{} = ({c}/{(4\pi f_c^{} d_0)})^2$ and $\eta_1=\sigma_{\mathrm{RCS}}^{} \, c^2 / ((4\pi)^3 f_c^{2} \, \| \mathbf{p}_{\mathrm{UE}}^{}-\mathbf{p}_{1}^{}\|^2 \, \| \mathbf{p}_{1}^{}-\mathbf{p}_{\mathrm{BS}}^{}\|^2)$, 
%\begin{align}
 %   \eta_0^{} &= \left(\frac{c}{4\pi f_c^{} d_0^{}}\right)^2,\\ 
  %  \eta_1^{} &= \frac{\sigma_{\mathrm{RCS}}^{} \, c^2 }{(4\pi)^3 f_c^{2} \, \| \mathbf{p}_{\mathrm{UE}}^{}-\mathbf{p}_{1}^{}\|^2 \, \| \mathbf{p}_1^{}-\mathbf{p}_{\mathrm{BS}}^{}\|^2},
%\end{align}
with $\sigma_{\mathrm{RCS}}^{} = 50\,\mathrm{m}^2$ the reflector \ac{rcs} \cite{zhang2025unified}. 

\subsection{Channel Parameter Estimation}
Channel estimation is performed using a fully automated low-complexity algorithm called FLEX~\cite{pourafzal2025flex}, which jointly recovers the tuple $(\tau_\ell^{}, \theta_\ell^{}, \phi_\ell^{})$ for all paths without requiring knowledge of $L$ or explicit parameter association. 

The FLEX algorithm first unfolds $\mathcal{Y}$ along the frequency dimension (mode-3), computes the inverse \ac{fft} per snapshot, and forms the Bartlett periodogram
%\begin{equation}
    $P(k) = \frac{1}{SM} \sum_{m=1}^{SM} \left| \mathscr{F}^{-1}([\mathcal{Y}]_3^{:,m})[k] \right|^2$, 
%\end{equation}
where $[\mathcal{Y}]_3^{} \in \mathbb{C}^{K \times SM}$, $SM$ is the total number of spatial snapshots, and $\mathscr{F}^{-1}(\cdot)$ is the inverse \ac{fft}. Significant delay bins are identified via a \ac{cfar} detector, yielding $\widehat{L}$ active paths. For each detected delay $\widehat{\tau}_\ell^{}$, the corresponding frequency component is phase-compensated, and \ac{aoa}/\ac{aod} estimation is performed by extracting tensor slices and applying ESPRIT-based angle recovery. Formally, the FLEX estimator yields $ \widehat{\boldsymbol{\rho}} = \mathrm{FLEX}(\mathcal{Y})$.

\subsection{Performance Metrics}
We define two location error metrics: 
\begin{itemize}
    \item \textbf{Positioning error}: 
    $\varepsilon_{\mathrm{est}}^{} = \left\Vert \widehat{\mathbf{p}}_\mathrm{UE}^{} - \mathbf{p}_\mathrm{UE}^{} \right\Vert$, 
    representing the error between the estimated position and the true position;

    \item \textbf{Spoofing deviation}: 
    $\varepsilon_{\mathrm{dev}}^{} = \left\Vert \widehat{\mathbf{p}}_\mathrm{UE}^{} - \overline{\mathbf{p}}_\mathrm{UE}^{} \right\Vert$, 
    representing the deviation between the estimated position and the {desired} spoofed position.
\end{itemize}
All results are expressed as \ac{rmse} values, computed over $T = 250$ independent Monte Carlo simulations. 
The \ac{rmse} is defined, for a generic error $\varepsilon$, as $\mathrm{RMSE}=\sqrt{{1}/{T}\sum^{T}_{t=1}\varepsilon_t^2}.$

\subsection{Simulation Analyses and Results}
% We conducted a comprehensive evaluation of the system across multiple dimensions. First, all localization and spoofing metrics are assessed as a function of the transmit power, in order to analyze the impact of the \ac{snr} on the performance. These results are also compared against the \ac{dais} method proposed in~\cite{li2025delay}. \black{Second, we assess and discuss the measurement deviation, i.e., the error between the estimated and the desired measurement, of our solutions.}
% Third, we examine the communication performance to verify that the spoofing strategy does not significantly degrade the quality of the communication link. Finally, we analyze the robustness of the spoofing attack by introducing uncertainty in the system’s assumptions, i.e., inaccurate \ac{gnss} positioning and \ac{hd} map alignment.

\blue{We evaluate the system along four dimensions. First, we analyze the communication rate to quantify how the selected spoofed geometry affects the link. Next, we evaluate localization and spoofing metrics as a function of transmit power, to study the impact of \ac{snr}. We then assess the measurement deviation, i.e., the error between the estimated and desired channel parameters, for the proposed solutions. Finally, we study sensitivity to imperfect \ac{gnss}-like position knowledge and imperfect \ac{csi}.}

% First, all localization and spoofing metrics are assessed as a function of the transmit power, in order to analyze the impact of the \ac{snr} on the performance. These results are also compared against the \ac{dais} method proposed in~\cite{li2025delay}. \black{Second, we assess and discuss the measurement deviation, i.e., the error between the estimated and the desired measurement, of our solutions.}
% Third, we examine the communication performance to verify that the spoofing strategy does not significantly degrade the quality of the communication link. Finally, we analyze the robustness of the spoofing attack by introducing uncertainty in the system’s assumptions, i.e., inaccurate \ac{gnss} positioning and \ac{hd} map alignment.

\begin{figure}[t]
    \centering    
    \resizebox{1.1\columnwidth}{!}{

\begin{tikzpicture}

\begin{axis}[%
width=3.566in,
height=2.07in,
at={(-0.1in,0.551in)},
scale only axis,
point meta min=0,
point meta max=4179020345.56836,
axis on top,
ymin=0,
ymax=50.5,
xlabel style={font=\large\color{white!15!black}},
xlabel={y [m]},
xmin=-43.5,
xmax=43.5,
ylabel style={font=\large\color{white!15!black}, yshift=-10pt},
ylabel={x [m]},
axis background/.style={fill=black},
axis x line*=bottom,
axis y line*=left,
colormap/jet,
colorbar,
colorbar style={ylabel style={font=\large\color{white!15!black},yshift=-70pt, rotate=180}, 
ylabel={R [bps]}},
legend style={legend pos=south west, legend cell align=left, align=left, draw=white!15!black, font=\large},
ticklabel style={font=\large}
]
\addplot [forget plot] graphics [ymin=0.5, ymax=50.5, xmin=-43.5, xmax=43.5] {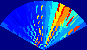};%] {figures/channel_rate_scan-1.png};

\addplot [color=black, only marks, mark=otimes, line width= 1.2, mark size=4, mark options={solid, black}]
  table[row sep=crcr]{%
5	10\\
};
\addlegendentry{true UE}

\addplot [color=white, only marks, mark=oplus, line width= 1.2, mark size=4, mark options={solid, lightgray}]
  table[row sep=crcr]{%
-20	30\\
};
\addlegendentry{location \#1}

\addplot [mark color= none, only marks, mark=halfcircle, line width= 1.2, mark size=4, mark options={solid, teal}]
  table[row sep=crcr]{%
19	30\\
};
\addlegendentry{location \#2}

\addplot [color=red, only marks, mark=triangle, line width= 3.5, mark options={solid, red}]
  table[row sep=crcr]{%
0	0\\
};
\addlegendentry{BS}

\addplot [color=blue, forget plot]
  table[row sep=crcr]{%
0	25\\
0	43\\
};
\addplot [color=blue, forget plot]
  table[row sep=crcr]{%
0	25\\
0	-43\\
};
\end{axis}

\begin{axis}[%
width=4.725in,
height=2.375in,
at={(0in,0in)},
scale only axis,
point meta min=0,
point meta max=1,
xmin=0,
xmax=1,
ymin=0,
ymax=1,
axis line style={draw=none},
ticks=none,
axis x line*=bottom,
axis y line*=left
]
\end{axis}
\end{tikzpicture}%
    }
    \caption{\black{Channel rate evaluation heatmap for each spoofing position for a BS sector of [-60, 60] deg and radius of 50 m. The red triangle represents the BS, the black cross the true UE position, the gray plus marker the UE-spoofed position, and the teal minus marker the target UE-spoofed position optimized for the channel rate.}
    }

    \label{fig:scanner}
    %\vspace{-10pt}
\end{figure}

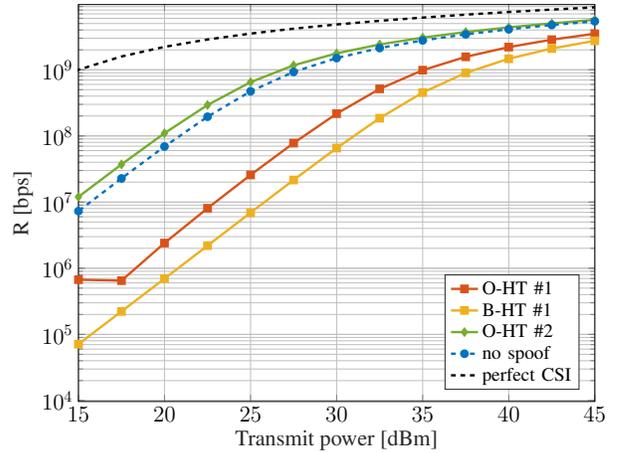
\begin{figure}[t]
    \centering
    \resizebox{\columnwidth}{!}{% This file was created by matlab2tikz.
\definecolor{mycolor1}{rgb}{0.00000,0.44700,0.74100}%
\definecolor{mycolor2}{rgb}{0.85098,0.32549,0.09804}%
\definecolor{mycolor3}{rgb}{0.92941,0.69412,0.12549}%
\definecolor{mycolor4}{rgb}{0.46667,0.67451,0.18824}%
\begin{tikzpicture}

\begin{axis}[%
width=4.521in,
height=3.3in,
at={(0.758in,0.568in)},
scale only axis,
xmin=15,
xmax=45,
xtick={15,20,25,30,35,40,45},
xlabel style={font=\Large\color{white!15!black}},
xlabel={Transmit power [dBm]},
ymode=log,
ymin=9674.43445164219,
ymax=9674434451.64219,
yminorticks=true,
ylabel style={font=\Large\color{white!15!black}},
ylabel={R [bps]},
axis background/.style={fill=white},
xmajorgrids,
ymajorgrids,
yminorgrids,
legend style={legend pos=south east, legend cell align=left, align=left, draw=white!15!black, font=\large},
ticklabel style={font=\Large}
]

\addplot [color=mycolor2, mark=square*, mark options={solid, mycolor2}, line width=1.5pt]
  table[row sep=crcr]{%
5	154146.021498974\\
7.5	989759.90637535\\
10	1879968.15513533\\
12.5	3478808.09088439\\
15	671247.896372461\\
17.5	649421.725543074\\
20	2402714.90143547\\
22.5	8069471.02748524\\
25	25735961.6450319\\
27.5	77921679.8980063\\
30	216958458.490187\\
32.5	514363049.078129\\
35	985940213.11814\\
37.5	1569439467.25011\\
40	2201512456.99658\\
42.5	2850894947.12024\\
45	3505999348.42167\\
};
\addlegendentry{O-HT, loc. \#1}

\addplot [color=mycolor3, mark=square*, mark options={solid, mycolor3},line width=1.5pt]
  table[row sep=crcr]{%
5	187027.527299095\\
7.5	817886.946252148\\
10	7303.70090977583\\
12.5	22915.3773369627\\
15	71188.9100223297\\
17.5	222318.877940911\\
20	696982.116836536\\
22.5	2193928.81517888\\
25	6909967.48025394\\
27.5	21564269.2443146\\
30	65550129.2423061\\
32.5	185510130.646467\\
35	453003320.772004\\
37.5	898042677.006522\\
40	1467624525.12054\\
42.5	2094039342.85535\\
45	2741724776.42338\\
};
\addlegendentry{B-HT, loc. \#1}

\addplot [color=mycolor4, mark=diamond*, mark options={solid, mycolor4},line width=1.5pt]
  table[row sep=crcr]{%
5	121209.679434079\\
7.5	383221.476751884\\
10	1210963.35367914\\
12.5	3820555.05447564\\
15	11995861.4569422\\
17.5	37103590.7399018\\
20	109935334.834607\\
22.5	293316756.863954\\
25	650390730.827301\\
27.5	1166890777.59815\\
30	1771554861.56401\\
32.5	2411453020.00323\\
35	3063435252.92672\\
37.5	3719338043.17832\\
40	4376491671.84914\\
42.5	5034019818.24341\\
45	5691738445.50601\\
};
\addlegendentry{O-HT, loc. \#2}

\addplot [color=mycolor1, dashed, mark=*, mark options={solid, mycolor1},line width=1.5pt]
  table[row sep=crcr]{%
5	73497.7327379041\\
7.5	232405.697009836\\
10	734545.209668339\\
12.5	2319694.66919237\\
15	7303546.64172356\\
17.5	22784169.1009901\\
20	69145251.50072\\
22.5	195009508.941038\\
25	472458697.404595\\
27.5	927187233.506437\\
30	1502139116.30141\\
32.5	2130930090.19097\\
35	2779257687.51049\\
37.5	3433987309.92053\\
40	4090723378.16567\\
42.5	4748232403.57451\\
45	5405834674.56553\\
};
\addlegendentry{no spoof}

\addplot [color=black, dashed, line width=1.5pt]
  table[row sep=crcr]{%
5	26619115.8850052\\
7.5	80258026.179073\\
10	222801233.948967\\
12.5	525371705.220489\\
15	1001273923.13108\\
17.5	1586947497.70439\\
20	2219868705.45973\\
22.5	2869530825.81285\\
25	3524693776.56665\\
27.5	4181618278.29202\\
30	4839100472.625\\
32.5	5496761061.09269\\
35	6154476179.75122\\
37.5	6812209925.8139\\
40	7469949347.22948\\
42.5	8127690021.67943\\
45	8785432296.32475\\
};
\addlegendentry{perfect CSI}
\end{axis}

\begin{axis}[%
width=5.833in,
height=4.375in,
at={(0in,0in)},
scale only axis,
xmin=0,
xmax=1,
ymin=0,
ymax=1,
axis line style={draw=none},
ticks=none,
axis x line*=bottom,
axis y line*=left
]
\end{axis}
\end{tikzpicture}%
}
    \caption{\black{Channel rate performance comparison.}}
    \label{fig:communication}
\end{figure}

\subsubsection{Channel Rate vs Transmit Power} 
\black{To assess how the communication link quality is influenced by the selected spoofing position, we analyze the channel rate $R$, as defined in \eqref{eq:comm_rate}, over the \ac{bs} coverage area within a radius of 50\,m and angular sector [-60, 60]\,deg. Fig.\,\ref{fig:scanner} illustrates the channel rate heatmap with $P_t^{}=35$\,dBm, where each square represents the candidate spoofing position. The \ac{sp} is always selected within the same distance from the spoofing position.
The initially selected location \#1 (marked with a gray plus) lies in a region of relatively low channel rate. For this reason, we also consider an \emph{optimized} location \#2 (marked with a teal minus), chosen to yield a channel rate similar to that of the true position (black cross).}

Fig.~\ref{fig:communication} reports the achievable channel rate in the absence of spoofing (dashed blue circles), and under spoofing, using the oracle solution O-HT (orange square) and the blind solution B-HT (yellow square) for spoofing position \#1, as well as the oracle solution (green diamond) for optimized position \#2. All results are compared to the ideal channel rate with perfect \ac{csi} (dashed black line). \blue{The results show that the communication penalty is geometry-dependent: spoofing position \#2 lies in a direction compatible with the selected beams and therefore achieves a rate close to the no-spoofing case, whereas position \#1 induces a larger rate loss. The blind solution has a lower rate in this example because its induced position is far from the intended spoofed geometry.}

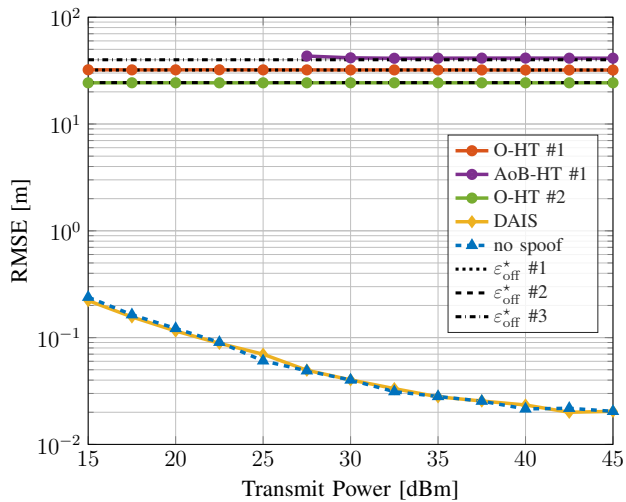
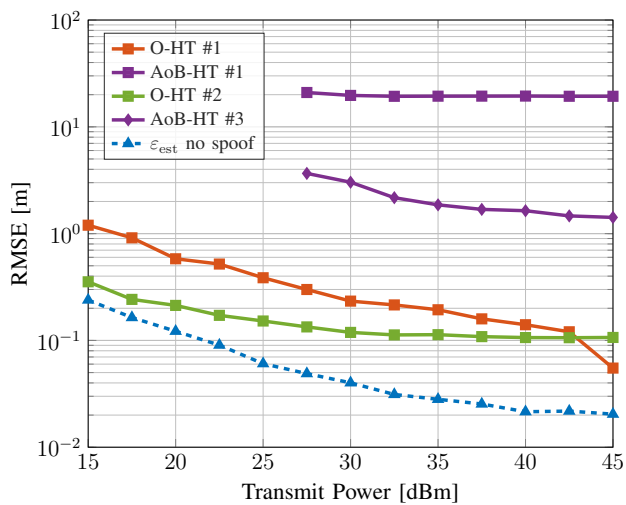
\begin{figure}[t]
    \centering
    \subfloat[Positioning error \label{fig:est_performance}]{ 
    \resizebox{\columnwidth}{!}{\definecolor{color_spoof}{rgb}{0.85098,0.32549,0.09804}
\definecolor{color_aob}{rgb}{0.49412,0.18431,0.55686}
\definecolor{color_bht}{rgb}{0.46667,0.67451,0.18824}
\definecolor{color_dais}{rgb}{0.92941,0.69412,0.12549}
\definecolor{color_clean}{rgb}{0.00000,0.44700,0.74100}

\begin{tikzpicture}
% \begin{axis}[%
% width=4.4in,
% height=3.6in,
% at={(0.758in,0.481in)},
% xlabel={Transmit Power [dBm]},
% ylabel={RMSE [m]},
% xlabel style={font=\large},
% ylabel style={font=\large},
% ymode=log,
% ymin=1e-2,
% ymax=1e2,
% xmin = 15,
% xmax = 45,
% yminorticks=true,
% ytick={1e-2,1e-1,1e0,1e1,1e2},
% yticklabels={$10^{-2}$,$10^{-1}$,$10^0$,$10^1$,$10^2$},
% xmajorgrids,
% ymajorgrids,
% yminorgrids,
% legend style={at={(0.74,0.725)}, anchor=north west, font=\normalsize, draw=white!15!black, fill opacity=0.85},
% legend cell align={left},
% ticklabel style={font=\large}
% ]
\begin{axis}[%
width=4.4in,
height=3.3in,
at={(0.758in,0.481in)},
xlabel={Transmit Power [dBm]},
ylabel={RMSE [m]},
xlabel style={font=\large},
ylabel style={font=\large},
ymode=log,
ymin=1e-2,
ymax=1e2,
xmin = 15,
xmax = 45,
yminorticks=true,
ytick={1e-2,1e-1,1e0,1e1,1e2},
yticklabels={$10^{-2}$,$10^{-1}$,$10^0$,$10^1$,$10^2$},
xmajorgrids,
ymajorgrids,
yminorgrids,
legend style={at={(0.585,0.725)}, anchor=north west, font=\normalsize, draw=white!15!black, fill opacity=0.85},
legend cell align={left},
ticklabel style={font=\large}
]

% % Spoofing: Oracle
% \addplot [color=color_spoof, line width=1.8pt, mark=square*, mark options={solid, color_spoof}]
% table[row sep=crcr]{
% 15	1.1991\\ 17.5	0.9134\\ 20	0.5815\\ 22.5	0.5191\\ 25	0.3863\\
% 27.5	0.2998\\ 30	0.2334\\ 32.5	0.2150\\ 35	0.1937\\ 37.5	0.1592\\
% 40	0.1403\\ 42.5	0.1205\\ 45	0.0551\\
% };
% %\addlegendentry{Spoofing: Oracle}
% \addlegendentry{$\varepsilon^{}_{\mathrm{dev}}$ O-HT}

% % Spoofing: AoB-HT
% \addplot [color=color_aob, line width=1.8pt, mark=diamond*, mark options={solid, color_aob}]
% table[row sep=crcr]{
% 27.5	20.9592\\ 30	19.6858\\ 32.5	19.3092\\ 35	19.3614\\
% 37.5	19.3814\\ 40	19.4080\\ 42.5	19.3239\\ 45	19.3227\\
% };
% %\addlegendentry{Spoofing: AoB-HT}
% \addlegendentry{$\varepsilon^{}_{\mathrm{dev}}$ AoB-HT}

% Obfuscation: Oracle
\addplot [color=color_spoof, line width=1.8pt, mark=*, mark options={solid, color_spoof}, mark size=2pt]
table[row sep=crcr]{
15	32.1822\\ 17.5	32.1259\\ 20	32.1442\\ 22.5	32.1834\\ 25	32.1253\\
27.5	32.1089\\ 30	32.0736\\ 32.5	32.1026\\ 35	32.0837\\
37.5	32.0585\\ 40	32.0422\\ 42.5	32.0289\\ 45	32.0002\\
};
%\addlegendentry{Obfuscation: Oracle}
\addlegendentry{O-HT, loc. \#1}

% Obfuscation: AoB-HT
\addplot [color=color_aob, line width=1.8pt, mark=*, mark options={solid, color_aob}, mark size=2pt]
table[row sep=crcr]{
27.5	43.2753\\ 30	41.6629\\ 32.5	41.2539\\ 35	41.3529\\
37.5	41.3926\\ 40	41.4313\\ 42.5	41.3316\\ 45	41.3334\\
};
%\addlegendentry{Obfuscation: AoB-HT}
\addlegendentry{AoB-HT, loc. \#1}

%% Obfuscation: B-HT
% \addplot [color=color_bht, line width=1.8pt, mark=*, mark options={solid, color_bht}]
% table[row sep=crcr]{
% 15	1561.5712\\ 17.5	37309.5373\\ 20	3926.4585\\ 22.5	9999.9902\\
% 25	90987.9872\\ 27.5	1228766.3598\\ 30	80042.7358\\
% 32.5	169802.0003\\ 35	480866.1068\\ 37.5	268007.0534\\
% 40	516510.5848\\ 42.5	1193306.0964\\ 45	9763217.8695\\
% };
% %\addlegendentry{Obfuscation: B-HT}
% \addlegendentry{$\varepsilon^{}_{\mathrm{est}}$ B-HT}

\addplot [color=color_bht, line width=1.8pt, mark=*, mark options={solid, color_bht}, mark size=2pt]
  table[row sep=crcr]{%
15	24.2993930715361\\
17.5	24.3297812389201\\
20	24.3195694192072\\
22.5	24.3017042211679\\
25	24.3024707829679\\
27.5	24.3053060827921\\
30	24.3085485390416\\
32.5	24.3094675328117\\
35	24.3050943904781\\
37.5	24.3078719188193\\
40	24.308405149984\\
42.5	24.3080982131877\\
45	24.3070608521079\\
};
\addlegendentry{O-HT, loc. \#2}

% Obfuscation: DAIS
\addplot [color=color_dais, line width=1.8pt, mark=diamond*, mark options={solid, color_dais}, mark size=2pt]
table[row sep=crcr]{
15	0.2221\\ 17.5	0.1562\\ 20	0.1153\\ 22.5	0.0884\\ 25	0.0699\\
27.5	0.0495\\ 30	0.0402\\ 32.5	0.0333\\ 35	0.0279\\
37.5	0.0255\\ 40	0.0234\\ 42.5	0.0200\\ 45	0.0204\\
};
%\addlegendentry{Obfuscation: DAIS}
\addlegendentry{DAIS}

% No Obfuscation (clean)
\addplot [color=color_clean,  dashed, line width=1.8pt, mark=triangle*, mark options={solid, color=color_clean}, mark size=2pt]
table[row sep=crcr]{
15	0.2390\\ 17.5	0.1638\\ 20	0.1219\\ 22.5	0.0904\\ 25	0.0606\\
27.5	0.0489\\ 30	0.0401\\ 32.5	0.0312\\ 35	0.0281\\
37.5	0.0254\\ 40	0.0215\\ 42.5	0.0218\\ 45	0.0204\\
};
%\addlegendentry{No Obfuscation}
\addlegendentry{no spoof}%{$\varepsilon^{\star}_{\mathrm{est}}$}

\addplot [black, dotted, line width=1.6pt]
table[row sep=crcr]{15	32.0156\\ 45	32.0156\\};
\addlegendentry{$\varepsilon^{\star}_{\mathrm{off}}$, loc. \#1}

\addplot [color=black, dashed, line width=1.6pt]
  table[row sep=crcr]{%
15	24.4131112314674\\
45	24.4131112314674\\
};
\addlegendentry{$\varepsilon^{\star}_{\mathrm{off}}$, loc. \#2}

\addplot [color=black, dashdotted, line width=1.6pt]
  table[row sep=crcr]{%
15	39.9537032502621\\
45	39.9537032502621\\
};
\addlegendentry{$\varepsilon^{\star}_{\mathrm{off}}$, loc. \#3}

\end{axis}
\end{tikzpicture}}
    }\\
    \subfloat[Spoofing deviation \label{fig:dev_performance}]{ 
    \resizebox{\columnwidth}{!}{\definecolor{color_spoof}{rgb}{0.85098,0.32549,0.09804}
\definecolor{color_aob}{rgb}{0.49412,0.18431,0.55686}
\definecolor{color_bht}{rgb}{0.46667,0.67451,0.18824}
\definecolor{color_dais}{rgb}{0.92941,0.69412,0.12549}
\definecolor{color_clean}{rgb}{0.00000,0.44700,0.74100}

\begin{tikzpicture}
\begin{axis}[%
width=4.4in,
height=3.3in,
at={(0.758in,0.481in)},
xlabel={Transmit Power [dBm]},
ylabel={RMSE [m]},
xlabel style={font=\large},
ylabel style={font=\large},
ymode=log,
ymin=1e-2,
ymax=1e2,
xmin = 15,
xmax = 45,
yminorticks=true,
ytick={1e-2,1e-1,1e0,1e1,1e2},
yticklabels={$10^{-2}$,$10^{-1}$,$10^0$,$10^1$,$10^2$},
xmajorgrids,
ymajorgrids,
yminorgrids,
legend style={at={(0.46,0.725)}, anchor=north west, font=\normalsize, draw=white!15!black, fill opacity=0.85},
legend cell align={left},
legend pos ={north west},
ticklabel style={font=\large}
]

% Spoofing: Oracle
\addplot [color=color_spoof, line width=1.8pt, mark=square*, mark options={solid, color_spoof}]
table[row sep=crcr]{
15	1.1991\\ 17.5	0.9134\\ 20	0.5815\\ 22.5	0.5191\\ 25	0.3863\\
27.5	0.2998\\ 30	0.2334\\ 32.5	0.2150\\ 35	0.1937\\ 37.5	0.1592\\
40	0.1403\\ 42.5	0.1205\\ 45	0.0551\\
};
%\addlegendentry{Spoofing: Oracle}
\addlegendentry{O-HT, loc. \#1}

% Spoofing: AoB-HT
\addplot [color=color_aob, line width=1.8pt, mark=square*, mark options={solid, color_aob}]
table[row sep=crcr]{
27.5	20.9592\\ 30	19.6858\\ 32.5	19.3092\\ 35	19.3614\\
37.5	19.3814\\ 40	19.4080\\ 42.5	19.3239\\ 45	19.3227\\
};
%\addlegendentry{Spoofing: AoB-HT}
\addlegendentry{AoB-HT, loc. \#1}

\addplot [color=color_bht, line width=1.8pt, mark=square*, mark options={solid, color_bht}]
  table[row sep=crcr]{%
15	0.354123266137906\\
17.5	0.242275834320756\\
20	0.212761075000605\\
22.5	0.171722910948071\\
25	0.152029388449887\\
27.5	0.133848991901416\\
30	0.118892190100176\\
32.5	0.112647326377211\\
35	0.113216176372243\\
37.5	0.108535354996469\\
40	0.106356162446992\\
42.5	0.106016021363265\\
45	0.106633332463708\\
};
\addlegendentry{O-HT, loc. \#2}

\addplot [color=color_aob, line width=1.8pt, mark=diamond*, mark options={solid, color_aob}, mark size=3pt]
  table[row sep=crcr]{
27.5	3.6631\\ 30	 3.0287\\ 32.5	2.17\\35	1.86\\
37.5	1.6864\\ 40	1.6386\\ 42.5	1.4654\\45	1.4186\\
};
\addlegendentry{AoB-HT, loc. \#3}

\addplot [color=color_clean,  dashed, line width=1.8pt, mark=triangle*, mark options={solid, color=color_clean}]
table[row sep=crcr]{
15	0.2390\\ 17.5	0.1638\\ 20	0.1219\\ 22.5	0.0904\\ 25	0.0606\\
27.5	0.0489\\ 30	0.0401\\ 32.5	0.0312\\ 35	0.0281\\
37.5	0.0254\\ 40	0.0215\\ 42.5	0.0218\\ 45	0.0204\\
};
%\addlegendentry{No Obfuscation}
\addlegendentry{$\varepsilon^{}_{\mathrm{est}}$ no spoof}
\end{axis}
\end{tikzpicture}}
    }
     \caption{\black{Performance metrics comparison between the different spoofing methods and locations.}}
    \label{fig:performance}
    \vspace{-5pt}
\end{figure}

\subsubsection{Performance Metrics vs Transmit Power}
\blue{We compare oracle HoloTrace (O-HT), blind HoloTrace (B-HT), angle-only blind HoloTrace (AoB-HT, obtained by setting $\mathbf{x}_{\mathrm{d}}^{}=\mathbf{1}_M^{}$), and \ac{dais} performance metrics. Oracle and blind HoloTrace are described in Sections~\ref{sec:oracle} and~\ref{sec:blind}, respectively. \ac{dais} is included as the closest \ac{csi}-free shift-based baseline, matched to our blind variants; a direct comparison with the other methods in Table~\ref{tab:privacy-comparison} is less equitable because those methods target single-parameter modifications or different array assumptions.}
\blue{\ac{dais} applies common delay and \ac{aod} shifts to all paths and assumes a synchronized \ac{ue} and a single-antenna \ac{bs}. For illustration, we use $\Delta_\tau = 15/c$ seconds and $\Delta_\phi = 0.17$\,rad, but in our asynchronous multi-antenna single-anchor setting, the position estimate depends on inter-path differences and on the absolute \ac{los} \ac{aoa}, which \ac{dais} does not affect. Hence, the common shifts cancel in \eqref{eq:d_los}, and the result is independent of the particular shift values.}
\black{Fig.~\ref{fig:performance} provides a unified view of the positioning error (Fig.~\ref{fig:est_performance}) and spoofing deviation (Fig.~\ref{fig:dev_performance}) across all the approaches and locations.
In Fig.~\ref{fig:est_performance}, the orange circle curve for O-HT in location \#1, the purple circle curve for AoB-HT in location \#1, the green circle curve for O-HT in location \#2,} and the yellow diamond curve for \ac{dais} show the actual positioning \ac{rmse} under spoofing or measurement shifts. \black{The desired spoofing offsets $\varepsilon^{\star}_{\mathrm{off}} = \Vert \overline{\mathbf{p}^{}}_{\mathrm{UE}}-\mathbf{p}^{}_{\mathrm{UE}}\Vert$ are indicated with the black dotted, dashed, and dot-dashed lines, depending on the spoofing location (\#1, \#2, or \#3), while the blue triangle line marks the performance without any spoofing/shifting applied. Location \#3 is derived from \eqref{eq:p_est} using the estimated measurements of the AoB-HT solution (see Section~\ref{sec:meas_dev}).
In Fig.~\ref{fig:dev_performance}, the orange square for O-HT in location \#1, the purple square for AoB-HT in location \#1, and the green square curve for O-HT in location \#2 indicate how well each method can ``move'' the estimated position to a desired (spoofed) location with respect to the performance without spoofing. The purple diamond curve, instead, shows the spoofing deviation of AoB-HT with respect to the location \#3. B-HT performances are not shown because of the large error.}

\blue{From the figures, we observe several distinct behaviors: \textit{(i)} in the evaluated oracle setting, O-HT closely realizes the target spoofed position, as its positioning error follows the desired spoofing offset $\varepsilon^{\star}_{\mathrm{off}}$ and its spoofing deviation remains small. \textit{(ii)} \ac{dais} does not spoof the location in our single-anchor, difference-based setting; its error remains close to the no-spoofing positioning error because common shifts in delays and angles do not change the measurement differences used in \eqref{eq:d_los}. \textit{(iii)} Blind (B-HT) can obfuscate, i.e., create large errors in the estimated position, but it does not reliably steer the estimate to the chosen spoofed location. This is primarily due to the pairing problem: the Kronecker product structure couples multiple path delays to the same angle, causing the estimator to recover inconsistent path tuples. Fig.~\ref{fig:pairing_problem} illustrates this limitation for a joint \ac{aoa}/\ac{toa} \ac{mf} scenario. \textit{(iv)} Angle-only blind (AoB-HT) provides a more controlled approximate spoofing behavior than B-HT in this scenario, especially at higher transmit powers. At low transmit powers, however, its performance degrades because the weaker path is masked by noise, making two-path localization infeasible for received powers below -44.6\,dBm in our settings.}
\begin{figure}[t]
    \centering
    \begin{tikzpicture}
    \node[] at (0,0) {\includegraphics[width=0.9\columnwidth, trim=0 1cm 0 1.5cm, clip]{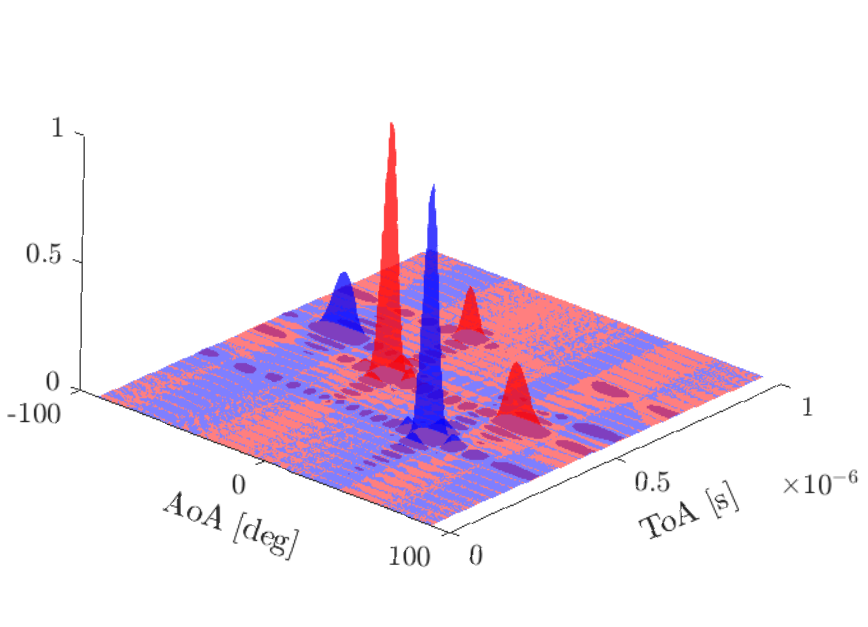}};
    % \begin{scope}[yshift=3.5mm]
    % \draw[black, dashed, line width = 1pt] (-0.35,-0.75)--(-0.345,1.3);
    % \draw[black, dashed, line width = 1pt] (1.43,-0.87)--(1.43,-0.25);
    % \end{scope}
     \begin{scope}[yshift=3.5mm]
    \draw[black, dashed, line width = 1pt] (-0.4,-0.75)--(-0.38,1.5);
    \draw[black, dashed, line width = 1pt] (1.62,-0.93)--(1.62,-0.33);
    \end{scope}
    \end{tikzpicture}
    \caption{\black{Pairing problem example for Blind HoloTrace solution, showing the joint AoA/ToA spectra. The red and blue curves are, respectively, the MF estimations with and without spoofing. For simplicity, the ToA spoofing is given with a single $\Delta^{}_\tau$ shift without path injection. The black dashed lines are the desired spoofing values.}}
    \label{fig:pairing_problem}
\end{figure}

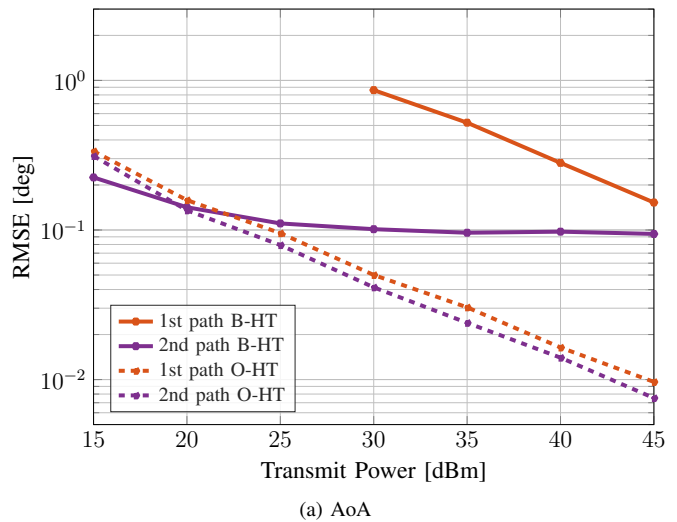
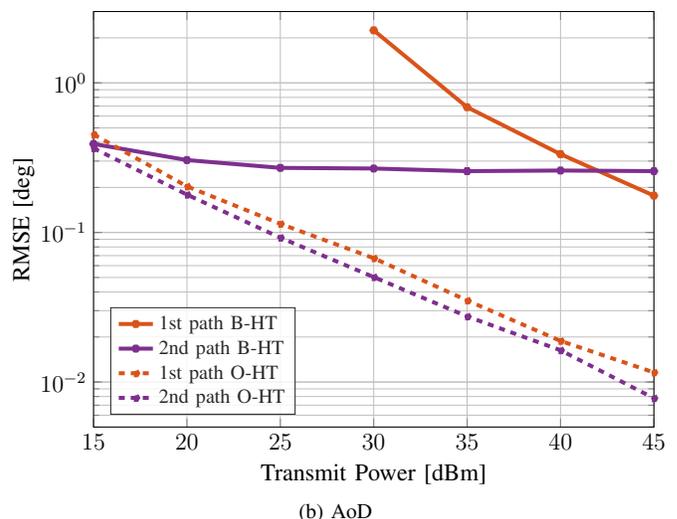
\begin{figure}[t]
    \centering
    \subfloat[AoA \label{fig:aoa_meas}]{ 
    \resizebox{\columnwidth}{!}{% This file was created by matlab2tikz.
\definecolor{mycolor1}{rgb}{0.85098,0.32549,0.09804}%
\definecolor{mycolor2}{rgb}{0.00000,0.44700,0.74100}%
\begin{tikzpicture}

\begin{axis}[%
width=4.5in,
height=2.8in,
at={(0.758in,0.481in)},
xlabel={Transmit Power [dBm]},
ylabel={RMSE [deg]},
xlabel style={font=\large},
ylabel style={font=\large},
ymode=log,
ymin=5e-3,
ymax=3e0,
xmin = 15,
xmax = 45,
yminorticks=true,
ytick={1e-3,1e-2,1e-1,1e0,1e1},
yticklabels={$10^{-3}$,$10^{-2}$,$10^{-1}$,$10^{0}$,$10^{1}$},
xtick={15,20,25,30,35,40,45},
xmajorgrids,
ymajorgrids,
yminorgrids,
legend style={at={(085,0.725)}, anchor=north west, font=\normalsize, draw=white!15!black, fill opacity=0.85},
legend cell align={left},
legend pos = {south west},
ticklabel style={font=\large}
]

\addplot [color=mycolor1, mark=asterisk, line width=2pt, mark size=3pt]
  table[row sep=crcr]{%
30	0.861479936759003\\
35	0.521802617811235\\
40	0.281013208797018\\
45	0.152675819859978\\
};
\addlegendentry{\nth{1} path B-HT}
\addplot [color=mycolor2, mark=asterisk,line width=2pt, mark size=3pt]
  table[row sep=crcr]{%
15	0.224624578662731\\
20	0.141876470235656\\
25	0.110628538958324\\
30	0.101259678825463\\
35	0.0958612126980995\\
40	0.0974502363736554\\
45	0.0941332644245541\\
};
\addlegendentry{\nth{2} path B-HT}

\addplot [color=mycolor1, dashed, mark=asterisk, line width=2pt, mark size=3pt]
  table[row sep=crcr]{%
15	0.336548310218893\\
20	0.158397313889866\\
25	0.0952433637903796\\
30	0.050035160681951\\
35	0.0304086506957039\\
40	0.0163875211572668\\
45	0.00964366964041818\\
};
\addlegendentry{\nth{1} path O-HT}
\addplot [color=mycolor2, dashed, mark=asterisk, line width=2pt, mark size=3pt]
  table[row sep=crcr]{%
15	0.312407646832556\\
20	0.135368390971179\\
25	0.0792872833740157\\
30	0.0413726569037083\\
35	0.0239511800681419\\
40	0.0140112240248822\\
45	0.00752341642568915\\
};
\addlegendentry{\nth{2} path O-HT}
\end{axis}
\end{tikzpicture}%
}
    }\\
    \subfloat[AoD \label{fig:aod_meas}]{ 
    \resizebox{\columnwidth}{!}{% This file was created by matlab2tikz.
\definecolor{mycolor1}{rgb}{0.85098,0.32549,0.09804}%
\definecolor{mycolor2}{rgb}{0.00000,0.44700,0.74100}%
\begin{tikzpicture}

\begin{axis}[%
width=4.5in,
height=2.8in,
at={(0.758in,0.481in)},
xlabel={Transmit Power [dBm]},
ylabel={RMSE [deg]},
xlabel style={font=\large},
ylabel style={font=\large},
ymode=log,
ymin=5e-3,
ymax=3e0,
xmin = 15,
xmax = 45,
yminorticks=true,
ytick={1e-3,1e-2,1e-1,1e0,1e1},
yticklabels={$10^{-3}$,$10^{-2}$,$10^{-1}$,$10^{0}$,$10^{1}$},
xtick={15,20,25,30,35,40,45},
xmajorgrids,
ymajorgrids,
yminorgrids,
legend style={at={(085,0.725)}, anchor=north west, font=\normalsize, draw=white!15!black, fill opacity=0.85},
legend cell align={left},
legend pos = {south west},
ticklabel style={font=\large}
]

\addplot [color=mycolor1, mark=asterisk, line width=2pt, mark size=3pt]
  table[row sep=crcr]{%
30	2.24610073476887\\
35	0.687844359276948\\
40	0.333909709456944\\
45	0.17628204004795\\
};
\addlegendentry{\nth{1} path B-HT}
\addplot [color=mycolor2, mark=asterisk,line width=2pt, mark size=3pt]
  table[row sep=crcr]{%
15	0.39071337321755\\
20	0.304723559836054\\
25	0.270218084135611\\
30	0.267590930569323\\
35	0.256939861296453\\
40	0.259424354540545\\
45	0.257012002290949\\
};
\addlegendentry{\nth{2} path B-HT}

\addplot [color=mycolor1, mark=asterisk, dashed, line width=2pt, mark size=3pt]
  table[row sep=crcr]{%
15	0.452153872441301\\
20	0.203380291306529\\
25	0.114240368078173\\
30	0.0672839815076981\\
35	0.0351114013246226\\
40	0.0188352034800083\\
45	0.0115896424985195\\
};
\addlegendentry{\nth{1} path O-HT}
\addplot [color=mycolor2, mark=asterisk, dashed, line width=2pt, mark size=3pt]
  table[row sep=crcr]{%
15	0.365560974091837\\
20	0.179090934526353\\
25	0.0922905361241793\\
30	0.0503022120550704\\
35	0.027456775888399\\
40	0.0163077988354574\\
45	0.00779659652368217\\
};
\addlegendentry{\nth{2} path O-HT}
\end{axis}

% \begin{axis}[%
% width=5.833in,
% height=4.375in,
% at={(0in,0in)},
% scale only axis,
% xmin=0,
% xmax=1,
% ymin=0,
% ymax=1,
% axis line style={draw=none},
% ticks=none,
% axis x line*=bottom,
% axis y line*=left
% ]
% \end{axis}
\end{tikzpicture}%
}
    }
    \caption{\black{HoloTrace (a) AoA and (b) AoD measurement deviation of the oracle (dashed line) and blind (solid line) solution. The orange and purple colors represent the first and the second path, respectively.}}
    \label{fig:angle_meas}
    %\vspace{-5pt}
\end{figure}

\subsubsection{Measurement Deviation Analysis}
\label{sec:meas_dev}
\blue{To better characterize HoloTrace, we analyze the measurement deviation between the estimated and desired measurements. Fig.~\ref{fig:angle_meas} reports the angle measurement mismatch of the oracle and blind HoloTrace solutions for both \ac{aoa} (Fig.~\ref{fig:aoa_meas}) and \ac{aod} (Fig.~\ref{fig:aod_meas}). For the oracle solution (dashed line), the \ac{rmse} decreases with transmit power and converges to small values for both paths, consistent with the closed-form oracle construction.}

\blue{For the blind solution (solid line), we observe a different behavior. Because the order of the paths cannot be controlled, the secondary path often becomes predominant. At low transmit powers (up to 27.5\,dBm), noise makes the secondary path, i.e., the \ac{los} path, undetectable by FLEX, preventing the blind angle-only scheme from producing a position estimate. Blind HoloTrace can still provide an estimate by duplicating the primary path, but this estimate can be far from the intended spoofed position. Unlike the oracle case, the measurement deviation for blind HoloTrace does not decrease consistently with transmit power, and a residual mismatch persists across all tested power levels. Thus, the blind design is best interpreted as an approximate spoofing or obfuscation mechanism rather than an exact multipath spoofing solution.}

\begin{figure}[t]
    \centering
    \resizebox{\columnwidth}{!}{% This file was created by matlab2tikz.
\definecolor{mycolor1}{rgb}{0.00000,0.44700,0.74100}%
\definecolor{mycolor2}{rgb}{0.85000,0.32500,0.09800}%
\begin{tikzpicture}

\begin{axis}[%
width=4.516in,
height=1.6in, %3.487in,
at={(0.758in,0.57in)},
scale only axis,
xmin=-1.53814580106126e-26,
xmax=1.00000000135465,
xlabel style={font=\Large\color{white!15!black}},
xlabel={$\sigma_\epsilon^\mathrm{UE}$ [m]},
ymin=2.75713700584129,
ymax=45.9941886308302,
ylabel style={font=\Large\color{white!15!black}},
ylabel={RMSE [m]},
axis background/.style={fill=white},
xmajorgrids,
ymajorgrids,
legend style={legend pos=south east, legend cell align=left, align=left, draw=white!15!black, font=\Large},
ticklabel style={font=\Large}
]
\addplot [color=mycolor1, line width=1.5pt, mark=o, mark options={solid, mycolor1}, mark size=2pt]
  table[row sep=crcr]{%
0	32.1044319644883\\
0.1	38.5373943767044\\
0.2	42.0017373450955\\
0.3	39.1234307898614\\
0.4	41.7318588443901\\
0.5	35.8034925960344\\
0.6	37.1139002670664\\
0.7	33.9790626618873\\
0.8	35.4422661748293\\
0.9	38.6404427751999\\
1	45.9941884554808\\
};
\addlegendentry{$\varepsilon^{}_{\mathrm{est}}$}

\addplot [color=mycolor2, line width=1.5pt, mark=o, mark options={solid, mycolor2}, mark size=2pt]
  table[row sep=crcr]{%
0	2.75713700584129\\
0.1	16.9940822758729\\
0.2	24.0190547553505\\
0.3	23.4347587321629\\
0.4	27.5843394054154\\
0.5	23.8873405107007\\
0.6	26.4585219939954\\
0.7	26.5216565137525\\
0.8	25.8610544126743\\
0.9	27.1884498245865\\
1	35.7215740814421\\
};
\addlegendentry{$\varepsilon^{}_{\mathrm{dev}}$}

\addplot [color=black, dotted, line width=1.5pt]
  table[row sep=crcr]{%
-1.53814580106126e-26	32.0156211871642\\
1.00000000135465	32.0156211871642\\
};
\addlegendentry{$\varepsilon^{\star}_{\mathrm{off}}$}

\end{axis}
\end{tikzpicture}%
}
    \caption{\blue{Oracle HoloTrace sensitivity to UE location uncertainty.}}
    \label{fig:uncertainty}
    \vspace{-2pt}
\end{figure}

\begin{figure}[t]
    \centering
    \resizebox{\columnwidth}{!}{% This file was created by matlab2tikz.
\definecolor{mycolor1}{rgb}{0.00000,0.44700,0.74100}%
\definecolor{mycolor2}{rgb}{0.85000,0.32500,0.09800}%
\begin{tikzpicture}

\begin{axis}[%
width=4.516in,
height=1.6in, %3.487in,
at={(0.758in,0.57in)},
scale only axis,
xmin=-20,
xmax=20,
xlabel style={font=\Large\color{white!15!black}},
xlabel={$\varrho_\mathrm{dB}^{}$ [dB]},
ymode=log,
ymin=1e-1,
ymax=1e2,
ylabel style={font=\Large\color{white!15!black}},
ylabel={RMSE [m]},
axis background/.style={fill=white},
xmajorgrids,
ymajorgrids,
legend style={ at={(0.02, 0.65)}, legend cell align=left, anchor=north west, align=left, draw=white!15!black, font=\Large},
ticklabel style={font=\Large}
]
\addplot [color=mycolor1, line width=1.5pt, mark=o, mark options={solid, mycolor1}, mark size=2pt]
  table[row sep=crcr]{%
-20.0000 24.3093 \\
-17.5000 24.3061 \\
-15.0000 24.3065 \\
-12.5000 24.3043 \\
-10.0000 24.3059 \\
-7.5000 24.3117 \\
-5.0000 24.3062 \\
-2.5000 24.3083 \\
0.0000 24.3096 \\
2.5000 24.3074 \\
5.0000 24.3038 \\
7.5000 24.3059 \\
10.0000 24.3075 \\
12.5000 24.3075 \\
15.0000 24.3562 \\
17.5000 24.7820 \\
20.0000 26.0608 \\ 
};
\addlegendentry{$\varepsilon^{}_{\mathrm{est}}$}

\addplot [color=mycolor2, line width=1.5pt, mark=o, mark options={solid, mycolor2}, mark size=2pt]
  table[row sep=crcr]{%
-20.0000 0.1095 \\
-17.5000 0.1128 \\
-15.0000 0.1124 \\
-12.5000 0.1140 \\
-10.0000 0.1121 \\
-7.5000 0.1092 \\
-5.0000 0.1179 \\
-2.5000 0.1128 \\
0.0000 0.1106 \\
2.5000 0.1228 \\
5.0000 0.1200 \\
7.5000 0.1110 \\
10.0000 0.1075 \\
12.5000 0.1084 \\
15.0000 2.8111 \\
17.5000 3.5601 \\
20.0000 10.3276 \\     
};
\addlegendentry{$\varepsilon^{}_{\mathrm{dev}}$}

\addplot [color=black, dotted, line width=1.5pt]
  table[row sep=crcr]{%
-20	24.413111231467404\\
20	24.413111231467404\\
};
\addlegendentry{$\varepsilon^{\star}_{\mathrm{off}}$}

\end{axis}
\end{tikzpicture}%
}
    \caption{\blue{Oracle HoloTrace sensitivity to \ac{csi} uncertainty.}}
    \label{fig:CSI_uncertainty}
    \vspace{-2pt}
\end{figure}

\subsubsection[Sensitivity to CSI and UE Location Uncertainty]{\blue{Sensitivity to CSI and UE Location Uncertainty}}
\blue{We assess the sensitivity of oracle HoloTrace to \ac{csi} and positioning uncertainty by fixing the transmit power to $P_t^{}=35$\,dBm. Since practical \acp{ue} may suffer from localization and channel-estimation errors, we consider two uncertainty sources: imperfect position knowledge and imperfect \ac{csi}.}

\blue{For location uncertainty, we perturb the \ac{ue} position $\mathbf{p}_\mathrm{UE}^{}$ with Gaussian noise of standard deviation $\sigma_\epsilon^{\mathrm{UE}} \in [0,1]$\,m, while the \ac{sp} position $\mathbf{p}_\mathrm{SP}^{}$ is perturbed with fixed standard deviation $\sigma_\epsilon^{\mathrm{SP}}=0.1$\,m. Estimates outside the \ac{bs} coverage area (150\,m) are discarded. As shown in Fig.~\ref{fig:uncertainty}, larger $\sigma_\epsilon^{\mathrm{UE}}$ increases the spoofing deviation, whereas the positioning error remains above the prescribed offset in the tested range.}
\blue{For \ac{csi} uncertainty, we add a complex perturbation to the channel coefficients,
$\widehat{\boldsymbol{\alpha}} = \boldsymbol{\alpha} + \delta\boldsymbol{\alpha}$,
with
$\mathbb{E}[\|\delta\boldsymbol{\alpha}\|^2] = \varrho\, \mathbb{E}[\|\boldsymbol{\alpha}\|^2]$,
where $\varrho$ is the \ac{nmse}. Fig.~\ref{fig:CSI_uncertainty} shows the results for $\varrho_\mathrm{dB}^{}=10\log_{10}\varrho \in [-20,20]$~dB. 
The spoofing performance remains controlled up to about 12.5\,dB; beyond this point, the spoofing deviation increases, while the position estimation error still follows the prescribed offset threshold.}

\blue{Overall, in the tested uncertainty ranges, oracle HoloTrace continues to displace the estimated position away from the true one, but its accuracy with respect to the desired spoofed location degrades as location or \ac{csi} uncertainty increases.}

\section{Conclusion and Future Work}
\label{sec:conclusion}
\vspace{-3pt}

\blue{In this paper, we presented HoloTrace, a signal-level location-privacy framework for single-\ac{bs} mmWave \ac{mimo}-\ac{ofdm} localization with analog arrays and an unsynchronized \ac{ue}\cyan{, relevant to next-generation \ac{its}}. 
By modifying pilot symbols rather than the spatial-domain precoder, we formulate \ac{aoa}, \ac{aod}, and \ac{tdoa} spoofing as a subspace projection-based signal design problem. 
The proposed framework is validated by theoretical analysis and numerical evaluation, characterizing achievable spoofing performance under oracle and blind operations.
With oracle \ac{csi} knowledge, closed-form designs realize a desired spoofed geometry, while blind designs induce obfuscation or approximate spoofing. 
In the evaluated two-path scenario, oracle HoloTrace steers the inferred position toward the target, whereas blind variants are limited by the angle-delay pairing ambiguity. 
The communication penalty is not fixed but geometry-dependent: locations compatible with the selected beams preserve rates close to the no-spoofing case, and oracle HoloTrace also remains effective over the tested location- and \ac{csi}-error ranges, with accuracy degrading as uncertainty grows.}

\blue{Despite these promising results, blind spoofing remains fundamentally limited in general multipath scenarios due to the unresolved path-association ambiguity. 
Furthermore, the presented evaluation focused on a single-\ac{bs}, two-path geometry with a specific estimator. 
Future work should address robust blind spoofing under richer multipath \cyan{and mobility conditions}, digital and hybrid arrays, multi-anchor settings, and adversaries whose positions and orientations are unknown to the \ac{ue}, as well as spoofed-location selection rules that balance location privacy and communication rate.}
%\cyan{Extending the analysis to dynamic \ac{cav} scenarios, where mobility, rapidly varying channels, and repeated positioning observations may affect spoofing consistency and detectability, is another relevant direction.}}
\vspace{-8pt}

\appendices

% \section{Proof Proposition \ref{prop:general_solution}}
% \label{app:general_solution}

\section{Single-Domain Oracle Spoofing}\label{app:single_domain_oracle}
\blue{For the single-domain model in \eqref{eq:single_domain_oracle_model}, the conditional \ac{ls} path-gain estimate for a candidate $\bm{\varpi}$ is $\widehat{\bm{\alpha}}(\bm{\varpi})=\mathbf{A}^{\dagger}(\bm{\varpi})\mathbf{y}_{\mathrm{tm}}^{}$, and substitution into the residual gives $C(\bm{\varpi};\mathbf{y}_{\mathrm{tm}}^{})=\|\mathbf{P}^{\perp}_{\mathbf{A}(\bm{\varpi})}\mathbf{y}_{\mathrm{tm}}^{}\|^2$. %, with $\mathbf{P}^{\perp}_{\mathbf{A}(\bm{\varpi})}=\mathbf{I}-\mathbf{A}(\bm{\varpi})\mathbf{A}^{\dagger}(\bm{\varpi})$. 
Hence perfect oracle spoofing at $\overline{\bm{\omega}}$ holds whenever the noiseless received vector belongs to $\mathcal{R}(\mathbf{A}(\overline{\bm{\omega}}))$. The pilot in \eqref{eq:single_domain_oracle_pilot} satisfies $\operatorname{diag}(\widetilde{\mathbf{x}})\mathbf{A}(\bm{\omega})\bm{\alpha}=\mathbf{A}(\overline{\bm{\omega}})\bm{\lambda}$ provided $\mathbf{A}(\bm{\omega})\bm{\alpha}$ has no zero entries; thus the projected residual at the spoofed parameter vector is zero in the absence of receiver noise. The \ac{aoa}, \ac{aod}, and \ac{toa} cases follow by setting $\mathbf{A}=\mathbf{B}$, $\mathbf{A}=\mathbf{C}$, and $\mathbf{A}=\mathbf{D}$, respectively.}
\vspace{-8pt}

\section{Oracle AoA Spoofing with Imperfect CSI}
\label{app:imperfect_csi_aoa}
\blue{
%We prove Proposition~\ref{prop:imperfect_csi_aoa}. 
Consider the \ac{simo} model in Section~\ref{sec:oracle_aoa}, where the oracle design uses $\mathbf{h}=\mathbf{B}(\bm{\theta})\bm{\alpha}$. In practice, the \ac{ue} may only have $\widehat{\mathbf{h}}=\mathbf{h}+\delta\mathbf{h}$, which covers both additive and first-order multiplicative gain errors mapped to the effective channel. Given the target $\overline{\mathbf{h}}=\mathbf{B}(\overline{\bm{\theta}})\bm{\lambda}$, the mismatched pilot $\widetilde{\mathbf{x}}=\operatorname{diag}^{-1}(\widehat{\mathbf{h}})\overline{\mathbf{h}}$ applied to the true channel yields the noiseless signal $\mathbf{y}=\operatorname{diag}(\widetilde{\mathbf{x}})\mathbf{h}=\operatorname{diag}(\mathbf{h}\oslash\widehat{\mathbf{h}})\overline{\mathbf{h}}=\boldsymbol{\Phi}\overline{\mathbf{h}}$, with $\boldsymbol{\Phi}=\operatorname{diag}(\mathbf{h}\oslash\widehat{\mathbf{h}})$ capturing the \ac{csi} mismatch. Including receiver noise, $\mathbf{y}=\overline{\mathbf{h}}+\delta\mathbf{y}+\mathbf{q}$ with $\delta\mathbf{y}=(\boldsymbol{\Phi}-\mathbf{I})\overline{\mathbf{h}}$. For small relative errors, $\mathbf{h}\oslash\widehat{\mathbf{h}}=\mathbf{1}\oslash(\mathbf{1}+\delta\mathbf{h}\oslash\mathbf{h})\approx\mathbf{1}-\delta\mathbf{h}\oslash\mathbf{h}$, hence $\delta\mathbf{y}\approx-\operatorname{diag}(\overline{\mathbf{h}})(\delta\mathbf{h}\oslash\mathbf{h})$. The \ac{aoa} estimator evaluates $C(\bm{\vartheta})=\|\mathbf{P}_{\mathbf{B}(\bm{\vartheta})}^{\perp}\mathbf{y}\|^2$; since $\mathbf{P}_{\mathbf{B}(\overline{\bm{\theta}})}^{\perp}\overline{\mathbf{h}}=\mathbf{0}$, at the spoofed angle $C(\overline{\bm{\theta}})=\|\mathbf{P}_{\mathbf{B}(\overline{\bm{\theta}})}^{\perp}(\delta\mathbf{y}+\mathbf{q})\|^2$, which proves \eqref{eq:imperfect_csi_projected_cost}.}
\vspace{-8pt}

\section{Blind AoA Spoofing}
\label{app:blindAOA}

\subsection{Single path case}
\black{Considering the case $L=1$, the minimal \ac{simo} model simplifies to $\mathbf{y} = \mathbf{b}(\theta_0^{}) \alpha_0^{} + \mathbf{q}$, where $\mathbf{b}(\theta_0^{})=\mathbf{W}^{\mathsf{H}} \mathbf{a}_{N_{\text{R}}^{}}(\theta_0^{})$. 
The channel gain estimate is $\widehat{\alpha}_0^{}(\theta_0^{}) =\mathbf{b}^{\mathsf{H}}(\theta_0^{})\mathbf{y} / \Vert \mathbf{b}(\theta_0^{})\Vert^2 $. 
Defining the projector $\mathbf{P}^{}_{\mathbf{b}(\vartheta)}={\mathbf{b}(\vartheta) \mathbf{b}^{\mathsf{H}}(\vartheta)}/{\left\Vert \mathbf{b}(\vartheta)\right\Vert^{2}}$, the cost function (see \eqref{eq:spoofingCriterion1}) is $C(\overline{\theta}_0^{}| \widetilde{\mathbf{x}}) = \Vert ( \mathbf{I}_M^{} - \mathbf{P}^{}_{\mathbf{b}(\overline{\theta}_0^{})} ) \text{diag}(\widetilde{\mathbf{x}}) \mathbf{b}(\theta_0^{}) \alpha_0^{} \Vert^2$, in which $\text{diag}(\widetilde{\mathbf{x}}) \mathbf{b}(\theta_0^{}) \alpha_0^{}$ is the noise-free observed signal under spoofed pilot $\widetilde{\mathbf{x}}$.
Perfect spoofing requires $C(\overline{\theta}_0^{}| \widetilde{\mathbf{x}}) = 0$. 
The sufficient condition for perfect spoofing from \eqref{eq:blindAOAreq} in the single-path case becomes
\(\text{diag}(\widetilde{\mathbf{x}})\mathbf{b}(\theta_0^{}) \in \text{span}\left( \mathbf{b}(\overline{\theta}_0^{}) \right)\), which does not require knowledge of $\alpha_0^{}$. Solving for $\widetilde{\mathbf{x}}$ yields $\widetilde{\mathbf{x}} = \lambda \, \mathbf{b}(\overline{\theta}_0^{}) \oslash \mathbf{b}({\theta}_0^{})$ when the corresponding component of $\mathbf{b}({\theta}_0^{})$ is not zero; otherwise the solution does not exist.}
\vspace{-8pt}

\subsection{Multiple path case}
Consider the noiseless spoofing model
%\begin{equation}
$\mathbf y_{\text{tm}}^{} = \operatorname{diag}(\widetilde{\mathbf{x}})\,\mathbf B(\bm\theta)\,\bm\alpha$, 
%\end{equation}
with $\mathbf B(\bm\theta)\in\mathbb C^{M\times L}$ and $\bm\alpha\in\mathbb C^L$. 
Perfect blind spoofing requires the subspace inclusion \eqref{eq:blindAOAreq}.
The condition \eqref{eq:blindAOAreq} is equivalent to the existence of a non-zero vector $\widetilde{\mathbf{x}}$ and a $L\times L$ matrix $\bm{\Omega}$ such that 
$\operatorname{diag}(\widetilde{\mathbf{x}})\,\mathbf B(\bm\theta) = \mathbf B(\bar{\bm\theta}) \bm{\Omega}$.
\blue{This is a homogeneous system of $ML$ linear equations  with $M+L^2$ unknowns (i.e., the elements of $\widetilde{\mathbf{x}}$ and $\bm{\Omega}$). The system is guaranteed to have a non-zero solution in general only if it is undetermined, that is, $M < L^2 /(L-1)$. Therefore, for $L>1$, solutions exist only for small values of $M$, which are not of practical relevance. For $M \ge L^2 /(L-1)$, a non-trivial solution exists only in special cases of limited relevance, e.g., when $\bar{\bm\theta}$ is a permutation of $\bm{\theta}$.}
\vspace{-8pt}

% Consider the noiseless spoofing model
% \begin{equation}
% \mathbf y_{\text{tm}}^{}  = \operatorname{diag}(\widetilde{\mathbf{x}})\,\mathbf B(\bm\theta)\,\bm\alpha,
% \end{equation}
% with $\mathbf B(\bm\theta)\in\mathbb C^{M\times L}$ and $\bm\alpha\in\mathbb C^L$. 
% Perfect blind spoofing requires \eqref{eq:blindAOAreq}. %, which is equivalent to 
%\begin{equation}
%\operatorname{diag}(\widetilde{\mathbf{x}})\,\mathbf B(\bm\theta)\,\bm\alpha \in \mathcal R\!\big(\mathbf B(\overline{\bm\theta})\big), 
%\qquad \forall\,\bm\alpha.
%\end{equation} 
% If $L\ge 2$, then $\operatorname{rank}(\mathbf B(\bm\theta))=L$ for generic $\bm\theta$, and since $\operatorname{diag}(\widetilde{\mathbf{x}})$ is nonsingular, 
% $
% \operatorname{rank}\!\big(\operatorname{diag}(\widetilde{\mathbf{x}})\mathbf B(\bm\theta)\big)=L.
% $
% Thus \eqref{eq:blindAOAreq} implies 
% \(
% \mathcal R(\mathbf B(\bm\theta)) = \mathcal R(\mathbf B(\overline{\bm\theta})),
% \)
% which does not hold in general for distinct angle vectors $\bm\theta \neq \overline{\bm\theta}$. 
% Hence, for $L\ge2$, blind perfect spoofing is impossible.
% }
%Only in the special case $L=1$ does there exist $\widetilde{\mathbf{x}}$ such that \eqref{eq:blindAOAreq} is satisfied.

\section{Proof of Blind ToA Spoofing}% Proposition \ref{prop:toa_spoof}}
\label{app:tdoa_blind}
% Consider the two-path case, the minimum required for localization. The delay steering matrix is $\mathbf{D}(\overline{\bm{\tau}}) \in \mathbb{C}^{K \times 2}$. To spoof the \ac{tdoa} estimate, we seek $\widetilde{\mathbf{x}} \in \mathbb{C}^K$ such that the manipulated received signal $\text{diag}(\widetilde{\mathbf{x}})\mathbf{D}(\bm{\tau})\bm{\alpha}$ lies in the range of $\mathbf{D}(\overline{\bm{\tau}})$, i.e.,
Consider the two-path case, the \blue{minimal case} for localization, with delay steering matrix $\mathbf{D}(\overline{\bm{\tau}}) \in \mathbb{C}^{K \times 2}$. To spoof the \ac{tdoa} estimate, we seek $\widetilde{\mathbf{x}} \in \mathbb{C}^K$ such that the manipulated signal $\text{diag}(\widetilde{\mathbf{x}})\mathbf{D}(\bm{\tau})\bm{\alpha}$ lies in the range of $\mathbf{D}(\overline{\bm{\tau}})$, i.e., $\text{diag}(\widetilde{\mathbf{x}})\mathbf{D}(\bm{\tau})\bm{\alpha} = \beta_0^{}\,\mathbf{d}(\overline{\tau}_0^{}) + \beta_1^{}\,\mathbf{d}(\overline{\tau}_1^{})$ for some $\beta_0^{}, \beta_1^{} \in \mathbb{C}$. Expanding the left side, $\alpha_0^{} \mathbf{d}(\tau_0^{}) \odot \widetilde{\mathbf{x}} + \alpha_1^{} \mathbf{d}(\tau_1^{}) \odot \widetilde{\mathbf{x}} = \text{diag}(\alpha_0^{} \mathbf{d}(\tau_0^{}) + \alpha_1^{} \mathbf{d}(\tau_1^{})) \widetilde{\mathbf{x}}$. If only one path is non-zero, e.g., $\alpha_1^{} = 0$, a solution is $\widetilde{\mathbf{x}}_0^{} = \operatorname{diag}^{-1}(\mathbf{d}(\tau_0^{}))\,\mathbf{D}(\overline{\tau}_0^{}, \overline{\tau}_1^{}) \bm{\beta}_0^{} = \mathbf{D}(\overline{\tau}_0^{}-\tau_0^{},\,\overline{\tau}_1^{}-\tau_0^{})\,\bm{\beta}_0^{}$, and similarly for $\alpha_0^{} = 0$, $\widetilde{\mathbf{x}}_1^{} = \operatorname{diag}^{-1}(\mathbf{d}(\tau_1^{}))\,\mathbf{D}(\overline{\tau}_0^{}, \overline{\tau}_1^{}) \bm{\beta}_1^{} = \mathbf{D}(\overline{\tau}_0^{}-\tau_1^{},\,\overline{\tau}_1^{}-\tau_1^{})\,\bm{\beta}_1^{}$.

To spoof both paths simultaneously, we require $\widetilde{\mathbf{x}}_0^{} = \widetilde{\mathbf{x}}_1^{}$, which, under linear independence of delay vectors, leads to the trivial solution (i.e., $\bm{\beta}_0^{}=\bm{\beta}_1^{}=\mathbf{0}^{}_L$)  unless we impose structure on the target delays
$\overline{\tau}_1^{} = \overline{\tau}_0^{} + (\tau_1^{} - \tau_0^{})$
or equivalently,
%\begin{align}
$\overline{\tau}_1^{} = \overline{\tau}_0^{} - (\tau_1^{} - \tau_0^{})$. 
%\end{align}
Then, the pilot simplifies to
%\begin{align}
$\widetilde{\mathbf{x}} = \lambda \mathbf{d}(\Delta^{}_{\tau_0^{}}) = \lambda \mathbf{d}(\Delta^{}_{\tau_1^{}})$. 
%\end{align}
This shifts both paths by the same amount, preserving their delay difference. 
\vspace{-7pt}
%Thus, \emph{with unknown gains, it is not possible to spoof the \ac{tdoa} estimator to a new delay difference using a simple shift}.

%\subsection{Derivation of CRB Expressions} \li{TODO: Hui}
%\hui{I would suggest using pseudotrue parameters for spoofing performance evaluation without considering SNRs. We can derive the MLB for selecting lambdas, but I do not feel we have sufficient time and space to derive (especially debug) 10 parameters (2 amplitude + 2 phase + 6 channel parameters).}

%\subsection*{Parameter Tables}

%\balance 
\bibliographystyle{IEEEtran}
\bibliography{sub/ref_abbrev}
\vspace{-10pt}
\begin{IEEEbiography}[{\includegraphics[width=1in,height=1.25in,clip,keepaspectratio]{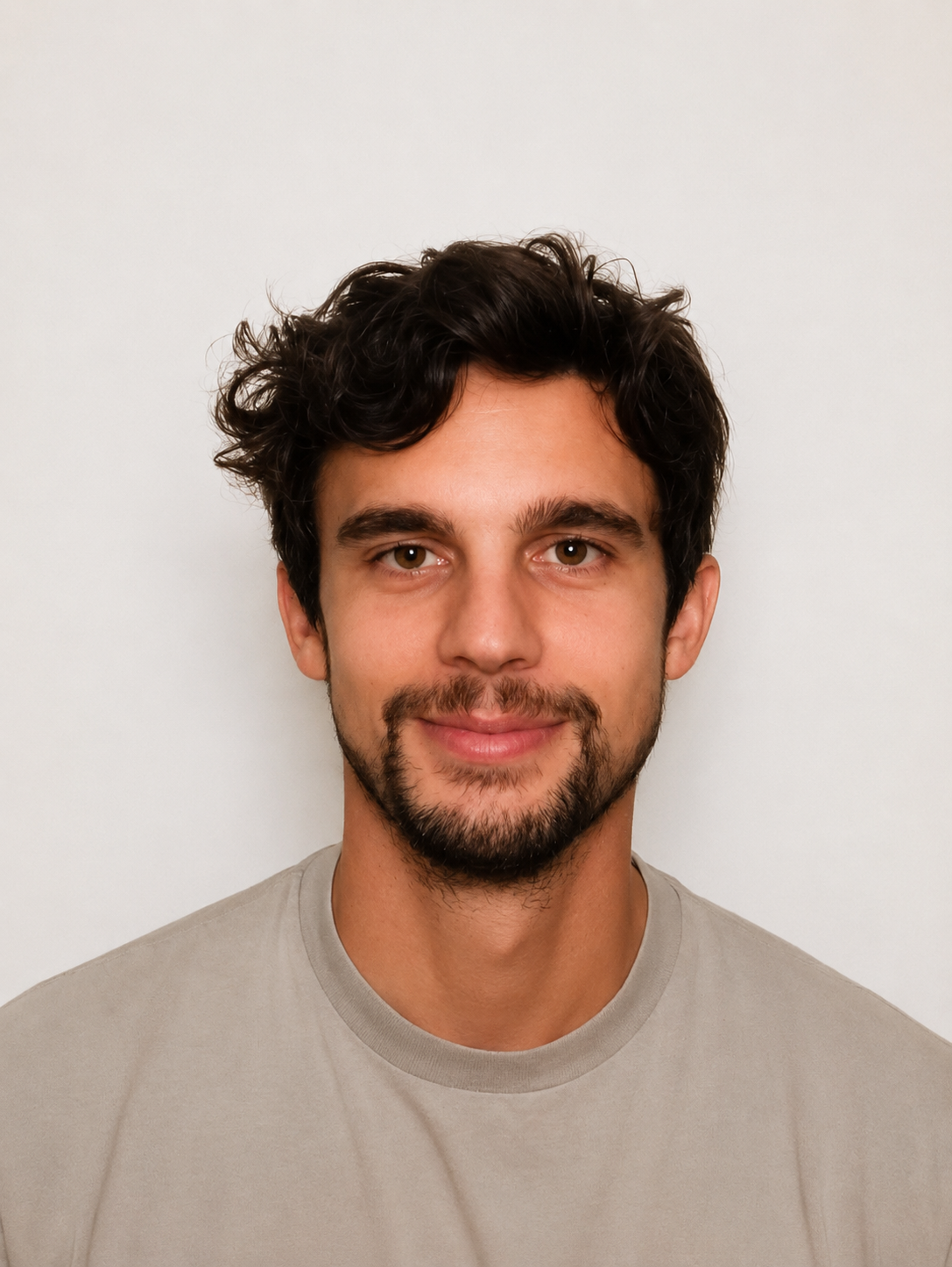}}]{Lorenzo Italiano} (Graduate Student Member, IEEE) received the B.Sc. degree in Computer Engineering from the University of Florence, Italy, in 2020, and the M.Sc. degree (cum laude) in Computer Science and Engineering from the Politecnico di Milano, Italy, in 2022, with a focus on cybersecurity and signal processing.
He started his Ph.D. at the Politecnico di Milano in November 2022 and was a Visiting Researcher at Chalmers University of Technology, Gothenburg, Sweden, in 2025. His research interests include 5G positioning, privacy, and V2X communications. He received the Distinguished Paper Award at ICSOC 2022 and the Francesco Carassa Award for Best Poster in Signal Processing at GTTI 2026.
\vspace{-15pt}
\end{IEEEbiography}

\begin{IEEEbiography}[{\includegraphics[width=1in,height=1.25in,clip,keepaspectratio]{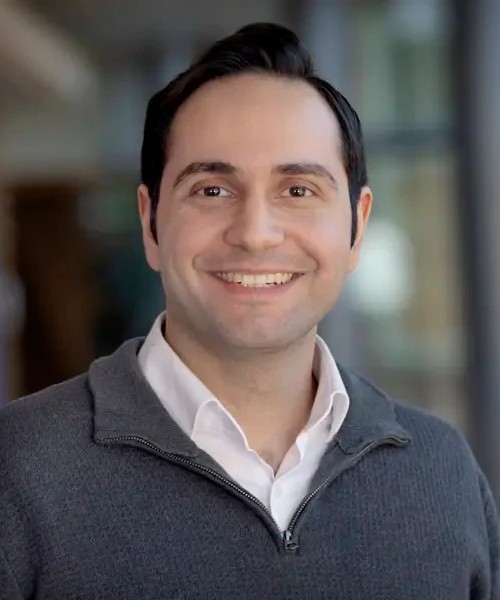}}]{Alireza Pourafzal} (Member, IEEE) received the
bachelor’s degree in electrical engineering and the master’s degree in wireless communications from the K.N. Toosi University of Technology, Tehran, Iran,
in 2016 and 2019, respectively, and the Ph.D. degree in information and communication technology from the Norwegian University of Science and Technology
(NTNU), in 2023. He is currently a Postdoctoral Researcher with Communication Group, Chalmers University of Technology, Gothenburg, Sweden. His research interests include array signal processing, channel estimation, DOA estimation, and machine learning for wireless and optical communications.
\vspace{-15pt}
\end{IEEEbiography}

\begin{IEEEbiography}[{\includegraphics[width=1in,height=1.25in,clip,keepaspectratio]{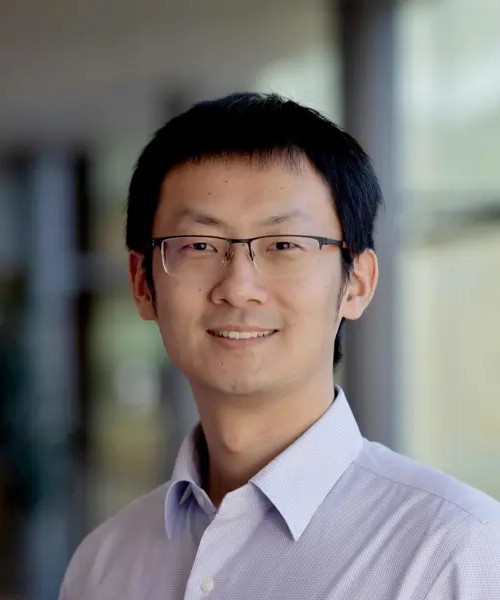}}]{Hui Chen} (Member, IEEE) received the Ph.D. degree in electrical and computer engineering from the King Abdullah University of Science and Technology, Saudi Arabia, in 2021. He is currently an Assistant Professor at Uppsala University, Sweden, and is also affiliated with Chalmers University of Technology, Sweden, as a Staff Scientist, and with University College London, U.K., as a Marie Skłodowska-Curie Fellow. His research focuses on 5G/6G localization and sensing.
\vspace{-15pt}
\end{IEEEbiography}

\begin{IEEEbiography}[{\includegraphics[width=1in,height=1.25in,clip,keepaspectratio]{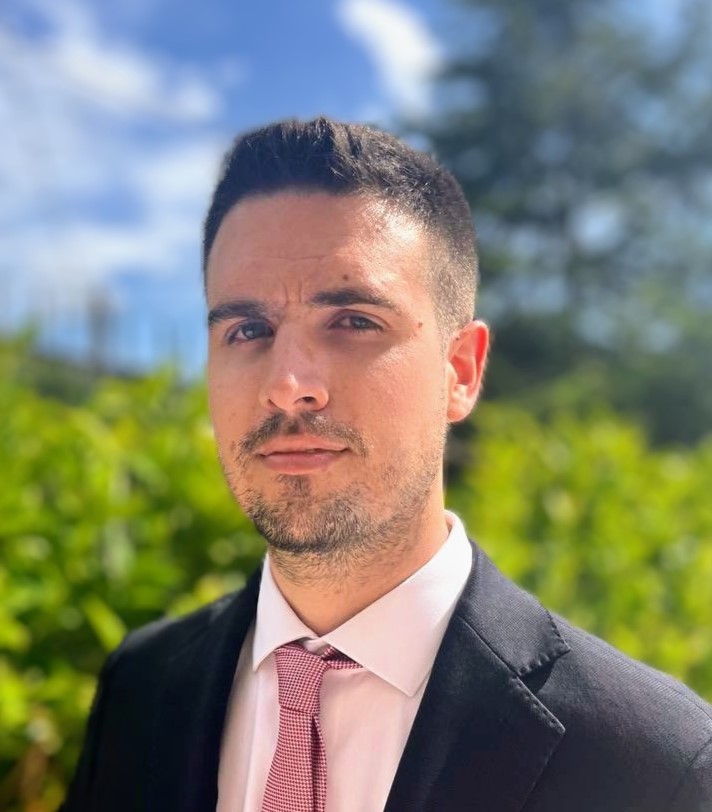}}]{Mattia Brambilla} (Senior Member, IEEE) received the B.Sc. and M.Sc. degrees in telecommunication engineering and the Ph.D. degree (cum laude) in information technology from the Politecnico di Milano, Milan, Italy, in 2015, 2017,
and 2021, respectively. In 2019, he was a Visiting Researcher with the
NATO Centre for Maritime Research and Experimentation, La Spezia, Italy. From 2021 he is a Research Fellow at Politecnico di Milano. His research interests include signal processing, statistical learning, and data fusion for cooperative localization and communications. 
Dr. Brambilla was a recipient of the 2024 Best Paper Award of Elsevier Vehicular Communications and the Best Student Paper Award at the 2018 IEEE Statistical Signal Processing Workshop. He is an Associate Editor of IEEE Open Journal of Signal Processing and ISIF Journal of Advances in Information Fusion.
\vspace{-20pt}
\end{IEEEbiography}

\begin{IEEEbiography}[{\includegraphics[width=1in,height=1.25in,clip,keepaspectratio]{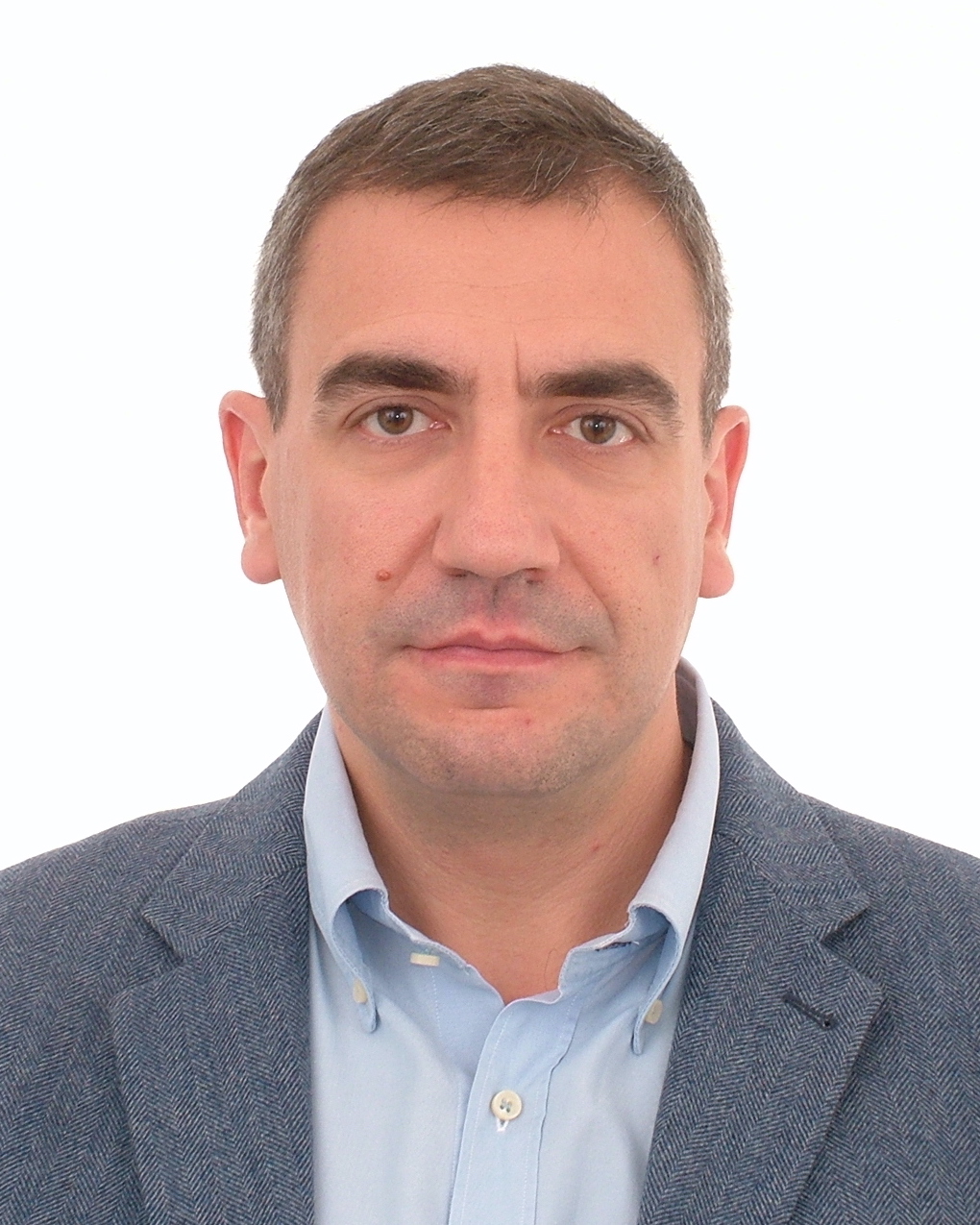}}]{Gonzalo Seco-Granados}(Fellow, IEEE) received the Ph.D. degree in telecommunications engineering from the Universitat Politècnica de Catalunya, Barcelona, Spain, in 2000, and the MBA degree from IESE Business School, Barcelona, in 2002. From 2002 to 2005, he was with the European Space Agency, Paris, France. He was a Fulbright Visiting Scholar with the University of California, Irvine, CA, USA, in 2015, 2019, and 2022, respectively. He is currently a Professor with the Department of Telecommunications, Universitat Autònoma de Barcelona, and Director of the Center for Space Studies and Research. He is also affiliated with the Institute of Space Studies of Catalonia, Barcelona, and ICREA Academia Fellow. His research interests include signal processing for GNSS, LEO-PNT, and 5G/6G integrated communications, localization, and sensing. He has been a member of the EURASIP Signal Processing for Multisensor Systems TC since 2022. He was the President of the Spanish Chapter of the IEEE AES Society (2018–2025). He was the recipient of the 2021 IEEE SP Society’s Best Paper Award.
\vspace{-15pt}
\end{IEEEbiography}

\begin{IEEEbiography}[{\includegraphics[width=1in,height=1.25in,clip,keepaspectratio]{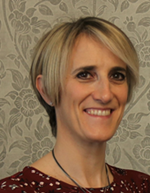}}]{Monica Nicoli} (Senior Member, IEEE) received the M.Sc. (Hons.) and her Ph.D. degree in communication engineering in 2002 from Politecnico di Milano, Italy, where she is an associate professor of telecommunications in the Department of Management, Economics, and Industrial Engineering. Her research interests include signal processing, machine learning, and wireless communications, with emphasis on smart mobility and Internet of Things applications. She was a visiting researcher with ENI Agip from 1998 to 1999 and with Uppsala University in 2001. She was the recipient of the Marisa Bellisario Award in 1999 and a co-recipient of the best paper awards of the Elsevier Vehicular Communications in 2025; the IEEE Symposium on Joint Communications and Sensing, in 2021; the IEEE Statistical Signal Processing Workshop, in 2018; and IET Intelligent Transport Systems, in 2014. She has been an associate editor of IEEE Transactions on Intelligent Transportation Systems (2021–2024).
\vspace{-15pt}
\end{IEEEbiography}

\begin{IEEEbiography}[{\includegraphics[width=1in,height=1.25in,clip,keepaspectratio]{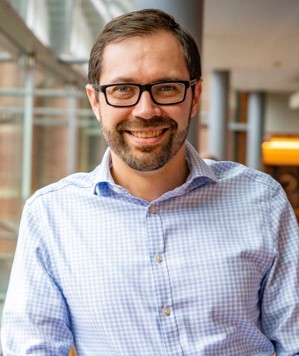}}]{Henk Wymeersch} (Fellow, IEEE) received the Ph.D.
degree in electrical engineering/applied sciences from Ghent University, Belgium, in 2005. From 2005 to 2009, he was a Postdoctoral Researcher with the
Laboratory for Information and Decision Systems, Massachusetts Institute of Technology, Cambridge, MA, USA. He is currently a Professor of communication systems with the Department of Electrical Engineering, Chalmers University of Technology,
Gothenburg, Sweden. His research interests include the convergence of communication and sensing in a 5G and Beyond 5G context. He was an Associate Editor for IEEE Communications Letters, IEEE Transactions on Wireless Communications, and IEEE Transactions on Communications. During 2019–2021, he was an IEEE Distinguished Lecturer with the Vehicular Technology Society. He is also a Senior Member of the IEEE Signal Processing Magazine Editorial Board.
\end{IEEEbiography}
\vfill

\end{document}